%% file: main.tex
\pdfoutput=1
\PassOptionsToPackage{dvipsnames}{xcolor}
\documentclass[acmsmall, screen]{acmart}

\input{util/0-main}

\makeatletter
\@ifundefined{newcounteralias}{}{%
  \renewcommand\thmt@autorefsetup{\@xa\def\csname\thmt@envname autorefname\@xa\endcsname\@xa{\thmt@thmname}}%
}
\makeatother

\begin{document}

\setcopyright{cc}
\setcctype{by}
\acmJournal{PACMPL}
\acmYear{2026} \acmVolume{10} \acmNumber{PLDI} \acmArticle{170}
\acmMonth{6} \acmDOI{10.1145/3808248}

\input{util/after-begin-commands}

\title{Fast Atomicity Monitoring}

\author[H. C. Tun\c{c}]{H\"{u}nkar Can Tun\c{c}}
\authornote{Both authors contributed equally to this work.}
\orcid{0000-0001-9125-8506} 
\affiliation{
	\institution{Uber Technologies}
    \city{Amsterdam}
	\country{Netherlands}
}
\email{tunc@uber.com}

\author[Y. Dong]{Yifan Dong}
\authornotemark[1]
\orcid{0009-0002-0696-839X}
\affiliation{
	\institution{Aarhus University}
    \city{Aarhus}
	\country{Denmark}
}
\email{yvan.dong@cs.au.dk}

\author[A. Pavlogiannis]{Andreas Pavlogiannis}
\orcid{0000-0002-8943-0722}
\affiliation{
	\institution{Aarhus University}
    \city{Aarhus}
	\country{Denmark}
}
\email{pavlogiannis@cs.au.dk}

%\author{...} % removed for anonymity

\input{sections/abstract}

\maketitle % should come after the abstract
%\pagestyle{plain} % should come right after \maketitle

\input{sections/intro-main}
\input{sections/prelim-main}
\input{sections/serializability-main}

\input{sections/increasing_cycle}
\input{sections/online}
\input{sections/experiments-main}

\input{sections/related_work-main}

\input{sections/conclusion-main}
\input{sections/data-availability-statement}

\input{sections/acknowledgments}

% \clearpage
\bibliographystyle{ACM-Reference-Format}
\bibliography{references}

% \clearpage

 \appendix
 \input{sections/app_serializability}
 \input{sections/app_increasing_cycle}
 \input{sections/app_experiments}
 \input{sections/app_scalability}

 \input{sections/app_doublechecker}
 \input{sections/app_aerodrome}
%\clearpage

%\input{draft}87879

\end{document}

%% file: util/0-main.tex
\input{util/packages}
\input{util/style}

\input{util/names}
\input{util/generic-notation}

\input{util/ds-notation}


\input{util/figures-format}

\input{util/algorithm-notation}
\input{util/velodrome-notation}

\input{util/increasing-cycle-notation}
\input{util/view-notation}

%% file: util/packages.tex
\usepackage{hyperref}
\usepackage{multirow, makecell}
\usepackage[utf8]{inputenc} 
\usepackage{amsmath, amsthm}
\usepackage{thmtools}
\usepackage[english]{babel}
\usepackage{pdfpages}
\usepackage{graphicx}
\usepackage{subcaption}
\usepackage{caption}
\usepackage{longtable}
\usepackage{booktabs}
\usepackage{tabu}
\usepackage[referable]{threeparttablex}
\usepackage{varwidth}
\usepackage{multirow}
\usepackage{wrapfig}
\usepackage[inline]{enumitem}
\usepackage{microtype}
\usepackage{xifthen}
\usepackage{versions}
\usepackage{dsfont}
\usepackage[dvipsnames]{xcolor}
\usepackage{pifont}
\usepackage{adjustbox}

\usepackage{enumitem}
\usepackage{rotating}

\usepackage[noend,ruled,linesnumbered]{algorithm2e}

\usepackage{parskip}

\usepackage{thm-restate}
\usepackage[capitalise]{cleveref}

\crefname{figure}{Figure}{Figure}
\crefformat{footnote}{#2\footnotemark[#1]#3}
\usepackage{comment}

\usepackage{tikz}
\usetikzlibrary{arrows,automata,shapes,decorations,decorations.markings,calc, matrix,decorations.pathmorphing, patterns,backgrounds,shapes.misc,arrows.meta,positioning}
\usetikzlibrary{trees}
\usetikzlibrary{calc}
\usetikzlibrary{tikzmark}

\usepackage{bm}
\usepackage{pgfplotstable}
\usepackage{pgfplots}
\usepackage{pbox}

\usepackage{calc}
\usepackage{fp}
\usepackage{multirow}
\usepackage{pbox}

\usepackage[mathscr]{euscript}
\usepackage{mathtools}

\usepackage{footnote}
\usepackage{bbm}
\usepackage{multicol}
\usepackage{graphicx} % omit 'demo' option in real doc.
\usepackage{soul}
\usepackage{listings, xcolor}
\usepackage{versions}
\usepackage[frozencache,cachedir=_minted]{minted}

\usepgfplotslibrary{groupplots}

%% file: util/style.tex
\newcommand{\Paragraph}[1]{\smallskip\noindent{\bf #1}}
\newcommand{\SubParagraph}[1]{\smallskip\noindent{\em #1}}

\DeclareFontFamily{U}{matha}{\hyphenchar\font45}
\DeclareFontShape{U}{matha}{m}{n}{
      <5> <6> <7> <8> <9> <10> gen * matha
      <10.95> matha10 <12> <14.4> <17.28> <20.74> <24.88> matha12
      }{}
\DeclareSymbolFont{matha}{U}{matha}{m}{n}

\DeclareMathSymbol{\Lt}{3}{matha}{"CE}
\DeclareMathSymbol{\Gt}{3}{matha}{"CF}

%% file: util/names.tex
\newcommand{\velodrome}{\mathtt{Velodrome}}

\newcommand{\aerodrome}{\mathtt{Aerodrome}}
\newcommand{\regiontrack}{\mathtt{RegionTrack}}
\newcommand{\ouralgo}{\mathtt{AtomSanitizer}}
\newcommand{\orig}{\mathtt{Origin}}
\newcommand{\ds}{\mathtt{DS}}
\newcommand{\doublechecker}{\mathtt{DoubleChecker}}
\newcommand{\octet}{\mathtt{Octet}}
\newcommand{\FarzanMadhusudan}{\mathtt{FM}}
\newcommand{\TSAN}{\mathtt{TSAN}}

\newcommand{\RAPID}{\mathtt{RAPID}}

\newcommand{\ltxn}{\mathtt{latestTxn}}

\newcommand{\sd}{\text{-}}
\newcommand{\sclabel}[3]{#1,#2,#3}

\newcommand{\ovhtsan}{\texttt{OvhTsan}}
\newcommand{\ovhorig}{\texttt{OvhOrig}}

%% file: util/generic-notation.tex
\newtheorem{remark}{Remark}

\newcommand{\tr}{\sigma}
\newcommand{\threads}[1]{\mathsf{Thr}_{#1}}
\newcommand{\events}[1]{\mathsf{Events}_{#1}}

\newcommand{\set}[1]{\{#1\}}

\newcommand{\idof}[1]{\operatorname{id}(#1)}
\newcommand{\threadof}[1]{\operatorname{tid}(#1)}

\newcommand{\acq}{\mathtt{acq}}
\newcommand{\rel}{\mathtt{rel}}
\newcommand{\rd}{\mathtt{r}}
\newcommand{\wt}{\mathtt{w}}

\newcommand{\tid}{\mathtt{tid}}
\newcommand{\op}{\mathtt{op}}
\newcommand{\var}{\mathtt{var}}

\newcommand{\PO}[2]{<_{#1}^{\mathtt{#2}}}
\newcommand{\POR}[2]{\leq_{#1}^{\mathtt{#2}}}
\newcommand{\thb}[1]{\PO{#1}{THB}}
\newcommand{\thbR}[1]{\POR{#1}{THB}}
\newcommand{\chb}[1]{\PO{#1}{CHB}}

\newcommand{\traceOrd}[1]{ \PO{#1}{TR}}

\newcommand{\threadOrd}[1]{\PO{#1}{TO}}
\newcommand{\threadOrdR}[1]{\POR{#1}{TO}}

\newcommand{\cf}[2]{#1 \asymp #2}
\newcommand{\variables}[1]{\mathtt{Variables}_{#1}}
\newcommand{\locks}[1]{\mathtt{Locks}_{#1}}
\newcommand{\reads}[1]{\mathtt{Reads}_{#1}}
\newcommand{\writes}[1]{\mathtt{Writes}_{#1}}

\newcommand{\txof}{\mathtt{txn}}
\newcommand{\txid}[1]{\mathtt{id}(#1)}
\newcommand{\txbegin}{\triangleright}
\newcommand{\txend}{\triangleleft}
\newcommand{\unary}[1]{\mathtt{unary}(#1)}

\newcommand{\numEvents}{n}
\newcommand{\numThreads}{k}
\newcommand{\numLocations}{v}
\newcommand{\numLocks}{\ell}
\newcommand{\numTransactions}{h}

\newlist{compactitem}{itemize}{3} % 3 is max-depth
\setlist[compactitem]{label=\textbullet,nosep,leftmargin=*}
\crefname{compactitemi}{Item}{Items}

\newlist{compactenum}{enumerate}{3} % 3 is max-depth
\setlist[compactenum]{label=\arabic*., nosep,leftmargin=*}
\crefname{compactenumi}{Item}{Items}

\newlist{compactwideenum}{enumerate}{3} % 3 is max-depth
\setlist[compactwideenum]{label=\textit{\arabic*}., nosep,leftmargin=*,wide}
\crefname{compactwideenumi}{Item}{Items}

\newcommand{\LatestRemoteTransaction}[2]{\operatorname{LRT}_{#1}^{#2}}
\newcommand{\EarliestReceiverTransaction}[2]{\operatorname{ELT}_{#1}^{#2}}
\newcommand{\XClock}{\operatorname{\mathsf{From}}}
\newcommand{\YClock}{\operatorname{\mathsf{To}}}

\newcommand{\LatestRemoteBegin}[2]{\operatorname{LRB}_{#1}^{#2}}
\newcommand{\EarliestLocalEvent}[2]{\operatorname{ELE}_{#1}^{#2}}

%% file: util/ds-notation.tex
\newcommand{\graphnode}[2]{{\langle #1, #2 \rangle}}

\newcommand{\upd}{\mathtt{update}}
\newcommand{\mn}{\mathtt{min}}

\newcommand{\argleq}{\mathtt{argleq}}

\newcommand{\width}{\mathtt{width}}

\newcommand{\readlock}{\mathtt{Rlock}}
\newcommand{\readunlock}{\mathtt{RUnlock}}
\newcommand{\writelock}{\mathtt{WLock}}

\newcommand{\trywritelock}{\mathtt{tryWLock}}
\newcommand{\writeunlock}{\mathtt{WUnlock}}
\newcommand{\elock}{\mathbb{L}}
\newcommand{\astore}{\mathtt{atomic\_store}}
\newcommand{\aload}{\mathtt{atomic\_load}}

\newcommand{\oeflag}{\mathtt{hasOutEdge}}

\newcommand{\pred}{\mathtt{predecessor}}
\newcommand{\suc}{\mathtt{successor}}

\newcommand{\addsuc}{\mathtt{addSuccessor}}

\newcommand{\insedge}{\mathtt{insertEdge}}
\newcommand{\insedgeback}{\mathtt{insertEdgeBackwards}}
\newcommand{\deledge}{\mathtt{deleteEdge}}

\newcommand{\reachable}{\mathtt{reachable}}

\newcommand{\bg}{\mathtt{begin}}
\newcommand{\en}{\mathtt{end}}

\newcommand{\node}{\mathtt{nd}}

\newcommand{\height}{\mathtt{height}}

%% file: util/figures-format.tex
\newcommand{\bbg}[1]{\color{blue}$#1\colon\rhd$}
\newcommand{\ben}[1]{\color{blue}$#1\colon\lhd$}
\newcommand{\pbg}{\color{blue}$\rhd$}
\newcommand{\pen}{\color{blue}$\lhd$}
\newcommand{\mopartial}
	
\colorlet{colorPO}{darkgray!80!black}
\colorlet{colorRF}{blue}
\colorlet{colorCO}{red!80!black}
\colorlet{colorMO}{red!80!black}
\colorlet{colorFR}{purple}
\colorlet{colorECO}{orange}
\colorlet{colorCOM}{magenta}
\colorlet{colorSW}{teal}
\colorlet{colorHB}{green!40!black}
\colorlet{colorPPO}{magenta}
\colorlet{colorRSEQ}{green!40!black}
\colorlet{colorSC}{violet}
\colorlet{colorHBMO}{brown}
\colorlet{colorPSC}{violet}
\colorlet{colorREL}{olive}
\colorlet{colorWB}{olive}
\colorlet{colorRMW}{brown}
\colorlet{colorVIO}{red}

\definecolor{bred}{RGB}{213,94,0}
\definecolor{bgreen}{RGB}{0,158,115}
\definecolor{bblue}{RGB}{0,114,178}

\newcommand\numberthis{\addtocounter{equation}{1}\tag{\theequation}}

\tikzset{
	circled/.style={circle, draw, fill=white, inner sep=1},
	every path/.style={>=stealth},
	po/.style={->,color=colorPO,shorten >=-0.5mm,shorten <=-0.5mm},
	ppo/.style={->,color=colorPPO,shorten >=-0.5mm,shorten <=-0.5mm},
	sw/.style={->,color=colorSW,shorten >=-0.5mm,shorten <=-0.5mm},
	rf/.style={->,color=colorRF,dotted, very thick,shorten >=-0.5mm,shorten <=-0.5mm},
	rfsolid/.style={->,color=bgreen, very thick,shorten >=-0.5mm,shorten <=-0.5mm},
	mrf/.style={->,color=colorRF,dashed,shorten >=-0.5mm,shorten <=-0.5mm},   
	urf/.style={->,color=colorRF,shorten >=-0.5mm,shorten <=-0.5mm},
	fr/.style={->,color=colorFR,dashed,shorten >=-0.5mm,shorten <=-0.5mm},
	hb/.style={->,color=colorPO,thick,shorten >=-0.5mm,shorten <=-0.5mm},
	co/.style={->,color=colorCO,dotted,very thick,shorten >=-0.5mm,shorten <=-0.5mm},
	mo/.style={->,color=bred,dotted,very thick,shorten >=-0.5mm,shorten <=-0.5mm},
	mosolid/.style={->,color=bred,very thick,shorten >=-0.5mm,shorten <=-0.5mm},
	rmw/.style={->,color=colorRMW,thick,shorten >=-0.5mm,shorten <=-0.5mm},
	rseq/.style={->,color=colorRSEQ,thick,dotted,shorten >=-0.5mm,shorten <=-0.5mm},
	com/.style={->,color=colorCOM,thick,shorten >=-0.5mm,shorten <=-0.5mm},
	thb/.style={->,color=colorHB,thick,shorten >=-0.5mm,shorten <=-0.5mm},
	vio/.style={->,color=colorVIO,thick,shorten >=-0.5mm,shorten <=-0.5mm},
}

\pgfplotsset{compat=1.11,
    /pgfplots/ybar legend/.style={
    /pgfplots/legend image code/.code={%
       \draw[##1,/tikz/.cd,yshift=-0.25em]
        (0cm,0cm) rectangle (3pt,0.8em);},
   },
}

%% file: util/algorithm-notation.tex
\SetKwProg{MyProcedure}{procedure}{}{}
\SetKwProg{MyFunction}{function}{}{}
\SetKwFunction{ProcNewEvent}{handleEvent}
\SetKwFunction{FunctionComputeNewOrderings}{computeOrderings}
\SetKwFunction{FunctionMaintainPO}{maintainPartialOrder}

\SetKwFunction{FunctionReachable}{{$\reachable$}}
\SetKwFunction{FunctionAddSucc}{{$\mathtt{addSuccessor}$}}

\SetKwFunction{FunctionQuery}{{$\reachable$}}
\SetKwFunction{FunctionInsertEdge}{{$\insedge$}}
\SetKwFunction{FunctionInsertEdgeBackwards}{{$\insedgeback$}}
\SetKwFunction{FunctionDeleteEdge}{{$\deledge$}}
\SetKwFunction{FunctionSucc}{{$\suc$}}
\SetKwFunction{FunctionPred}{{$\pred$}}

\SetKwFunction{FunctionMin}{{$\mn$}}
\SetKwFunction{FunctionArgLeq}{{$\argleq$}}
\SetKwFunction{FunctionCheckAndUpdate}{{$\mathtt{checkAndUpdate}$}}
\SetKwFunction{FunctionUpdate}{{$\upd$}}
\SetKwFunction{FunctionUpdateHelper}{{$\mathtt{updateHelper}$}}
\SetKwFunction{FunctionCreateNode}{{$\mathtt{createNode}$}}
\SetKwFunction{FunctionCreateIntermediateNodes}{{$\mathtt{createIntermediateNodes}$}}
\SetKwFunction{FunctionCreateLowestCommonAncestor}{{$\mathtt{createLowestCommonAncestor}$}}

\SetKwFunction{FunctionReachable}{{$\mathtt{reachable}$}}
\SetKwFunction{FunctionUnordered}{{$\mathtt{unordered}$}}
\SetKwFunction{FunctionInitialize}{{$\mathtt{initialize}$}}

\SetKwFunction{FunctionAcquire}{{$\mathtt{acquire}$}}
\SetKwFunction{FunctionRelease}{{$\mathtt{release}$}}
\SetKwFunction{FunctionRead}{{$\mathtt{read}$}}
\SetKwFunction{FunctionWrite}{{$\mathtt{write}$}}
\SetKwFunction{FunctionBegin}{{$\mathtt{begin}$}}
\SetKwFunction{FunctionEnd}{{$\mathtt{end}$}}
\SetKwFunction{FunctionTxn}{{$\mathtt{txn}$}}
\SetKwFunction{FunctionId}{{$\mathtt{id}$}}
\SetKwFunction{FunctionSetTxn}{{$\mathtt{setTxn}$}}

\SetKwFunction{FunctionAddPath}{{$\mathtt{addPath}$}}
\SetKwFunction{FunctionCheckViolation}{{$\mathtt{checkViolation}$}}

\SetKwFunction{FunctionReadFromClosure}{{$\mathtt{readFromClosure}$}}
\SetKwFunction{FunctionLockClosure}{{$\mathtt{lockClosure}$}}
\SetKwFunction{FunctionTxnClosure}{{$\mathtt{txnClosure}$}}
\SetKwFunction{FunctionSaturate}{{$\mathtt{saturate}$}}
\SetKwFunction{FunctionCheckSER}{{$\mathtt{checkSerializability}$}}
\SetKwFunction{FunctionNextEvent}{{$\mathtt{nextEvent}$}}
\SetKwFunction{FunctionSpeculate}{{$\mathtt{speculate}$}}

\SetKwFunction{FunctionAddEvent}{addEvent}
\SetKw{Continue}{continue}
\SetKw{Break}{break}

\newcommand{\reportError}[1]{\textnormal{declare}~\textbf{#1}}

\newcommand{\lIfElse}[3]{\lIf{#1}{#2 \textbf{else}~#3}}

\newcommand{\greytcp}[1]{\textcolor{gray}{\tcp{#1}}}

%% file: util/velodrome-notation.tex
\newcommand{\thread}{\mathtt{t}}

\newcommand{\lastRelease}[1]{\mathtt{lastRelease}_{#1}}
\newcommand{\lastWrite}[1]{\mathtt{lastWrite}_{#1}}
\newcommand{\lastReads}[1]{\mathtt{lastReads}_{#1}}
\newcommand{\lrd}[1]{\mathtt{lastReadTids}_{#1}}

%% file: util/view-notation.tex
\newcommand{\pair}[2]{\langle #1, #2 \rangle}

%% file: util/after-begin-commands.tex
\newcommand*\circled[1]{\tikz[baseline=(char.base)]{
		\node[shape=circle,draw,fill=white,inner sep=1pt] (char) {#1};}}

%% file: sections/abstract.tex
\begin{abstract}
Atomicity is a fundamental abstraction in concurrency, specifying that program behavior can be understood by considering specific code blocks executing atomically.
However, atomicity invariants are tricky to maintain while also optimizing for code efficiency, and atomicity violations are a common root cause of many concurrency bugs.
To address this problem, several dynamic techniques have been developed for testing whether a program execution adheres to an atomicity specification, most often instantiated as \emph{conflict-serializability}.
The efficiency of the analysis has been targeted in various papers,
with the state-of-the-art algorithms $\regiontrack$ and $\aerodrome$ achieving a time complexity $O(\numEvents\numThreads^3)$ and $O(\numEvents\numThreads(\numThreads + \numLocations + \numLocks))$, respectively,
for a trace $\tr$ of $\numEvents$ events, $\numThreads$ threads, $\numLocations$ locations, and $\numLocks$ locks.

In this paper we introduce $\ouralgo$, a new algorithm for testing conflict-serializability, with time complexity $O(\numEvents\numThreads^2)$.
$\ouralgo$ operates in an efficient streaming style, is theoretically faster than all existing algorithms, and also has a smaller memory footprint.
Moreover, it is the first algorithm designed to use little locking when deployed in a concurrent monitoring setting.
Experiments on standard benchmarks indicate that $\ouralgo$ is always faster in practice than all existing conflict-serializability testers.
Finally, we also implement $\ouralgo$ inside the $\TSAN$ framework, for monitoring atomicity in real time.
Our experiments reveal that $\ouralgo$ incurs only a marginal time and space overhead over the data-race detection engine of $\TSAN$, and thus is the first algorithm for conflict-serializability demonstrated to be suitable for a runtime monitoring setting.
\end{abstract}

\keywords{concurrency, transactions, conflict-serializability, dynamic runtime analysis}


\ccsdesc[500]{Software and its engineering~Software verification and validation}
\ccsdesc[500]{Software and its engineering~Software testing and debugging}

%% file: sections/intro-main.tex
%!TEX root = ../main.tex
\section{Introduction}

\input{sections/intro-beginning}
\input{sections/intro-contributions}

%% file: sections/intro-beginning.tex
%!TEX root = ../main.tex

Writing correct concurrent programs is a notoriously challenging task, as interprocess communication is inherently non-deterministic, encompassing unanticipated scheduling patterns that often escape the programmers' mental model.
For the same reason, concurrency bugs are difficult to reproduce during debugging, and are commonly referred to as \emph{heisenbugs}~\cite{Musuvathi-2008}.

One fundamental building block in the mental model of programmers is \emph{atomicity}, where the whole program consists of individual blocks that can be considered to execute atomically.
Although this is not actually the case due to the scheduler, atomicity allows one to reason about programs in terms of atomic blocks, meaning that their interleaving cannot produce new behaviors.
However, violations of atomicity are common in practice and are the source of many real-world concurrency bugs~\cite{Park-2009, Lu-2006, Lu-2008}.
\input{figures/inc_code_example} 
For example, \cref{fig:cs-real-example} depicts an atomicity violation in MySQL~\cite{mysql-04}, taken from~\cite{Lu-2008}.
Two threads access concurrently the shared variable \texttt{thd->proc\_info}.
The developers assumed that \texttt{S1} and \texttt{S3}  always execute atomically in thread 1 due to synchronization by code not shown in this snippet, thus \texttt{thd->proc\_info} in \texttt{S3} will never be \texttt{NULL}. 
However, \texttt{S2} of thread 2 can, in fact, be interleaved with \texttt{S1} and \texttt{S3}, breaking this atomicity invariant and leading to a bug that has resurfaced several times~\cite{mysql-08, mysql-08-2}.
We remark that atomicity violations are different from data races; they occur when concurrent accesses from different threads interleave in an unexpected way, even, e.g., if such accesses are synchronizing, and can thus occur even on race-free programs.

The programming languages community has responded with techniques for detecting atomicity violations automatically.
One standard approach is via dynamic analyses that report warnings based on violations of \emph{conflict-serializability}.
The analysis observes executions $\tr$ of the concurrent program which are annotated with intended atomic blocks, called transactions.
Annotations are external to the analysis ---they might come explicitly from the programmer~\cite{Flanagan2003}, or be inferred automatically~\cite{Lu-2006}.
Although transactions may interleave in $\tr$, the analysis must decide whether there is an equivalent execution $\tr'$, free of transaction interleavings, obtainable from $\tr$ by successive swaps of non-conflicting events.
If not, $\tr$ witnesses an atomicity violation.
The size of $\tr$ is usually huge, which puts pressure on the analysis to be as fast as possible.
In the following, consider that $\tr$ contains $\numEvents$ events, $\numThreads$ threads, $\numLocations$ locations, and $\numLocks$ locks.

In high level, conflict-serializability violations can be detected by creating a conflict graph of all transactions, whereby two transactions are ordered $T_1\to T_2$ iff each $T_i$ contains an event $e_i$ such that $e_1$ and $e_2$ conflict (both access the same shared resource and at least one modifies it) and $e_1$ executes before $e_2$ in $\tr$.
Then, a violation corresponds to a cycle $T_{1}\to\dots\to T_{m}$, representing the fact that these $m$ transactions cannot execute serially without breaking the conflict order of the events in $\tr$ (see, e.g., the cycle $T_1\to T_2\to T_1$ in \cref{fig:cs-real-example}).
Constructing the conflict graph and checking which transactions appear on a cycle can be easily performed in $O(\numEvents)$ time, but this approach has two notable disadvantages.
First, the space usage is $\Theta(\numEvents)$, which can be prohibitive in practice.
Second, practical applications pose an online setting, whereby $\tr$ is processed in a streaming style, and a violation is reported as early as it occurs.
Consequently, there have been several efforts in designing atomicity-checking algorithms that operate in this online fashion and are as efficient as possible (see \cref{tab:intro-conflict-complexity}).

One of the first techniques for dynamic atomicity testing was developed in~\cite{Alur-2000}.
That algorithm, however, employed aggressive pruning of its data structures in order to keep its memory footprint small, which lead to incorrect results.
This was reported in~\cite{Farzan-2008}, which proposed an automata-based algorithm (which we refer to as $\FarzanMadhusudan$) with running time $O(\numEvents\numThreads(\numLocations+\numLocks))$.
At the same time, the graph-based algorithm $\velodrome$ was developed~\cite{Flanagan-2008}, and has become a standard reference in atomicity testing.
$\velodrome$ maintains explicitly a conflict graph that can grow as large as $\Theta(\numEvents)$, with cycles representing atomicity violations.
It has $O(\numEvents^2)$ time and $O(\numEvents)$ space complexity, which are prohibitive in theory, but is made performant in practice via graph-pruning heuristics.
$\doublechecker$~\cite{Biswas-2014} further improved of $\velodrome$ by observing that precise cycle detection is required infrequently.
$\aerodrome$ was recently proposed as the first algorithm utilizing vector clocks as its main data structure \cite{Mathur-2020}, which are efficient and common in dynamic concurrency analyses.
$\aerodrome$ has time and space complexity  $O(\numEvents\numThreads(\numThreads + \numLocations + \numLocks))$ and $O(\numThreads(\numThreads + \numLocations + \numLocks))$, respectively.
The latest development in this sequence is $\regiontrack$, which also utilizes vector clocks, towards time and space complexity $O(\numEvents\numThreads^3)$ and $O(\numThreads(\numThreads\numLocations+\numLocks))$, respectively~\cite{Ma-2021}.
The running time of $\regiontrack$ is incomparable to $\aerodrome$ in general, but it becomes faster for reasonable input parameters (i.e., typically $\numThreads<<\numLocations$).
However, $\regiontrack$ has increased space usage by a term $\numLocations\numThreads^2$.
Finally, when deployed in a runtime monitoring setting, where various threads access the components of the algorithm concurrently, all existing testers require each thread to acquire a global lock.
This incurs congestion and leads to noticeable monitoring overhead, beyond what is considered acceptable in an industrial scale (e.g., $\TSAN$'s reported 5$\times$-15$\times$ time and 5$\times$-10$\times$ memory overhead~\cite{Serebryany-2009}).

\input{material/tables/intro-complexity-table}

The long study of atomicity checking, as witnessed by the above sequence of improvements, highlights the need for methods that are as performant as possible.
\emph{
Is there an atomicity checker with even better complexity? 
Does it also lead to practical improvements?
Is atomicity monitorable at overheads acceptable by industrial practice?
Is the complexity increase of online monitoring necessary, compared to the $O(\numEvents)$ cost for offline detection?
}
This paper answers these questions in positive, by making the following contributions.

%% file: figures/inc_code_example.tex
%!TEX root = ../main.tex
%
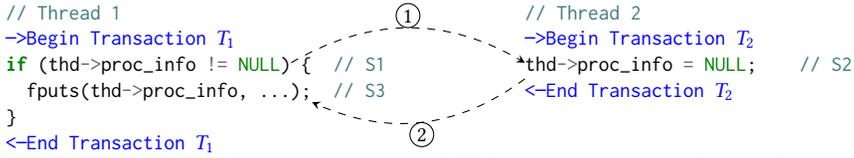
\begin{figure}[tbp]
% \begin{wrapfigure}[12]{r}{0.45\textwidth}
\renewcommand\theFancyVerbLine{\small\arabic{FancyVerbLine}}
\hfill
\begin{minipage}[t]{0.4\textwidth}
\begin{minted}[
% frame=lines,
framesep=2mm,
tabsize=2,
baselinestretch=1.0,
% bgcolor=LightGray!50,
numbersep=2pt,
fontsize=\footnotesize,
% linenos,
escapeinside=??,
]{C}
// Thread 1
?\color{blue}--->Begin Transaction $T_1$?
if (thd->proc_info != NULL)?\tikzmark{s1}? {  // S1
	fputs(thd->proc_info, ...);?\tikzmark{s3}?  // S3
}
?\color{blue}<---End Transaction $T_1$?
\end{minted}
\end{minipage}%
\hfill%
\begin{minipage}[t]{0.4\textwidth}
\begin{minted}[
% frame=lines,
framesep=2mm,
tabsize=2,
baselinestretch=1.0,
% bgcolor=LightGray!50,
numbersep=2pt,
fontsize=\footnotesize,
% linenos,
escapeinside=??,
]{C}
// Thread 2
?\color{blue}--->Begin Transaction $T_2$?
?\tikzmark{s2}?thd->proc_info = NULL;		// S2
?\color{blue}<---End Transaction $T_2$?
\end{minted}
\end{minipage}
\begin{tikzpicture}[remember picture]
\draw[overlay, ->, dashed] ([yshift=3pt]pic cs:s1) to[bend left=30] node[midway,above,font=\footnotesize,circled,solid] {$1$} ([yshift=3pt]pic cs:s2);
\draw[overlay, ->, dashed] ([yshift=-3pt]pic cs:s2) to[bend left=30] node[midway,below,font=\footnotesize,circled,solid] {$2$} ([yshift=-2pt]pic cs:s3);
\end{tikzpicture}
\caption{
An atomicity violation bug in MySQL 4.1.2 release, taken from~\cite{Lu-2008}.
%An atomicity violation in MySQL 4.1.2.
}
\label{fig:cs-real-example}
% \end{wrapfigure}
\end{figure}

%% file: material/tables/intro-complexity-table.tex
%!TEX root = ../../main.tex

\begin{table}[!t]
\caption{The complexity of atomicity algorithms on executions of 
$\numEvents$ events,
$\numThreads$ threads,
$\numLocations$ locations, and
$\numLocks$ locks.
Some works give more fine-grained complexity bounds, e.g., involving  the number of transactions in the trace.
}
\label{tab:intro-conflict-complexity}
\centering
\begin{tabular}{|c|c|c|}
\hline
\textbf{Algorithm} & \textbf{Time} & \textbf{Space} \\ \hline
$\FarzanMadhusudan$~\cite{Farzan-2008} & $O(\numEvents\numThreads(\numLocations+\numLocks))$ & $O(\numThreads(\numThreads + \numLocations + \numLocks))$ \\
$\velodrome$~\cite{Flanagan-2008} & $O(\numEvents^2)$ & $O(\numEvents)$ \\ 
$\doublechecker$~\cite{Biswas-2014} & $O(\numEvents^2)$ & $O(\numEvents)$ \\ 
$\aerodrome$~\cite{Mathur-2020} & $O(\numEvents\numThreads(\numThreads + \numLocations + \numLocks))$ & $O(\numThreads(\numThreads + \numLocations + \numLocks))$ \\ 
$\regiontrack$~\cite{Ma-2021} & $O(\numEvents\numThreads^3)$ & $O(\numThreads(\numThreads\numLocations+\numLocks))$ \\ \hline
$\ouralgo$ [This Paper] & $O(\numEvents\numThreads^2)$ & $O(\numThreads(\numThreads+\numLocations)+ \numLocks)$ \\ \hline
\end{tabular}
\end{table}

%% file: sections/intro-contributions.tex
%\subsection{Contributions}\label{subsec:contributions}
%The main contributions of this paper are as follows.

\Paragraph{1.~AtomSanitizer for conflict-serializability.}
We develop $\ouralgo$, a new algorithm for testing conflict-serializability.
For a trace $\tr$ of $\numEvents$ events, $\numThreads$ threads, $\numLocations$ locations, and $\numLocks$ locks, $\ouralgo$ uses $O(\numEvents\numThreads^2)$ time and $O(\numThreads(\numThreads+\numLocations)+\numLocks)$ space.
This is better than all existing algorithms in both time and space (see \cref{tab:intro-conflict-complexity}), for reasonable input parameters.
In particular, $\ouralgo$ is faster when $\numThreads<<\numLocations$,
while its space complexity is smaller by a term $\numThreads\numLocks$.
%It is worth noting that $\ouralgo$ is the first algorithm improving running time of the problem in over $15$ years (since $\FarzanMadhusudan$) without blowing up the memory footprint.

We further remark that the space complexity of $\ouralgo$ is also bounded by $O(\numEvents + \numThreads^2)$, but the aforementioned space bound is better as long as $\numThreads\numLocations<\numEvents$, which is typically the case.
Moreover, the term $\numThreads\numLocations$ can be refined to $\sum_{x}\numThreads_x$, where $\numThreads_x$ is the number of threads reading on variable $x$.
Although $\sum_{x}\numThreads_x=O(\numThreads\numLocations)$ in the worst-case, in practice most variables are accessed by only a few threads, reducing the space complexity to $O(\numThreads^2+\numLocations+\numLocks)$.
In contrast, existing algorithms do not enjoy this property.
This makes $\ouralgo$ the first algorithm with a memory footprint small enough to be employed in a runtime setting.
%, achieving the same memory footprint as the TSAN race detector~\cite{Serebryany-2009}.

\Paragraph{2.~AtomSanitizer for increasing-cycle violations.}
We present a version of $\ouralgo$ focused on increasing-cycle violations, which are specific conflict-serializability violations in which the root cause can be localized to a single offending transaction, and has been utilized in prior work~\cite{Flanagan-2008, Biswas-2014, Ma-2021}.
$\ouralgo$ detects increasing-cycle violations in time $O(\numEvents\numThreads)$ and space $O(\numThreads(\numThreads+\numLocations)+\numLocks)$.
Although this does not improve the time complexity of the state-of-the-art (in particular $\regiontrack$~\cite{Ma-2021}), it improves the space complexity by a term $\numThreads\numLocks$, and in practice achieves a memory footprint of $O(\numThreads^2+\numLocations+\numLocks)$, by a similar analysis as above.
%This is significantly less than the memory footprint of $\regiontrack$, which remains $O(\numThreads(\numThreads\numLocations+\numLocks))$ for increasing-cycle violations.

\Paragraph{3.~Lower bounds.}
We complement our improved upper bounds with two interesting lower bounds that characterize fundamental computational resources required to detect conflict-serializability violations in an streaming fashion.
First, recall that, if we relinquish the practical requirement that each violation is reported in an online fashion as soon as it occurs (or shortly after, e.g., when the latest offending transaction completes), the problem can be trivially solved in $O(\numEvents)$ time.
\emph{Is it possible to achieve linear time also in the online setting?}
We prove that this is unlikely, in the sense that, reporting each transaction that violates atomicity in an online fashion must take $\Omega(\numEvents^{3/2-\epsilon})$ time, for any fixed $\epsilon>0$, under the popular Online Matrix-Vector hypothesis~\cite{Henzinger2015}.
Second, note that the space usage of $\ouralgo$ is practically sublinear in $\numEvents$, but can become linear in $\numEvents$ when there are $\numLocations=\Omega(\numEvents)$ locations.
\emph{Is a sublinear space bound always possible?}
We prove that this is not the case, in the sense that any streaming algorithm for conflict-serializability (or even increasing-path) violations must use $\Omega(\numEvents)$ space when $\numLocations=\Omega(\numEvents)$, in the worst case.

\Paragraph{4.~Implementation and experiments.}
We implement $\ouralgo$ inside the $\RAPID$ framework of dynamic analyses~\cite{Umang-rapid}.
Our tool processes a trace $\tr$ and reports the first transaction which completes a cycle in the underlying conflict graph, as per standard. 
We compare the performance of our tool to existing algorithms in the literature, on standard benchmarks.
$\ouralgo$ is faster on each benchmark, achieving an average speedup of $2.0\times$, $4.5\times$, and $12\times$ over $\regiontrack$, $\aerodrome$, and $\velodrome$, respectively.
Moreover, $\ouralgo$ is always faster than $\regiontrack$ in detecting increasing-cycle violations (this is the only baseline that makes sound and complete reports of increasing-cycles), with an average speedup of $1.7\times$.

We also implement $\ouralgo$ inside the widely used LLVM $\TSAN$ framework~\cite{Serebryany-2009} for runtime analysis. 
The core data structure of $\ouralgo$ is designed to use little locking in a runtime setting, and our implementation exploits this feature to avoid expensive synchronization between threads most of the time.
We perform scalability experiments on a standard set of benchmarks, and our results show that $\ouralgo$ incurs only a marginal slowdown and memory cost over the data-race detector engine of $\TSAN$ ($1.49\times$ and $1.06\times$, respectively, on average).
Relative to the original execution, the average time and memory overhead is $4.63\times$ and $3.83\times$, respectively,
which generally aligns with industrial expectations~\cite{Serebryany-2009}.
To our knowledge, this is the first case of an atomicity tester that achieves such a small overhead, making it now possible to monitor atomicity at overheads that are industry standard.

%% file: sections/prelim-main.tex
%!TEX root = ../main.tex
\section{Preliminaries}
We begin with general notation on concurrent executions and the atomicity notions we consider.

\input{sections/definitions}
\input{sections/atomicity}

%% file: sections/definitions.tex
\subsection{Concurrent Execution Model}\label{SUBSEC:DEFINITIONS}

\Paragraph{Events and traces.}
A trace is a sequence of events $\tr=e_1,\dots, e_{\numEvents}$, encapsulating the execution of a concurrent program.
Given an event $e$ of $\tr$, we write $\tr[{:}e]$ to denote the prefix of $\tr$ up to (including) $e$.
In turn, an event is a tuple $e = \langle t, i, op \rangle$, where $t$ is an identifier of the thread executing $e$, $i$ is an incremental id for thread $t$ (i.e. it denotes the $i$-th event of $t$), and $op$ is the operation performed by $e$,
which is one of the following.
\begin{compactenum}
\item $\bg$, denoting that $e$ is the beginning of a transaction.
\item $\en$, denoting that $e$ is the end of a transaction.
\item $\rd(x)$, denoting that $e$ reads global variable $x$.
\item $\wt(x)$, denoting that $e$ writes to global variable $x$.
\item $\acq(\ell)$, denoting that $e$ acquires the lock $\ell$.
\item $\rel(\ell)$, denoting that $e$ releases the lock $\ell$.
\end{compactenum}

We often denote the $\bg$ and $\en$ events as $\triangleright$ and $\triangleleft$, respectively. 
We use $\tid(e)$, $\op(e)$, and $\idof{e}$, to denote the thread identifier, the operation, and identifier of $e$, respectively.
If $e$ accesses a variable $v$ or a lock $\ell$, we let $\var(e)$ denote $v$ or $\ell$, respectively.
We often denote an event $e$ as $\langle t, op \rangle$, omitting its identifier.
Note that we ignore the values read/written, as these are not important in our setting.
We also don't explicitly consider fork/join events, though these can be handled naturally.

Traces have to be well-formed:~each $\acq(\ell)$ is paired with a corresponding $\rel(\ell)$ in the same thread, and no lock is held by more than one thread simultaneously.
Moreover, for any thread $t$, the sequence of $\bg$ and $\en$ events starts with $\bg$, alternates between $\bg$ and $\en$, and ends on $\en$\footnote{Wlog, we do not consider nested critical sections/transactions. These are handled straightforwardly in our implementation.}.
We denote by $\events{\tr}$ the set of events, $\variables{\tr}$ the global variables, 
$\locks{\tr}$ the set of locks,
$\reads{\tr}$ the set of read events, and $\writes{\tr}$ the set of write events in $\tr$.

\Paragraph{Transactions.} 
Given a trace $\tr$, a transaction $T$ is a sequence of events in a thread $t$ that begins with $\langle t, \txbegin \rangle$ and  ends with the matching $\langle t, \txend \rangle$, at which point $T$ is completed.
We use $T.\bg$ to denote the begin event of $T$.
Transaction $T$ is assigned a unique, incremental id in thread $t$, accessed as $\txid{T}$.
$T$ is \emph{unary} if it contains a single event besides $\txbegin$ and $\txend$.
$T$ is \emph{open} (or \emph{active}) in a prefix $\tr[{:}e]$ if it has executed its begin but not its end event in $\tr[{:}e]$.
We write $e \in T$ if $e$ appears in $T$.
We let $\txof(e)$ be the transaction that $e$ appears in.
As per standard, we assume that every event belongs to a (possibly unary) transaction.

\Paragraph{Conflicting events.}
Two events $e_1$ and $e_2$ \emph{conflict}, denoted by $\cf{e_1}{e_2}$, if
(\romannumeral 1\relax) $\tid(e_1) = \tid(e_2)$, or
(\romannumeral 2\relax) $\op(e_1)$, $\op(e_2)\in \{\acq(\ell),\rel(\ell)\}$ for some lock $\ell$, or
(\romannumeral 3\relax) $\op(e_1)$, $\op(e_2) \in \{\wt(x), \rd(x)\}$ for some variable $x$ and at least one writes to $x$.

\Paragraph{Partial orders.} 
We make frequent use of partial orders over events $\events{\tr}$, or subsets thereof.
A (strict) partial order $P$, denoted $\PO{\tr}{P}$ is a transitive, antisymmetric and irreflexive relation over $\events{\tr}$.
We let $\POR{\tr}{P}$ be the reflexive closure of $P$.
Given two events $e_1$ and $e_2$, we say that $e_1$ is a \emph{$\PO{\tr}{P}$-immediate predecessor} of $e_2$ if
(i)~$e_1\PO{\tr}{P}e_2$ and
(ii)~there does not exist an event $e$ such that $e_1\PO{\tr}{P}e\PO{\tr}{P}e_2$.
Note that $\PO{\tr}{P}$-immediate predecessors are not unique in general, unless $P$ is total.

The trace $\tr$ defines a total order $\traceOrd{\tr}$ over $\events{\tr}$, with $e_1 \traceOrd{\tr} e_2$ iff $e_1$ occurs before $e_2$ in $\tr$. 
The \emph{thread order} $\threadOrd{\tr}$ is the smallest partial order such that $e_1 \threadOrd{\tr} e_2$ if $\tid(e_1) = \tid(e_2)$ and $e_1 \traceOrd{\tr} e_2$.

%% file: sections/atomicity.tex
\subsection{Atomicity as Serializability}\label{SUBSEC:ATOMICITY}
\label{sec:prelim-ser}

\Paragraph{Serial traces.}
As a natural consequence of concurrency, different transactions may overlap in a trace $\tr$.
We call $\tr$ \emph{serial} if there is no such overlap, i.e., for any two transaction begin events $e_1, e_2$ such that $e_1 \traceOrd{\tr} e_2$,
we have $e_3\traceOrd{\tr} e_2$, where $e_3$ is the matching end event of $e_1$.
\emph{Serializability} captures the fact that, although $\tr$ might not be serial, its behavior is equivalent to a serial trace $\tr'$ that can be obtained by permuting adjacent non-conflicting events in $\tr$.
If such a transformation exists, $\tr$ is deemed \emph{serializable}.
Otherwise, $\tr$ witnesses an atomicity violation.
In the following, we make these concepts formal.

\Paragraph{Conflict-serializability.}
For a trace $\tr$, \emph{conflict-happens-before} is the smallest partial order over $\events{\tr}$ such that, for any two events $e_1$, $e_2$, if $\cf{e_1}{e_2}$ and $e_1 \traceOrd{\tr} e_2$, then $e_1 \chb{\tr} e_2$.
We say that  $\tr$ is \emph{conflict-equivalent} to another trace $\tr'$ iff (\romannumeral 1\relax) $\events{\tr} = \events{\tr '}$, and (\romannumeral 2\relax) $\chb{\tr} = \chb{\tr'}$.
We say that $\tr$ is \emph{conflict-serializable} if there exists a serial trace $\tr'$ such that $\tr$ conflict-equivalent to $\tr'$.
Otherwise, $\tr$ exhibits a \emph{conflict-serializability} violation.

\input{figures/cs_counter_example}

\Paragraph{Examples.}
Consider the trace $\tr_1$ in \cref{fig:cs_positive_example}.
It consists of 14 events executed by three threads $t_1$, $t_2$ and $t_3$.
It has four transactions $T_1 = \txbegin, e_6, e_{9}, \txend$,
$T_2 = \txbegin, e_3, \txend$, $T_3 = \txbegin,  e_{12}, \txend$ and
$T_4 = \txbegin, e_8, e_{10}, \txend$.
We use an arrow $\rightarrow$ to represent $\chb{\tr}$ orderings across threads.
The trace $\tr_1$ is conflict-serializable, since we can construct a serial trace $\tr_2$, shown in \cref{fig:cs_positive_example_reordering}, that is conflict-equivalent to $\tr_1$.
In contrast, the trace $\tr_3$ in \cref{fig:cs_example1} is not conflict-serializable, since making the transactions atomic would require reversing the order of at least one of the $\chb{\tr}$ arrows.

\Paragraph{Increasing-cycle violations.}
When $\tr$ witnesses an atomicity violation, it is not clear which transactions to blame, and fixing it requires looking closely into the trace and possibly the program source.
\emph{Increasing-cycle violations} constitute a simple yet common class of atomicity violations on which there is a single candidate transaction to be blamed.
In particular, $\tr$ exhibits an increasing-cycle violation on a transaction $T$ if $\tr$ contains
three distinct events $e_1$, $e_2$, and $e_3$ such that
(i)~$e_1$ and $e_2$ are the $\bg$ and $\en$ events of $T$, respectively,
(ii)~$e_3\not \in T$, and
(iii)~$e_1 \chb{\tr} e_3 \chb{\tr} e_2$.
%Intuitively, $T$ exhibits a cyclic $\chb{\tr}$ path that involves events outside $T$.

\input{figures/inc_example}

\Paragraph{Examples.}
The trace $\tr_4$ in \cref{fig:inc_example} witnesses an increasing-cycle violation on $T_1$, since 
(i)~$e_1$ and $e_{12}$ are the $\bg$ and $\en$ events of $T_1$,
(ii)~$e_{9}\not \in T_1$, and
(iii)~$e_1\chb{\tr_4}e_9\chb{\tr_4}e_{12}$.
In contrast, $\tr_3$ does not witness an increasing-cycle violation, although it is also not conflict-serializable.

%% file: figures/cs_counter_example.tex
%!TEX root = ../main.tex
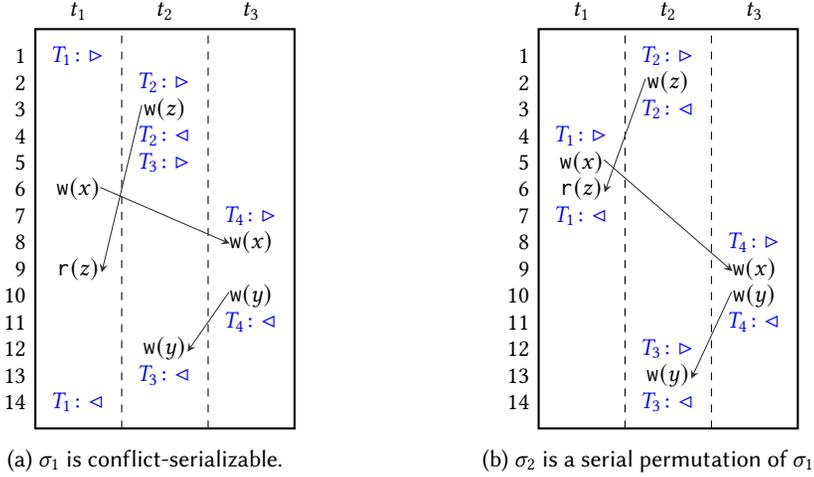
\begin{figure}[h!]
\centering
\def\scaleboxvalue{1}
\begin{subfigure}[b]{0.45\textwidth}
\centering
\scalebox{\scaleboxvalue} {
\begin{tikzpicture}[xscale=1.15, yscale=0.35, 
every node/.style={font=\small},
transaction/.style={draw, dashed, thick},
event/.style={anchor=center, inner sep=1pt},
-, >={Triangle[width=10pt,length=5pt]}]

% starting coordinates
\def\xstart{0} % starting x coordinate
\def\ystart{0} % starting y coordinate

% width and height of the rectangle
\def\width{3}  % width of the rectangle
\def\height{15} % height of the rectangle

% write labels for the three threads
\node[anchor=south] at (\xstart + \width/6, \ystart + \height) {$t_1$};
\node[anchor=south] at (\xstart + \width/2, \ystart + \height) {$t_2$};
\node[anchor=south] at (\xstart + 5*\width/6, \ystart + \height) {$t_3$};

% Draw the three transaction lines
\draw[dashed] (\xstart + \width/3,\ystart) -- (\xstart + \width/3,\ystart + \height);
\draw[dashed] (\xstart + 2*\width/3, \ystart) -- (\xstart + 2*\width/3, \ystart + \height);

% Draw the grid
\foreach \y in {1,...,14} {
\pgfmathtruncatemacro\yy{15-\y} % Compute reverse ID: 10 -> 1
\node[left] at (\xstart,\ystart+\y) {\yy};  % Keep the coordinates, but reverse the label
}

% events in t1
\node[event] (e1) at (\xstart + \width/6, \ystart + \height - 1) {\bbg{T_1}};   
\node[event] (e6) at (\xstart + \width/6, \ystart + \height - 6) {$\wt(x)$};

\node[event] (e9) at (\xstart + \width/6, \ystart + \height - 9) {$\rd(z)$};
\node[event] (e14) at (\xstart + \width/6, \ystart + \height - 14) {\ben{T_1}};

% events in t2
\node[event] (e2) at (\xstart + \width/2, \ystart + \height - 2) {\bbg{T_2}};
\node[event] (e3) at (\xstart + \width/2, \ystart + \height - 3) {$\wt(z)$};
\node[event] (e4) at (\xstart + \width/2, \ystart + \height - 4) {\ben{T_2}};

\node[event] (e5) at (\xstart + \width/2, \ystart + \height - 5) {\bbg{T_3}};
\node[event] (e12) at (\xstart + \width/2, \ystart + \height - 12) {$\wt(y)$};
\node[event] (e13) at (\xstart + \width/2, \ystart + \height - 13) {\ben{T_3}};

% events in t3
\node[event] (e7) at (\xstart + 5*\width/6, \ystart + \height - 7) {\bbg{T_4}};
\node[event] (e8) at (\xstart + 5*\width/6, \ystart + \height - 8) {$\wt(x)$};
\node[event] (e10) at (\xstart + 5*\width/6, \ystart + \height - 10) {$\wt(y)$};
\node[event] (e11) at (\xstart + 5*\width/6, \ystart + \height - 11) {\ben{T_4}};

% e6 -> e8
\draw[po] (e6.east) to (e8.west);
% e10 -> e12
\draw[po] (e10.west) to (e12.east);
% e3 -> e9
\draw[po] (e3.west) to (e9.east);

% Draw border
\draw[thick] (\xstart, \ystart) rectangle ++(\width, \height);

% \draw[thick] (0.5,0.5) rectangle (3.5,12.5);
\end{tikzpicture}
}
\caption{$\tr_1$ is conflict-serializable.}
\label{fig:cs_positive_example}
\end{subfigure}
\quad
\begin{subfigure}[b]{0.45\textwidth}
\centering
\scalebox{\scaleboxvalue} {
\begin{tikzpicture}[xscale=1.15, yscale=0.35, 
every node/.style={font=\small},
transaction/.style={draw, dashed, thick},
event/.style={anchor=center, inner sep=1pt},
-, >={Triangle[width=10pt,length=5pt]}]

% starting coordinates
\def\xstart{0} % starting x coordinate
\def\ystart{0} % starting y coordinate

% width and height of the rectangle
\def\width{3}  % width of the rectangle
\def\height{15} % height of the rectangle

% write labels for the three threads
\node[anchor=south] at (\xstart + \width/6, \ystart + \height) {$t_1$};
\node[anchor=south] at (\xstart + \width/2, \ystart + \height) {$t_2$};
\node[anchor=south] at (\xstart + 5*\width/6, \ystart + \height) {$t_3$};

% Draw the three transaction lines
\draw[dashed] (\xstart + \width/3,\ystart) -- (\xstart + \width/3,\ystart + \height);
\draw[dashed] (\xstart + 2*\width/3, \ystart) -- (\xstart + 2*\width/3, \ystart + \height);

% Draw the grid
\foreach \y in {1,...,14} {
\pgfmathtruncatemacro\yy{15-\y} % Compute reverse ID: 10 -> 1
\node[left] at (\xstart,\ystart+\y) {\yy};  % Keep the coordinates, but reverse the label
}

% events in t1
\node[event] (e4) at (\xstart + \width/6, \ystart + \height - 4) {\bbg{T_1}};   
\node[event] (e5) at (\xstart + \width/6, \ystart + \height - 5) {$\wt(x)$};

\node[event] (e6) at (\xstart + \width/6, \ystart + \height - 6) {$\rd(z)$};
\node[event] (e7) at (\xstart + \width/6, \ystart + \height - 7) {\ben{T_1}};

% events in t2
\node[event] (e1) at (\xstart + \width/2, \ystart + \height - 1) {\bbg{T_2}};
\node[event] (e2) at (\xstart + \width/2, \ystart + \height - 2) {$\wt(z)$};
\node[event] (e3) at (\xstart + \width/2, \ystart + \height - 3) {\ben{T_2}};

\node[event] (e12) at (\xstart + \width/2, \ystart + \height - 12) {\bbg{T_3}};
\node[event] (e13) at (\xstart + \width/2, \ystart + \height - 13) {$\wt(y)$};
\node[event] (e14) at (\xstart + \width/2, \ystart + \height - 14) {\ben{T_3}};

% events in t3
\node[event] (e8) at (\xstart + 5*\width/6, \ystart + \height - 8) {\bbg{T_4}};
\node[event] (e9) at (\xstart + 5*\width/6, \ystart + \height - 9) {$\wt(x)$};
\node[event] (e10) at (\xstart + 5*\width/6, \ystart + \height - 10) {$\wt(y)$};
\node[event] (e11) at (\xstart + 5*\width/6, \ystart + \height - 11) {\ben{T_4}};

% e2 -> e6
\draw[po] (e2.west) to (e6.east);
% e5 -> e9
\draw[po] (e5.east) to (e9.west);
% e10 -> e13
\draw[po] (e10.west) to (e13.east);

% Draw border
\draw[thick] (\xstart, \ystart) rectangle ++(\width, \height);

% \draw[thick] (0.5,0.5) rectangle (3.5,12.5);
\end{tikzpicture}
}
\caption{$\tr_2$ is a serial permutation of $\tr_1$}
\label{fig:cs_positive_example_reordering}
\end{subfigure}

\caption{A non-serial trace that is conflict-serializable.}
\label{fig:cs_counter_example}
\Description{This figure is an example case of serializable trace}
\end{figure}

%% file: figures/inc_example.tex
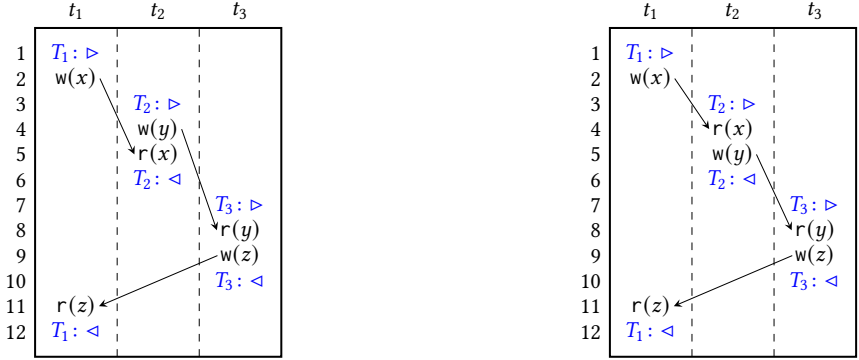
\begin{figure}[htbp]
    \centering
    \def\scaleboxvalue{0.95}
    \begin{subfigure}[b]{0.45\textwidth}
    		\centering
        \scalebox{\scaleboxvalue} {
        \begin{tikzpicture}[xscale=1.15, yscale=0.35, 
            every node/.style={font=\small},
            transaction/.style={draw, dashed, thick},
            event/.style={anchor=center, inner sep=1pt},
            -, >={Triangle[width=10pt,length=5pt]}]
        
            % starting coordinates
            \def\xstart{0} % starting x coordinate
            \def\ystart{0} % starting y coordinate
            
            % width and height of the rectangle
            \def\width{3}  % width of the rectangle
            \def\height{13} % height of the rectangle
        
            % write labels for the three threads
            \node[anchor=south] at (\xstart + \width/6, \ystart + \height) {$t_1$};
            \node[anchor=south] at (\xstart + \width/2, \ystart + \height) {$t_2$};
            \node[anchor=south] at (\xstart + 5*\width/6, \ystart + \height) {$t_3$};
            
            % Draw the three transaction lines
            \draw[dashed] (\xstart + \width/3,\ystart) -- (\xstart + \width/3,\ystart + \height);
            \draw[dashed] (\xstart + 2*\width/3, \ystart) -- (\xstart + 2*\width/3, \ystart + \height);
            
            % Draw the grid
            \foreach \y in {1,...,12} {
                \pgfmathtruncatemacro\yy{13-\y} % Compute reverse ID: 12 -> 1
                \node[left] at (\xstart,\ystart+\y) {\yy};  % Keep the coordinates, but reverse the label
            }
            
            % events in t1
            \node[event] (e1) at (\xstart + \width/6, \ystart + \height - 1) {\bbg{T_1}};   
            \node[event] (e2) at (\xstart + \width/6, \ystart + \height - 2) {$\wt(x)$};
            \node[event] (e11) at (\xstart + \width/6, \ystart + \height - 11) {$\rd(z)$};
            \node[event] (e12) at (\xstart + \width/6, \ystart + \height - 12) {\ben{T_1}}; 
            
            % events in t2
            \node[event] (e3) at (\xstart + \width/2, \ystart + \height - 3) {\bbg{T_2}};
            \node[event] (e4) at (\xstart + \width/2, \ystart + \height - 4) {$\wt(y)$};
            \node[event] (e5) at (\xstart + \width/2, \ystart + \height - 5) {$\rd(x)$};
            \node[event] (e6) at (\xstart + \width/2, \ystart + \height - 6) {\ben{T_2}};
            
            % events in t3
            \node[event] (e7) at (\xstart + 5*\width/6, \ystart + \height - 7) {\bbg{T_3}};
            \node[event] (e8) at (\xstart + 5*\width/6, \ystart + \height - 8) {$\rd(y)$};
            \node[event] (e9) at (\xstart + 5*\width/6, \ystart + \height - 9) {$\wt(z)$};
            \node[event] (e10) at (\xstart + 5*\width/6, \ystart + \height - 10) {\ben{T_3}};
            
            % e2 -> e5
            \draw[->] (e2.east) -- (e5.west);
            % e4 -> e8
            \draw[->] (e4.east) -- (e8.west);
            % e9 -> e11
            \draw[->] (e9.west) -- (e11.east);
            
            % Draw border
            \draw[thick] (\xstart, \ystart) rectangle ++(\width, \height);
        
            % \draw[thick] (0.5,0.5) rectangle (3.5,12.5);
        \end{tikzpicture}
        }
        \caption{$\tr_3$ witnesses a conflict-serializability violation.}
        \label{fig:cs_example1}
    \end{subfigure}
\hfill
    \begin{subfigure}[b]{0.45\textwidth}
    \centering
    \scalebox{\scaleboxvalue} {
    \begin{tikzpicture}[xscale=1.15, yscale=0.35, 
        every node/.style={font=\small},
        transaction/.style={draw, dashed, thick},
        event/.style={anchor=center, inner sep=1pt},
        -, >={Triangle[width=10pt,length=5pt]}]
    
        % starting coordinates
        \def\xstart{0} % starting x coordinate
        \def\ystart{0} % starting y coordinate
        
        % width and height of the rectangle
        \def\width{3}  % width of the rectangle
        \def\height{13} % height of the rectangle
    
        % write labels for the three threads
        \node[anchor=south] at (\xstart + \width/6, \ystart + \height) {$t_1$};
        \node[anchor=south] at (\xstart + \width/2, \ystart + \height) {$t_2$};
        \node[anchor=south] at (\xstart + 5*\width/6, \ystart + \height) {$t_3$};
        
        % Draw the three transaction lines
        \draw[dashed] (\xstart + \width/3,\ystart) -- (\xstart + \width/3,\ystart + \height);
        \draw[dashed] (\xstart + 2*\width/3, \ystart) -- (\xstart + 2*\width/3, \ystart + \height);
        
        % Draw the grid
        \foreach \y in {1,...,12} {
            \pgfmathtruncatemacro\yy{13-\y} % Compute reverse ID: 12 -> 1
            \node[left] at (\xstart,\ystart+\y) {\yy};  % Keep the coordinates, but reverse the label
        }
        
        % events in t1
        \node[event] (e1) at (\xstart + \width/6, \ystart + \height - 1) {\bbg{T_1}};   
        \node[event] (e2) at (\xstart + \width/6, \ystart + \height - 2) {$\wt(x)$};
        \node[event] (e11) at (\xstart + \width/6, \ystart + \height - 11) {$\rd(z)$};
        \node[event] (e12) at (\xstart + \width/6, \ystart + \height - 12) {\ben{T_1}}; 
        
        % events in t2
        \node[event] (e3) at (\xstart + \width/2, \ystart + \height - 3) {\bbg{T_2}};
        \node[event] (e4) at (\xstart + \width/2, \ystart + \height - 4) {$\rd(x)$};
        \node[event] (e5) at (\xstart + \width/2, \ystart + \height - 5) {$\wt(y)$};
        \node[event] (e6) at (\xstart + \width/2, \ystart + \height - 6) {\ben{T_2}};
        
        % events in t3
        \node[event] (e7) at (\xstart + 5*\width/6, \ystart + \height - 7) {\bbg{T_3}};
        \node[event] (e8) at (\xstart + 5*\width/6, \ystart + \height - 8) {$\rd(y)$};
        \node[event] (e9) at (\xstart + 5*\width/6, \ystart + \height - 9) {$\wt(z)$};
        \node[event] (e10) at (\xstart + 5*\width/6, \ystart + \height - 10) {\ben{T_3}};
        
        % e2 -> e4
        \draw[->] (e2.east) -- (e4.west);
        % e5 -> e8
        \draw[->] (e5.east) -- (e8.west);
        % e9 -> e11
        \draw[->] (e9.west) -- (e11.east);
        
        % Draw border
        \draw[thick] (\xstart, \ystart) rectangle ++(\width, \height);
    
        % \draw[thick] (0.5,0.5) rectangle (3.5,12.5);
    \end{tikzpicture}
    }
    \caption{$\tr_4$ witnesses an increasing-cycle violation.}
    \label{fig:inc_example}
    \end{subfigure}
    \caption{Two non-conflict-serializable traces.}
    \Description{This figure is an example case of increasing cycle.}
\end{figure}

%% file: sections/serializability-main.tex
%!TEX root = ../main.tex
\section{Detecting Atomicity Violations}\label{sec:atomicity_testing}

In this section we present $\ouralgo$ for detecting conflict-serializability violations.
We start with some basic lemmas that capture serializability (\cref{subsec:serializability_conditions}), 
followed by the presentation of $\ouralgo$ (\cref{subsec:atomsanitizer}), and its correctness and complexity (\cref{subsec:ouralgo_correctness_complexity}).
We conclude with some lower bounds on the time and space complexity of monitoring atomicity (\cref{subsec:lower_bounds}).

\input{sections/serializability_conditions}
\input{sections/ouralgo}

\input{sections/ouralgo_example}

\input{sections/ouralgo_correctness_complexity}
\input{sections/lower_bounds}

%% file: sections/serializability_conditions.tex
\subsection{Serializability Conditions}\label{subsec:serializability_conditions}

\Paragraph{Transaction-happens-before.}
Given a trace $\tr$, the \emph{transaction-happens-before} relation $\thb{\tr}$ is the smallest transitive relation over the transactions of $\tr$ with the property that, for any two distinct transactions $T_1$ and $T_2$, if there exist $e_1\in T_1$ and $e_2\in T_2$ such that $e_1\chb{\tr}e_2$, then $T_1\thb{\tr}T_2$.
Conceptually, if $\thb{\tr}$ is a partial order, we can arrive at a serial trace $\tr'$ witnessing the serializability of $\tr$ by linearizing $\thb{\tr}$.
This is formally captured in the following lemma.

\begin{lemma}[\cite{Flanagan-2008}]\label{lem:thb_acyclicity_condition}
A trace $\tr$ is conflict-serializable iff $\thb{\tr}$ is irreflexive.
\end{lemma}
For example, the trace $\tr_3$ (\cref{fig:cs_example1}) yields a cyclic $\thb{\tr_3}$ due to the orderings $T_1\to T_2\to T_3\to T_1$.
It is easy to see that $\thb{\tr}$ satisfies the following monotonicity property.

\begin{remark}\label{rem:thb_monotonic}
For any event $e$ of $\tr$, we have $\thb{\tr[{:}e]}\subseteq \thb{\tr}$.
\end{remark}

Monotonicity implies that there is a smallest prefix of $\tr$ when a conflict-serializability violation occurs (i.e., $\thb{\tr[:e]}$ becomes cyclic).
Different atomicity testers employ different techniques for detecting cycles in $\thb{\tr}$ by performing a single pass (i.e., streaming-style) of $\tr$ and checking whether $\thb{\tr}$ is irreflexive at each step.
$\ouralgo$ is based on the notions of latest remote and earliest local transactions, that we develop below.

\Paragraph{Latest remote transactions and earliest local transactions.}
For a trace $\tr$ and two threads $t_1$, $t_2$, we define the \emph{latest remote transaction} of $t_2$ known to $t_1$, $\LatestRemoteTransaction{\tr}{t_1}(t_2)$, as follows.
\begin{compactitem}
\item If $t_1$ is not executed in $\tr$, then $\LatestRemoteTransaction{\tr}{t_1}(t_2)=\bot$. 
Otherwise, let $T_1$ be the latest transaction of $t_1$ in $\tr$.
\item If $t_2$ has executed a transaction $T_2$ such that $T_2\thb{\tr}T_1$, then $\LatestRemoteTransaction{\tr}{t_1}(t_2)$ is the latest such transaction $T_2$, otherwise $\LatestRemoteTransaction{\tr}{t_1}(t_2)=\bot$.
\end{compactitem}

If $\LatestRemoteTransaction{\tr}{t_1}(t_2)\neq \bot$, we define the \emph{earliest local transaction} in $t_1$ for $t_2$ as $\EarliestReceiverTransaction{\tr}{t_1}(t_2)=T'_1$, where $T'_1$ is the earliest transaction in $t_1$ such that $\LatestRemoteTransaction{\tr}{t_1}(t_2)\thb{\tr}T'_1$.
Note that such a transaction must exist, being either $T_1$ or some earlier transaction.
On the other hand, if $\LatestRemoteTransaction{\tr}{t_1}(t_2) = \bot$, we simply let $\EarliestReceiverTransaction{\tr}{t_1}(t_2)=\bot$.
In words, $\LatestRemoteTransaction{\tr}{t_1}(t_2)$ is the latest transaction of thread $t_2$ that $t_1$ is aware of,
while $\EarliestReceiverTransaction{\tr}{t_1}(t_2)$ is the earliest transaction of $t_1$ that was made aware of that transaction of $t_2$.

\Paragraph{Example.}
Consider the trace $\tr_5$ in \cref{fig:cs_involved_example}.
At the event $e_{6}$, we have $\LatestRemoteTransaction{\tr_5[:e_{6}]}{t_3}(t_1) = T_1$ and $\EarliestReceiverTransaction{\tr_5[:e_{6}]}{t_3}(t_1) = T_2$.
At the event $e_{13}$, we have $\LatestRemoteTransaction{\tr_5[:e_{13}]}{t_3}(t_2) = T_3$ and $\EarliestReceiverTransaction{\tr_5[:e_{13}]}{t_3}(t_2) = T_5$.
Finally, at the end of $\tr_5$, we have $\LatestRemoteTransaction{\tr_5}{t_3}(t_2) = T_4$ and $\EarliestReceiverTransaction{\tr_5}{t_3}(t_2) = T_6$, as well as 
$\LatestRemoteTransaction{\tr_5}{t_3}(t_1) = T_1$ and $\EarliestReceiverTransaction{\tr_5}{t_3}(t_1) = T_2$.

The operating principle of $\ouralgo$ is based on the following lemma, which characterizes atomicity violations in terms of $\chb{\tr}$, $\LatestRemoteTransaction{\tr}{t_1}(t_2)$, and $\EarliestReceiverTransaction{\tr}{t_1}(t_2)$.

\begin{restatable}{lemma}{lemopenthbcondition}\label{lem:open_thb_condition}
A trace $\tr$ is not conflict-serializable iff there exist two distinct events $e'$ and $e''$ such that the following conditions hold, where $t'=\threadof{e'}$ and $t''=\threadof{e''}$:
\begin{enumerate*}[label=(\roman*)]
\item $t'\neq t''$ and $e'$ is the $\chb{\tr}$-immediate predecessor of $e''$,
\item $\LatestRemoteTransaction{\tr[{:}e]}{t'}(t'')=\txof(e'')$, and
\item $\EarliestReceiverTransaction{\tr[{:}e]}{t'}(t'').\bg \threadOrdR{\tr} \txof(e').\bg$,
where $e$ is a $\traceOrd{\tr}$-immediate predecessor of $e''$.
\end{enumerate*}
\end{restatable}
\input{figures/cs_involved_example}
Intuitively, condition (i) captures that $\txof(e')\thb{\tr}\txof(e'')$ by the definition of $\thb{\tr}$, while conditions (ii) and (iii) imply that $\txof(e'')\thb{\tr}\txof(e')$.
This is because (ii) states that $t'$ has learned the transaction $\txof(e'')$, and (iii) states that there exists a transaction $T'$ in $t'$ such that $\txof(e'') \thb{\tr} T' \thbR{\tr} \txof(e')$.
Thus, we have a violation as per \cref{lem:thb_acyclicity_condition}, since $\thb{\tr}$ is not irreflexive.

\Paragraph{Example.}
Consider the trace $\tr_3$ in \cref{fig:cs_example1}. 
At the event $e_{11}$, \cref{lem:open_thb_condition}, reports a conflict-serializability violation.
The events $e_9$ and $e_{11}$ correspond to $e'$ and $e''$ in \cref{lem:open_thb_condition}, respectively.
Here we have that
(i)~$e_9$ is a $\chb{\tr_3}$-immediate predecessor of $e_{11}$,
(ii)~$\LatestRemoteTransaction{\tr_3[:e_{10}]}{t_3}(t_1)=T_1$, and
(iii)~$\EarliestReceiverTransaction{\tr_3[:e_{10}]}{t_3}(t_1).\bg \threadOrdR{\tr_3} T_3.\bg$.
Now consider the trace $\tr_1$ in \cref{fig:cs_positive_example}.
At the event $e_{12}$, we have that $e_{10}$ is a $\chb{\tr_1}$-immediate predecessor of $e_{12}$, but $\LatestRemoteTransaction{\tr_1[:e_{11}]}{t_3}(t_2)=T_2 \neq \txof(e_{12})$, thus \cref{lem:open_thb_condition} does not report a conflict-serializability violation. 

%% file: figures/cs_involved_example.tex
%!TEX root = ../main.tex
\begin{wrapfigure}[15]{r}{0.32\textwidth}
	% \hfill
	\vspace{-0.5cm}
	% \begin{subfigure}[b]{0.35\textwidth}
		\def\scaleboxvalue{1}
		\centering
		% \scalebox{\scaleboxvalue}{
		\begin{tikzpicture}[
			xscale=1.15,
			yscale=0.35,
			every node/.style={font=\small},
			transaction/.style={draw, dashed, thick},
			event/.style={anchor=center, inner sep=1pt},
			-,
			>={Triangle[width=10pt,length=5pt]}
		]
			% starting coordinates
			\def\xstart{0} % starting x coordinate
			\def\ystart{0} % starting y coordinate

			% width and height of the rectangle
			\def\width{3} % width of the rectangle
			\def\height{17} % height of the rectangle

			% write labels for the three threads
			\node[anchor=south] at (\xstart + \width/6, \ystart + \height) {$t_{1}$};
			\node[anchor=south] at (\xstart + \width/2, \ystart + \height) {$t_{2}$}; \node[anchor=south]
			at (\xstart + 5*\width/6, \ystart + \height) {$t_{3}$};

			% Draw the three transaction lines
			\draw[dashed] (\xstart + \width/3, \ystart) -- (\xstart + \width/3, \ystart + \height);
			\draw[dashed] (\xstart + 2*\width/3, \ystart) -- (\xstart + 2*\width/3, \ystart + \height);

			% Draw the grid
			\foreach \y in {1,...,16}
			{ \pgfmathtruncatemacro\yy{17-\y} % Compute reverse ID: 16 -> 1
			\node[left] at (\xstart,\ystart+\y) {\yy}; % Keep the coordinates, but reverse the label
			}

			% events in t1
			\node[event] (e1) at (\xstart + \width/6, \ystart + \height - 1) {\bbg{T_1}};
			\node[event] (e2) at (\xstart + \width/6, \ystart + \height - 2) {$\wt(x)$};
            % \node[event] (e3) at (\xstart + \width/6, \ystart + \height - 3) {\ben{T_1}};

			% events in t2
			\node[event] (e6) at (\xstart + \width/2, \ystart + \height - 6) {\bbg{T_3}};
            \node[event] (e7) at (\xstart + \width/2, \ystart + \height - 7) {$\wt(y)$};
            \node[event] (e8) at (\xstart + \width/2, \ystart + \height - 8) {$\wt(x)$};
            \node[event] (e9) at (\xstart + \width/2, \ystart + \height - 9) {\ben{T_3}};

			\node[event] (e10) at (\xstart + \width/2, \ystart + \height - 10) {\bbg{T_4}};
            \node[event] (e11) at (\xstart + \width/2, \ystart + \height - 11) {$\wt(z)$};

			% events in t3
			\node[event] (e3) at (\xstart + 5*\width/6, \ystart + \height - 3) {\bbg{T_2}};
            \node[event] (e4) at (\xstart + 5*\width/6, \ystart + \height - 4) {$\rd(x)$};
            \node[event] (e5) at (\xstart + 5*\width/6, \ystart + \height - 5) {\ben{T_2}};

			\node[event] (e12) at (\xstart + 5*\width/6, \ystart + \height - 12) {\bbg{T_5}};
			\node[event] (e13) at (\xstart + 5*\width/6, \ystart + \height - 13) {$\rd(y)$};
			% \node[event] (e17) at (\xstart + 5*\width/6, \ystart + \height - 17) {$\rd(y)$};
			\node[event] (e14) at (\xstart + 5*\width/6, \ystart + \height - 14) {\ben{T_5}}; 

			\node[event] (e15) at (\xstart + 5*\width/6, \ystart + \height - 15) {\bbg{T_6}};
			\node[event] (e16) at (\xstart + 5*\width/6, \ystart + \height - 16) {$\rd(z)$};

			% e2 -> e4
			\draw[->] (e2.east) -- (e4.west);
			% e2 -> e8
			\draw[->] (e2.east) -- (e8.west);
			% e4 -> e8
			\draw[->] (e4.west) -- (e8.east);
			% e7 -> e13
			\draw[->] (e7.east) -- (e13.west);
			% e11 -> e16
			\draw[->] (e11.east) -- (e16.west);

			% Draw border
			\draw[thick] (\xstart, \ystart) rectangle ++(\width, \height);

			% \draw[thick] (0.5,0.5) rectangle (3.5,12.5);
		\end{tikzpicture}
		% }
		\Description{This figure is an example case of serializable trace}
		\caption{A trace $\tr_5$.}
		\label{fig:cs_involved_example}
	% \end{subfigure}
\end{wrapfigure}

%% file: sections/ouralgo.tex
\subsection{The AtomSanitizer Algorithm}\label{subsec:atomsanitizer}

We now present $\ouralgo$ for detecting conflict-serializability violations based on \cref{lem:open_thb_condition}.

\Paragraph{Intuition.}
$\ouralgo$ performs a single pass on the input trace $\tr$, and produces a warning after processing the smallest prefix of $\tr$ that contains a conflict-serializability violation (recall \cref{rem:thb_monotonic}).
Similarly to other conflict-serializability algorithms, after processing an event $e$, it maintains the maximal events of $\chb{\tr[{:}e]}$, in order to derive corresponding $\thb{\tr[{:}e]}$ orderings.
Its primary (and distinctive from other works) data structure consists of two vector clocks $\XClock_{t}$ and $\YClock_{t}$, for each thread $t$, of size $\numThreads$ each.
After processing an event $e$, for each thread $t_2$, $\XClock_{t_1}(t_2)$ and $\YClock_{t_1}(t_2)$ store identifiers of a transaction $T_2$ of $t_2$ and a transaction $T_1$ of $t_1$, respectively, such that $T_2\thb{\tr[{:}e]}T_1$.
As a first attempt, in order to report violations based on \cref{lem:open_thb_condition}, it is natural to try and maintain at all events $e$ the following invariant
\[
\XClock_{t_1}(t_2)=\txid{\LatestRemoteTransaction{\tr[{:}e]}{t_1}(t_2)} \qquad \text{and} \qquad \YClock_{t_1}(t_2)=\txid{\EarliestReceiverTransaction{\tr[{:}e]}{t_1}(t_2)}
\numberthis\label{eq:invariant}
\]

as, together with the $\chb{\tr[{:}e]}$-maximal events, it allows to test for the conditions of \cref{lem:open_thb_condition}.

Unfortunately, \cref{eq:invariant} appears difficult to maintain efficiently \emph{at all events $e$}.
This is because, when processing an event $e$ of thread $t_1$, a transaction-happens-before path might be formed from $\LatestRemoteTransaction{\tr}{t_1}(t_2)$ to $t_1$ via some intermediate thread $t_3$.
If the transaction $\LatestRemoteTransaction{\tr}{t_1}(t_2)$ is not open, $t_3$ might be aware of a later transaction of $t_2$, and the information about $\LatestRemoteTransaction{\tr}{t_1}(t_2)$ is lost (see \cref{fig:lrt_lost_example}).

\input{figures/lrt_lost_example}

In order to overcome the above obstacle and regain the invariant of \cref{eq:invariant}, it is natural to attempt to store some additional information in the state of the algorithm, perhaps at the cost of impacting its running time.
One key idea behind $\ouralgo$ is that this invariant is, in fact, \emph{not necessary}.
In particular, $\ouralgo$ guarantees that \cref{eq:invariant} holds \emph{if $\LatestRemoteTransaction{\tr[{:}e]}{t_1}(t_2)$ is open in $\tr[{:}e]$}.
This is sufficient since, in \cref{lem:open_thb_condition}, $e''$ is the currently processed event of $t''$ and thus $\LatestRemoteTransaction{\tr[{:}e]}{t'}(t'')=\txof(e'')$ must be open in $\tr[{:}e]$.
Moreover, this relaxation makes the $\XClock_{t_1}$ and $\YClock_{t_1}$ clocks easy to maintain inductively, in terms of the corresponding clocks of the other threads.

\Paragraph{AtomSanitizer.}
We now describe $\ouralgo$ in detail.
The state of the algorithm is represented in the following components.
\begin{compactitem}
\item For each thread $t$, we have the vector clocks $\XClock_{t}$ and $\YClock_{t}$, as outlined above.
\item For each location $x$, we have 
\begin{enumerate*}[label=(\roman*)]
\item a variable $\lastWrite{x}$, storing the most recent write on $x$,
\item a map $\lastReads{x}$, storing the most recent read of each thread, and
\item a set $\lrd{x}$, storing the threads that have performed a read on $x$ after the most recent write on $x$.
\end{enumerate*}
\item For each lock $\ell$, we have a variable $\lastRelease{\ell}$ storing the most recent release on $\ell$.
\end{compactitem}

$\ouralgo$ is implemented in three parts.
\begin{compactitem}
\item \cref{algo:bookkeping} is the top level, defining a set of handlers for processing each event $e$ according to its type.
It keeps track of (a superset of) the maximal events of $\chb{\tr[{:}e]}$, and
performs checks on whether $\tr[{:}e]$ witnesses a conflict-serializability violation, by calling \cref{algo:cs}.
\item \cref{algo:cs} handles the conflict-serializability checks, based on the clocks $\XClock_{t_1}(t_2)$ and $\YClock_{t_1}(t_2)$ maintained collectively (for all threads) in a data structure called $\ds$.
It also updates these vector clocks by performing the appropriate data structure call in \cref{algo:ds}.
\item \cref{algo:ds} implements the data structure $\ds$ for maintaining the clocks $\XClock_{t_1}(t_2)$ and $\YClock_{t_1}(t_2)$.
\end{compactitem}

\input{material/algorithms/bookkeping}

\Paragraph{AtomSanitizer event handlers (\cref{algo:bookkeping}).} 
We now describe each handler, referring to \cref{algo:bookkeping} for details.
We do not list handlers for transaction $\bg$ and $\en$ events,
as these are only required to keep track of the id of each transaction, in a straightforward manner.
We use $w, r$ and $rel$ to denote the write, read and release events here.

\SubParagraph{Function \FunctionRead{$e$}.}
When processing a read event $e=\rd(x)$, it may case an ordering $\wt(x)\chb{\tr[{:}e]}e$, where $\wt(x)$ is the latest write on $x$.
The algorithm calls $\FunctionCheckAndUpdate$ (\cref{line:algo_atomsanitizer_read_chb_write}) to check whether this ordering creates a cycle in $\thb{\tr[{:}e]}$, and to update $\ds$ accordingly.

\SubParagraph{Function \FunctionWrite{$e$}.}
When processing a write event $e=\wt(x)$, this can result in $\numThreads$ orderings $\rd(x)\chb{\tr[{:}e]} e$, one for the last read $\rd(x)$ of each thread $t'$,
as well as one additional ordering  $\wt'(x)\chb{\tr[{:}e]}e$, where $\wt'(x)$ is the latest write on $x$.
For each of these events, the algorithm calls $\FunctionCheckAndUpdate$ (\cref{line:algo_atomsanitizer_write_chb_reads} and \ref{line:algo_atomsanitizer_write_chb_write}), as in the case of \FunctionRead{$e$}.

\SubParagraph{Function \FunctionAcquire{$e$}.}
When processing an acquire event $e=\acq(\ell)$, this can result in one ordering $\rel(\ell)\chb{\tr[{:}e]}e$, where $\rel(\ell)$ is the latest release on $\ell$.
This is handled as previously (\cref{line:algo_atomsanitizer_acquire_chb_release}).

\SubParagraph{Function \FunctionRelease{$e$}.}
When processing a release event $e=\rel(\ell)$, the relation $\chb{\tr[{:}e]}$ is only updated trivially, since all events that are $\chb{\tr[{:}e]}$-ordered before $e$ are also $\chb{\tr[{:}e']}$ before $e'$, where $e'$ is the $\threadOrd{\tr}$-immediate predecessor of $e$.
Therefore conflict-serializability checks can be avoided here, and the function only performs simple bookkeeping.

\input{material/algorithms/cs}

\Paragraph{Conflict-serializability checking (\cref{algo:cs}).}
This part of the algorithm handles new $\chb{\tr}$ orderings passed by the handlers of \cref{algo:bookkeping}, so as to check for violations and update the conflict graph.

\SubParagraph{Function \FunctionCheckAndUpdate{$v, e$}.} 
When a handler of \cref{algo:bookkeping} discovers a new ordering $v\chb{\tr[{:}e]}e$, it invokes $\FunctionCheckAndUpdate{v, e}$.
This function first queries $\ds$ to check whether \cref{lem:open_thb_condition} applies, yielding a conflict-serializability violation (\cref{line:algo_csc_reachability}).
Otherwise, it proceeds to record the new edge $\txof(v)\thb{\tr[{:}e]}\txof(e)$ in $\ds$ (\cref{line:algo_csc_insert_edge}).

\SubParagraph{Function \FunctionInsertEdge{$\pair{t_1}{j_1}, \pair{t_2}{j_2}$}.}
This function handles an edge insertion recording the ordering $\txof(v)\thb{\tr[{:}e]}\txof(e)$ in $\ds$.
Each transaction $T$ is represented as a pair $\pair{t}{j}$, where $t$ is the thread of $T$ and $j$ is the id of the begin event of $T$.
Since $\thb{\tr[{:}e]}$ is transitively closed, this function inserts in $\ds$ all transitive orderings implied through this edge, which can be discovered by making appropriate calls to $\ds$ (\cref{line:algo_csc_call_pred} and \cref{line:algo_csc_call_suc}).

\input{material/algorithms/ds} 
\Paragraph{Data Structure (\cref{algo:ds}).}
The data structure $\ds$ succinctly maintains a conflict graph by manipulating the vector clocks $\XClock_{t_1}(t_2)$ and $\YClock_{t_1}(t_2)$, for every pair of threads $t_1$ and $t_2$, initially set to $\bot$ (\cref{line:algo_ds_initialize_at}).
It provides the following interface, which is used by \cref{algo:cs}.

\SubParagraph{Function \FunctionSucc{$\pair{t_1}{j_1}, t_2$}.}
This function returns the earliest transaction $T_2$ in thread $t_2$ such that $T_1\thb{\tr[{:}e]}T_2$, where $T_1$ is the $j_1$-th transaction of thread $t_1$.
It first checks whether $t_2$ is aware of $T_1$ (i.e., whether $T_1\thb{\tr[{:}e]}T'_2$, where $T'_2$ is the latest transaction of $t_2$) by using $\XClock_{t_2}(t_1)$, which stores the id of the latest transaction in $t_1$ that $t_2$ knows of (thus, $t_2$ knows of all the earlier transactions of $t_1$).
If not, then no such transaction $T_2$ exists, and the function returns $\bot$.
Otherwise, $\txid{T_2}$ is guaranteed to be stored in $\YClock_{t_2}(t_1)$, which is returned.

\SubParagraph{Function \FunctionPred{$\pair{t_1}{j_1}, t_2$}.}
This function is symmetric to \FunctionSucc, and returns the id of the latest transaction $T_2$ of $t_2$ such that $T_2\thb{\tr[{:}e]}T_1$, where $T_1$ is the $j_1$-th transaction of $t_1$.

\SubParagraph{Function \FunctionReachable{$\pair{t_1}{j_1}, \pair{t_2}{j_2}$}.}
This function returns True iff $T_1\thb{\tr[{:}e]}T_2$, where each $T_i$ is the $j_i$-th transaction of thread $t_i$.
It is implemented via a simple $\FunctionSucc$ query.

\SubParagraph{Function \FunctionAddSucc{$\pair{t_1}{j_1}, \pair{t_2}{j_2}$}.}
This function records a new ordering $T_1\thb{\tr[{:}e]}T_2$, where each $T_i$ is the $j_i$-th transaction of thread $t_i$.
It starts by checking if $T_1$ is a later transaction of $t_1$ known to $t_2$ than the one $t_2$ is currently aware of, and if so, it updates $\XClock_{t_2}(t_1)$ and $\YClock_{t_2}(t_1)$ accordingly (\cref{line:ds-cs-interface-addsucc-1}).
Otherwise, it checks if $T_1$ is the latest transaction of $t_1$ known to $t_2$, but $T_2$ is a later transaction of $t_2$ that becomes aware of $T_1$ than the current transaction of $t_2$.
If so, it updates $\YClock_{t_2}(t_1)$, accordingly.
If none of the above holds, it is guaranteed that the ordering $T_1\thb{\tr[{:}e]}T_2$ is already present in the data structure, and is ignored.

%% file: figures/lrt_lost_example.tex
\begin{wrapfigure}[14]{r}{0.4\textwidth}
\centering
\vspace{-0.5cm}
% \begin{subfigure}[b]{0.45\textwidth}
\begin{tikzpicture}[xscale=1.5, yscale=0.35, 
every node/.style={font=\small},
transaction/.style={draw, dashed, thick},
event/.style={anchor=center, inner sep=1pt},
-, >={Triangle[width=10pt,length=5pt]}]

% starting coordinates
\def\xstart{0} % starting x coordinate
\def\ystart{0} % starting y coordinate

% width and height of the rectangle
\def\width{3}  % width of the rectangle
\def\height{14} % height of the rectangle

% write labels for the three threads
\node[anchor=south] at (\xstart + \width/6, \ystart + \height) {$t_1$};
\node[anchor=south] at (\xstart + \width/2, \ystart + \height) {$t_2$};
\node[anchor=south] at (\xstart + 5*\width/6, \ystart + \height) {$t_3$};

% Draw the three transaction lines
\draw[dashed] (\xstart + \width/3,\ystart) -- (\xstart + \width/3,\ystart + \height);
\draw[dashed] (\xstart + 2*\width/3, \ystart) -- (\xstart + 2*\width/3, \ystart + \height);

% Draw the grid
\foreach \y in {1,...,13} {
\pgfmathtruncatemacro\yy{14-\y} % Compute reverse ID: 12 -> 1
\node[left] at (\xstart,\ystart+\y) {\yy};  % Keep the coordinates, but reverse the label
}

% events in t2
\node[event] (e1) at (\xstart + \width/2, \ystart + \height - 1) {\bbg{T_1}};   
\node[event] (e2) at (\xstart + \width/2, \ystart + \height - 2) {$\wt(x)$};
\node[event] (e3) at (\xstart + \width/2, \ystart + \height - 3) {\ben{T_1}};
\node[event] (e8) at (\xstart + \width/2, \ystart + \height - 8) {\bbg{T_2}}; 
\node[event] (e9) at (\xstart + \width/2, \ystart + \height - 9) {$\wt(y)$};
% \node[event] (e12) at (\xstart + \width/6, \ystart + \height - 12) {\ben{T_2}};

% events in t3
\node[event] (e4) at (\xstart + 5*\width/6, \ystart + \height - 4) {\bbg{T_3}};
\node[event] (e5) at (\xstart + 5*\width/6, \ystart + \height - 5) {$\wt(z)$};
\node[event] (e6) at (\xstart + 5*\width/6, \ystart + \height - 6) {$\rd(x)$};
\node[event] (e7) at (\xstart + 5*\width/6, \ystart + \height - 7) {\ben{T_3}};
\node[event] (e10) at (\xstart + 5*\width/6, \ystart + \height - 10) {\bbg{T_4}};
\node[event] (e11) at (\xstart + 5*\width/6, \ystart + \height - 11) {$\rd(y)$};
% \node[event] (e13) at (\xstart + \width/2, \ystart + \height - 13) {\ben{T_3}};

% events in t1
\node[event] (e12) at (\xstart + \width/6, \ystart + \height - 12) {\bbg{T_5}};
\node[event] (e13) at (\xstart + \width/6, \ystart + \height - 13) {$\rd(z)$};
% \node[event] (e14) at (\xstart + 5*\width/6, \ystart + \height - 14) {\ben{T_4}};

% e2 -> e6
\draw[->] (e2.east) -- (e6.west);
% e9 -> e11
\draw[->] (e9.east) -- (e11.west);
% e5 -> e13
\draw[->, bend left=5] (e5.west) to (e13.east);

% Draw border
\draw[thick] (\xstart, \ystart) rectangle ++(\width, \height);

% \draw[thick] (0.5,0.5) rectangle (3.5,12.5);
\end{tikzpicture}
\caption{
A trace $\tr_6$.
$\LatestRemoteTransaction{\tr_6}{t_1}(t_2) = T_1$ via $T_3$ of $t_3$, but this fact is not recoverable from $\LatestRemoteTransaction{\tr_6}{t_3}(t_2)$, as it has progressed to $T_2$.
}
\label{fig:lrt_lost_example}
\Description{This figure is an example case of conflict serializable}
% \end{subfigure}
\end{wrapfigure}

%% file: material/algorithms/bookkeping.tex
%!TEX root = ../../main.tex
%\setlength{\textfloatsep}{12pt}

\begin{algorithm*}[ht!]
	%\small
	\DontPrintSemicolon
	\SetInd{0.28em}{0.35em}
	\BlankLine
	\vspace*{-\multicolsep}
	\begin{multicols}{3}
        \MyFunction{\FunctionInitialize{}} {
            %$\thrTxn{} \gets \bot$\\
			\ForEach{$x \in \variables{\tr}$} {
				$\lastWrite{x} \gets \bot$\\
                $\lastReads{x} \gets \bot$\\
                $\lrd{x} \gets \emptyset$\\
			}
			\ForEach{$\ell \in \locks{\tr}$} {
				$\lastRelease{\ell} \gets \bot$\\
			}
        }    
        \BlankLine

        \MyFunction{\FunctionRelease{$e$}} {
            \label{line:bookkeping-release-start}
            $\ell \gets \var(e)$\\
            $\lastRelease{\ell} \gets e$\\
        }
        \BlankLine

        \MyFunction{\FunctionRead{$e$}} {
            \label{line:bookkeping-read-start}
            $\pair{x}{t} \gets \pair{var(e)}{\tid(e)}$\\
            $w \gets \lastWrite{x}$\\
            % $t \gets \tid(e)$\\
            $\lastReads{x}[t] \gets e$\\
            $\lrd{x} \gets \lrd{x} \cup \{t\}$\\ \label{line:algo_atomsanitizer_update_lrds}
            $\FunctionCheckAndUpdate(w, e)$\\ \label{line:algo_atomsanitizer_read_chb_write}
        }
        \BlankLine

        \MyFunction{\FunctionAcquire{$e$}} {
            \label{line:bookkeping-acquire-start}
            $\ell \gets \var(e)$\\
            $rel \gets \lastRelease{\ell}$\\
            $\FunctionCheckAndUpdate(rel, e)$ \label{line:algo_atomsanitizer_acquire_chb_release}
        }
        \BlankLine

        \MyFunction{\FunctionWrite{$e$}} {
            \label{line:bookkeping-write-start}
            $x \gets var(e)$\\
            $tids \gets \lrd{x}$\\
            \ForEach{$t \in tids$} {\label{line:algo_atomsanitizer_loop_over_lrds}
                \label{line:bookkeeping-write-foreach-start}
                $r \gets \lastReads{x}[t]$\\
                $\FunctionCheckAndUpdate(r, e)$\label{line:algo_atomsanitizer_write_chb_reads}
            }
            $w \gets \lastWrite{x}$\\
            $\lastWrite{x} \gets e$\\
            $\lrd{x} \gets \emptyset$\\ \label{line:algo_atomsanitizer_write_empty_lrds}
            $\FunctionCheckAndUpdate(w, e)$\\ \label{line:algo_atomsanitizer_write_chb_write}
        }
        \BlankLine

    \end{multicols}
    \BlankLine

\caption{
$\ouralgo$ event handlers.
}
\label{algo:bookkeping}
\vspace*{-\multicolsep}
\end{algorithm*}
\normalsize

%% file: material/algorithms/cs.tex
%!TEX root = ../../main.tex
%\setlength{\textfloatsep}{12pt}

\begin{algorithm}[ht!]
	\small
	\DontPrintSemicolon
	\SetInd{0.28em}{0.35em}
	\BlankLine
	\vspace*{-\multicolsep}
	\begin{multicols}{2}
		\MyFunction{\FunctionCheckAndUpdate{$v$, $e$}} {
			\uIf{$\tid(v)\neq \tid(e)$}{
			$n_1 \gets \pair{\tid(v)}{\txid{\txof(v)}}$\\ \label{line:conflict-serializability-check-start1}
			$n_2 \gets \pair{\tid(e)}{\txid{\txof(e)}}$\\ \label{line:conflict-serializability-check-start2}
			\If{$\ds.\reachable(n_2, n_1)$} {
				\label{line:algo_csc_reachability}
				$\reportError{Conflict serializability violation}$\label{line:algo_ccs_violation}
			}
			$\insedge(n_1, n_2)$\label{line:algo_csc_insert_edge}
		}
		}    
		\BlankLine

			\MyFunction{\FunctionInsertEdge{$\pair{t_1}{j_1}, \pair{t_2}{j_2}$}}{
				\ForEach{$t_{1}' \in \threads{\tr}$}{\label{line:ds-cs-interface-ins-1}
					\ForEach{$t_{2}' \in \threads{\tr} \setminus \set{t_1'}$}{\label{line:ds-cs-interface-ins-2}
						
						\lIf{$t_{1}' = t_1$}{$j_{1}' \xleftarrow[]{}  j_1$\label{line:ds-cs-interface-ins-3}}
						\lElse{$j_{1}' \xleftarrow[]{} \ds.\pred(\pair{t_1}{ j_1}, t_{1}')$\label{line:algo_csc_call_pred}} 
						
						\lIf{$t_{2}' = t_2$}{$j_{2}' \xleftarrow[]{} j_2$\label{line:ds-cs-interface-ins-5}}
						\lElse{$j_{2}' \xleftarrow[]{} \ds.\suc(\pair{t_2}{j_2}, t_{2}')$\label{line:algo_csc_call_suc}} 
						
						$\ds.\addsuc(\pair{t_{1}'}{j_{1}'}, \pair{t_{2}'}{j_{2}'})$\label{line:ds-cs-interface-ins-8}
					}
				}
			}
			\BlankLine
		
		\end{multicols}
		\BlankLine
		
		\caption{
			Checking conflict serializability.
		}
		\label{algo:cs}
	\end{algorithm}
	\normalsize

%% file: material/algorithms/ds.tex
%!TEX root = ../../main.tex
%\setlength{\textfloatsep}{12pt}

\begin{algorithm}[h!]
	\small
	\DontPrintSemicolon
	\SetInd{0.28em}{0.35em}
	\BlankLine
	\vspace*{-\multicolsep}
	\begin{multicols}{2}
		\MyFunction{\FunctionInitialize{}} {
			\ForEach{$t \in \threads{\tr}$}{
				\ForEach{$t' \in \threads{\tr}$}{
					$\XClock_t(t') \gets \bot$; \label{line:algo_ds_initialize_at}
					$\YClock_t(t') \gets \bot$ 
				}
			}
		}    
		\BlankLine
		
		\MyFunction{\FunctionReachable{$\pair{t_1}{j_1}$, $\pair{t_2}{j_2}$}} {
			\lIf{$t_1 = t_2$}{\Return $j_1 \leq j_2$}
			\uIf{$\suc(\pair{ t_1}{ j_1}, t_2) \leq j_2$}{\Return True}
			\lElse{\Return False}
		}    
		\BlankLine
		
		\MyFunction{\FunctionSucc{$\pair{t_1}{j_1}$, $t_2$}} {
			\lIfElse{$\XClock_{t_2}(t_1) \geq j_1$}{\Return $\YClock_{t_2}(t_1)$}{\Return $\bot$} \label{line:ds-cs-interface-succ}
		}    
		\BlankLine
		
		\MyFunction{\FunctionPred{$\pair{t_1}{j_1}$, $t_2$}} {
			 \lIfElse{$\YClock_{t_1}(t_2) \leq j_1$}{\Return $\XClock_{t_1}(t_2)$}{\Return $\bot$} \label{line:ds-cs-interface-pred}
		}    
		\BlankLine

		\MyFunction{\FunctionAddSucc{$\pair{t_1}{j_1}$, $\pair{t_2}{j_2}$}} {
			\If{$\XClock_{t_2}(t_1) < j_1$}{	\label{line:ds-cs-interface-addsucc-1}
				$\XClock_{t_2}(t_1) \gets j_1$;
				$\YClock_{t_2}(t_1) \gets j_2$ 
			} 
			\ElseIf{$\XClock_{t_2}(t_1) = j_1 \land \YClock_{t_2}(t_1) > j_2$} {
				$\YClock_{t_2}(t_1) \gets j_2$
			}
		}    
		\BlankLine
	\end{multicols}
	\BlankLine
	
	\caption{
		Data structure.
	}
	\label{algo:ds}
	% \vspace*{-\multicolsep}
\end{algorithm}
\normalsize

%% file: sections/ouralgo_example.tex
\Paragraph{Example.}
For brevity, we write $\XClock_{i}(j) = T$ and $\YClock_i(j) = T$ to mean that they store $\txid{T}$.
Consider the trace $\tr_7$ in \cref{fig:cs_algo_example}.
Edge numbers depict the order in which $\chb{\tr_7}$ relations are processed.
The data structure $\ds$ is updated as follows.
First, executing $e_7$ updates $\ds$ to $\XClock_{t_4}(t_3) = T_2$ and $\YClock_{t_4}(t_3) = T_3$, capturing that $T_2\thb{\tr_7[:e_7]}T_3$.
Second, executing $e_{11}$ updates $\ds$ to $\XClock_{t_2}(t_1) = T_1$ and $\YClock_{t_2}(t_1) = T_4$, capturing that $T_1\thb{\tr_7[:e_{11}]}T_4$.
\input{figures/cs_algo_example}%
Third, executing $e_{17}$ updates $\ds$ to $\XClock_{t_2}(t_1) = T_5$ and $\YClock_{t_2}(t_1) = T_7$, capturing that $T_5\thb{\tr_7[:e_{17}]}T_7$.
Fourth, executing $e_{18}$ updates $\ds$ to $\XClock_{t_4}(t_2) = T_4$ and $\YClock_{t_4}(t_2) = T_6$, capturing that $T_4\thb{\tr_7[:e_{18}]}T_6$.
Note that $\XClock_{t_4}(t_1)$ is still $\bot$, missing the ordering $T_1\thb{\tr_7[{:}e_{18}]}T_6$.
This is because $T_1\thb{\tr_7[{:}e_{18}]}T_6$ is formed transitively, through $T_4$, but $t_2$ is aware of the later transaction $T_5$ of $t_1$ at this point (due to edge 3).
Fifth, executing $e_{19}$ updates $\ds$ to $\XClock_{t_3}(t_2) = T_7$ and $\YClock_{t_3}(t_2) = T_2$, capturing that $T_7\thb{\tr_7[:e_{19}]}T_2$.
In addition, $\ds$ also updates $\XClock_{t_3}(t_1) = T_5$ and $\YClock_{t_3}(t_1) = T_2$,
$\XClock_{t_4}(t_2) = T_7$ and $\YClock_{t_4}(t_2) = T_3$,
$\XClock_{t_4}(t_1) = T_5$ and $\YClock_{t_4}(t_1) = T_3$ by computing the transitive closure, capturing respectively
$T_5\thb{\tr_7[:e_{19}]}T_2$, $T_7\thb{\tr_7[{:}e_{19}]}T_3$, and $T_5\thb{\tr_7[{:}e_{19}]}T_3$.
Finally, when executing $e_{20}$, we have $\XClock_{t_4}(t_1)=T_5$ and $\YClock_{t_4}(t_1)=T_3$.
Since $T_3.\bg \threadOrdR{\tr_7} \txof(e_7).\bg$ (in fact $\txof(e_7)=T_3$), $\ouralgo$ reports a conflict-serializability violation.

%% file: figures/cs_algo_example.tex
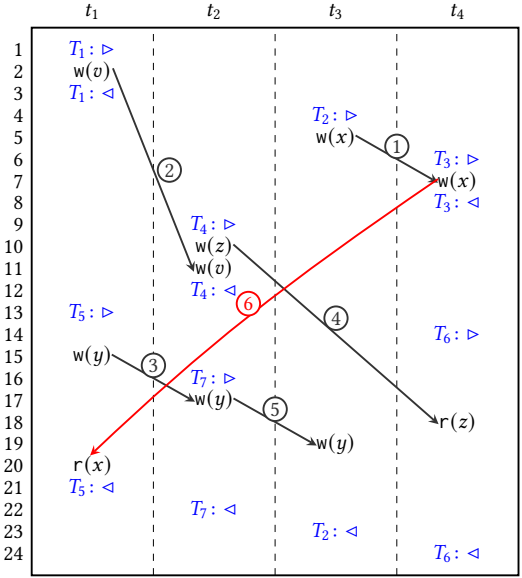
\begin{wrapfigure}[21]{r}{0.51\textwidth}
    % \vspace{-0.5mm}
% \begin{subfigure}[htbp]{0.5\textwidth}
    % \centering
    % \begin{subfigure}[t]{\textwidth}
    \def\scaleboxvalue{0.9}
    \scalebox{\scaleboxvalue} {
    \begin{tikzpicture}[xscale=1.8, yscale=0.32, 
        every node/.style={font=\small},
        transaction/.style={draw, dashed, thick},
        event/.style={anchor=center, inner sep=1pt},
        -, >={Triangle[width=10pt,length=5pt]}]
    
        % starting coordinates
        \def\xstart{0} % starting x coordinate
        \def\ystart{0} % starting y coordinate
        
        % width and height of the rectangle
        \def\width{4}  % width of the rectangle
        \def\height{25} % height of the rectangle
    
        % write labels for the three threads
        \node[anchor=south] at (\xstart + \width/8, \ystart + \height) {$t_1$};
        \node[anchor=south] at (\xstart + 3*\width/8, \ystart + \height) {$t_2$};
        \node[anchor=south] at (\xstart + 5*\width/8, \ystart + \height) {$t_3$};
        \node[anchor=south] at (\xstart + 7*\width/8, \ystart + \height) {$t_4$};
        
        % Draw the three transaction lines
        \draw[dashed] (\xstart + \width/4,\ystart) -- (\xstart + \width/4,\ystart + \height);
        \draw[dashed] (\xstart + 2*\width/4, \ystart) -- (\xstart + 2*\width/4, \ystart + \height);
        \draw[dashed] (\xstart + 3*\width/4, \ystart) -- (\xstart + 3*\width/4, \ystart + \height);
        
        % Draw the grid
        \foreach \y in {1,...,24} {
            \pgfmathtruncatemacro\yy{25-\y} % Compute reverse ID: 25 -> 1
            \node[left] at (\xstart,\ystart+\y) {\yy};  % Keep the coordinates, but reverse the label
        }
        
        % events in t1
        % T1: rhd, wt(v), T1: lhd
        \node[event] (e1) at (\xstart + \width/8, \ystart + \height - 1) {\bbg{T_1}};   
        \node[event] (e2) at (\xstart + \width/8, \ystart + \height - 2) {$\wt(v)$};
        \node[event] (e3) at (\xstart + \width/8, \ystart + \height - 3) {\ben{T_1}};
        % T4: rhd, wt(y), rd(x), T2: lhd
        \node[event] (e13) at (\xstart + \width/8, \ystart + \height - 13) {\bbg{T_5}};
        \node[event] (e15) at (\xstart + \width/8, \ystart + \height - 15) {$\wt(y)$};
        \node[event] (e20) at (\xstart + \width/8, \ystart + \height - 20) {$\rd(x)$};
        \node[event] (e21) at (\xstart + \width/8, \ystart + \height - 21) {\ben{T_5}};

        % events in t2
        % T4:rhd, wt(v) T4:lhd T6: rhd, wt(y), T6: lhd
        \node[event] (e9)  at (\xstart + 3*\width/8, \ystart + \height - 9) {\bbg{T_4}};
        \node[event] (e10)  at (\xstart + 3*\width/8, \ystart + \height - 10) {$\wt(z)$};
        \node[event] (e11)  at (\xstart + 3*\width/8, \ystart + \height - 11) {$\wt(v)$};
        \node[event] (e12)  at (\xstart + 3*\width/8, \ystart + \height - 12) {\ben{T_4}};

        \node[event] (e16) at (\xstart + 3*\width/8, \ystart + \height - 16) {\bbg{T_7}};
        \node[event] (e17) at (\xstart + 3*\width/8, \ystart + \height - 17) {$\wt(y)$};
        \node[event] (e22) at (\xstart + 3*\width/8, \ystart + \height - 22) {\ben{T_7}};

        % events in t3
        % T2: rhd, wt(x), wt(y), T2: lhd
        \node[event] (e4) at (\xstart + 5*\width/8, \ystart + \height - 4) {\bbg{T_2}};
        \node[event] (e5) at (\xstart + 5*\width/8, \ystart + \height - 5) {$\wt(x)$};
        \node[event] (e19) at (\xstart + 5*\width/8, \ystart + \height - 19) {$\wt(y)$};
        \node[event] (e23) at (\xstart + 5*\width/8, \ystart + \height - 23) {\ben{T_2}};

        % events in t4
        % T3: rhd, rd(x), T3: lhd
        \node[event] (e6)  at (\xstart + 7*\width/8, \ystart + \height - 6) {\bbg{T_3}};
        \node[event] (e7) at (\xstart + 7*\width/8, \ystart + \height - 7) {$\wt(x)$};
        \node[event] (e8) at (\xstart + 7*\width/8, \ystart + \height - 8) {\ben{T_3}};

        % T5: rhd, rd(v), wt(z), T4: lhd
        \node[event] (e14) at (\xstart + 7*\width/8, \ystart + \height - 14) {\bbg{T_6}};
        \node[event] (e18) at (\xstart + 7*\width/8, \ystart + \height - 18) {$\rd(z)$};
        \node[event] (e24) at (\xstart + 7*\width/8, \ystart + \height - 24) {\ben{T_6}};
        
        % e5 -> e7
        \draw[hb, ->] (e5.east) to node[above, circled]{$1$} (e7.west);
        % e2 -> e11
        \draw[hb, ->, bend right=0] (e2.east) to node[right=0.05,circled]{$2$} (e11.west);
        % e10 -> e18
        \draw[hb, ->, bend left=0] (e10.east) to node[above=0.05,circled]{$4$} (e18.west);
        % e15 -> e17
        \draw[hb, ->, bend right=0] (e15.east) to node[left, above, circled]{$3$} (e17.west);
        % e17 -> e19
        \draw[hb, ->, bend left=0] (e17.east) to node[above, circled]{$5$} (e19.west);
        % e7 -> e20
        \draw[vio, ->, bend right=2] (e7.west) to node[above, circled]{$6$} (e20.north);

        % e9 -> e4
        %\draw[thb, dashed, bend left=6] (e9.east) to node[left, circled]{$5$} (e4.west);
        % e9 -> e6
        %\draw[thb, dashed, bend left=6] (e9.east) to node[below=0.3,left, circled]{$6$} (e6.west);
        
        % Draw border
        \draw[thick] (\xstart, \ystart) rectangle ++(\width, \height);
    
        % \draw[thick] (0.5,0.5) rectangle (3.5,12.5);
    \end{tikzpicture}
    }
    % \end{subfigure}
    \caption{
    A trace $\tr_7$ processed by $\ouralgo$.
    Numbers indicate the sequence of $\chb{\tr_7}$ edges processed by $\ouralgo$.
    Edge 5 leads to a violation report.
}
    \label{fig:cs_algo_example}
    \Description{This figure is an example case of conflict serializable}
    % \end{subfigure}
\end{wrapfigure}

%% file: sections/ouralgo_correctness_complexity.tex
\subsection{Correctness and Complexity}\label{subsec:ouralgo_correctness_complexity}

We now state the correctness and complexity of $\ouralgo$.
The principle of operation of $\ouralgo$ is by keeping track of the latest (open) remote transactions and earliest local transactions, as described in \cref{subsec:serializability_conditions}.
Formally, it maintains the following invariant.

\begin{restatable}{lemma}{lematomsanitizerinvariant}\label{lem:atomsanitizer_invariant}
After $\ouralgo$ has processed an event $e$, for all threads $t_1$, $t_2$, if $\LatestRemoteTransaction{\tr[{:}e]}{t_1}(t_2)$ is open in $\tr[{:}e]$, we have $\XClock_{t_1}(t_2)=\txid{\LatestRemoteTransaction{\tr[{:}e]}{t_1}(t_2)}$ and $\YClock_{t_1}(t_2)=\txid{\EarliestReceiverTransaction{\tr[{:}e]}{t_1}(t_2)}$.
\end{restatable}

In turn, the invariant of \cref{lem:atomsanitizer_invariant} combines with \cref{lem:open_thb_condition} to yield the correctness of $\ouralgo$, as stated below.

\begin{restatable}{lemma}{lematomsanitizercorrectness}\label{lem:atomsanitizer_correctness}
$\ouralgo$ reports a conflict-serializability violation iff $\tr$ is not conflict-serializable.
Moreover, if $\tr$ is not conflict-serializable, then $\ouralgo$ reports the first event $e$ of $\tr$ such that $\thb{\tr[{:}e]}$ is cyclic.
\end{restatable}

Next, we turn our attention to the complexity of $\ouralgo$.
At a high level, each event $e$ that is not a write event results in a single call to \FunctionCheckAndUpdate, which takes $O(\numThreads^2)$ time, determined by the two nested loops in \cref{algo:cs} (\cref{line:ds-cs-interface-ins-1} and \cref{line:ds-cs-interface-ins-2}) for performing a transitive closure operation after inserting a single edge in \cref{line:algo_csc_insert_edge}.
On the other hand, when $e$ is a write event, \FunctionCheckAndUpdate may be called $O(\numThreads)$ times, due to the loop in \cref{line:algo_atomsanitizer_loop_over_lrds} of \cref{algo:bookkeping}, resulting in $O(\numThreads^3)$ time.
Owing to the use of the sets $\lrd{x}$, however, an amortization argument shows that across all events of $\tr$, the total number of calls to \FunctionCheckAndUpdate is $O(\numEvents)$, leading to $O(\numThreads^2)$ time per event of $\tr$, on average.
In particular, this holds because $\lrd{x}$ is reset after each write on $x$ (\cref{line:algo_atomsanitizer_write_empty_lrds}), and thus the iterations of that loop can be amortized to the read events that were executed after the last write on $x$.

\begin{restatable}{lemma}{lematomsanitizercomplexity}\label{lem:atomsanitizer_complexity}
Given a trace $\tr$ of $\numEvents$ events, $\numThreads$ threads, $\numLocations$ variables and $\numLocks$ locks,
$\ouralgo$ runs in $O(\numEvents \numThreads^2)$ time and $O(\numThreads(\numThreads+\numLocations) + \numLocks)$ space.
\end{restatable}

We remark that the number of locks contribute to the space complexity by a term $O(\numLocks)$, as opposed to $O(\numThreads\numLocks)$ in existing algorithms.
Moreover, although $\ouralgo$ uses $O(\numThreads)$ space for each location $x$ in the worst case, a more precise bound is $O(\numThreads_x)$, where $\numThreads_x$ is the number of threads reading from $x$.
In practice, most variables are accessed/read by a few threads, decreasing the space complexity  to $O(\numThreads^2+\numLocations+\numLocks)$.
The following theorem concludes the guarantees of $\ouralgo$.

\begin{restatable}{theorem}{thmatomsanitizer}\label{thm:atomsanitizer}
Let $\tr$ be an input trace of $\numEvents$ events, $\numThreads$ threads, $\numLocations$ variables and $\numLocks$ locks.
$\ouralgo$ reports a violation iff $\tr$ is not conflict-serializable, and uses $O(\numEvents \numThreads^2)$ time and $O(\numThreads(\numThreads+\numLocations) + \numLocks)$ space.
If $\tr$ is not conflict-serializable, $\ouralgo$ reports the first event $e$ of $\tr$ such that $\thb{\tr[{:}e]}$ is cyclic.
\end{restatable}

We further remark that the space usage of $\ouralgo$ is also bounded by $O(\numEvents + \numThreads^2)$, but the space bound of \cref{thm:atomsanitizer} is preferred as long as $\numThreads\numLocations < \numEvents$ (which is typically the case).

%% file: sections/lower_bounds.tex
\subsection{Lower Bounds}\label{subsec:lower_bounds}

Finally, in this section we establish two fundamental lower bounds on the time and space complexity of monitoring conflict-serializability.

\Paragraph{A super-linear online time lower bound.}
Recall that, if we relinquish the practical requirement that each violation is reported in an online fashion as soon as it occurs, transactions that violate conflict-serializability can be detected in $O(\numEvents)$ time, by simply recording the direct $\thb{\tr}$ orderings and reporting each transaction that appears on a cycle.
$\ouralgo$ operates in an online setting, and although it also achieves linear running time when $\numThreads=O(1)$, its complexity is generally superlinear.
Here we prove that such an overhead for the online problem is likely unavoidable when we require to report all violating transactions, using techniques from fine-grained complexity theory.
In particular, the Online Matrix-Vector multiplication (OMv) problem is stated as follows.
Given a $m\times m$ matrix $M$ and an online sequence of vectors $v_1,\dots, v_m$, 
after a possibly polynomial-time preprocessing of $M$, the task is to compute each product $Mv_i$ before $v_{i+1}$ is revealed.
The corresponding OMv hypothesis~\cite{Henzinger2015}
states that any algorithm for OMv requires $\Omega(m^{3-\epsilon})$ time, for any fixed $\epsilon>0$.
We establish the following super-linear lower bound based on OMv, which states that the detection on conflict-serializability violations in an online manner is a fundamentally more difficult problem.

\begin{restatable}{lemma}{lemtimelowerbound}\label{lem:time_lower_bound}
Any  algorithm for conflict-serializability violations over traces $\tr$ that reports each violating transaction when it occurs must take $\Omega(\numEvents^{3/2-\epsilon})$ time, for any fixed $\epsilon>0$, 
under the OMv hypothesis.
\end{restatable}

We further remark that the lower bound of \cref{lem:time_lower_bound} holds for a ``reasonable'' number of threads, in particular, for $\numThreads=\Theta(\sqrt{\numEvents})$ (as opposed to, e.g., $\numThreads=\Theta(\numEvents)$).
In this case, $\ouralgo$ runs in $O(\numEvents^2)$ time, as per \cref{thm:atomsanitizer}, and is thus within a  $\sqrt{\numEvents}$ factor (i.e., sublinear) from optimal.

\Paragraph{A linear space lower bound.}
Note that the space usage of $\ouralgo$ can become linear in $\numEvents$ when there are many locations, e.g., $\numLocations=\Omega(\numEvents)$.
Can we detect atomicity violations using always sublinear space (i.e., not only when $\numLocations=o(\numEvents)$)?
The following lemma states that this is impossible.

\begin{restatable}{lemma}{lemspacelowerbound}\label{lem:space_lower_bound}
Any streaming algorithm for conflict-serializability/increasing-path violations even over traces $\tr$ with only two threads and two transactions must use $\Omega(\numEvents)$ space.
\end{restatable}

%% file: sections/increasing_cycle.tex
%!TEX root = ../main.tex
\section{Detecting Increasing-Cycle Violations}\label{sec:increasing_cycle}

In this section we turn our attention to increasing-cycle violations. 
Towards this, we adapt the conditions of \cref{lem:open_thb_condition} to work for the increasing-cycle setting.
This will allow us to tune $\ouralgo$ to report increasing-cycle violations by making slight changes to \cref{algo:cs}.

\Paragraph{Latest remote begins and earliest local events.}
For a trace $\tr$ and two threads $t_1$, $t_2$, we define the \emph{latest remote begin} of $t_2$ known to $t_1$ $\LatestRemoteBegin{\tr}{t_1}(t_2)$, as follows.
\begin{compactitem}
\item If $t_1$ is not executed in $\tr$, then $\LatestRemoteBegin{\tr}{t_1}(t_2)=\bot$. 
Otherwise, let $e_1$ be the latest event of $t_1$ in $\tr$.
\item If $t_2$ has a transaction begin event $\bg_2$ such that  $\bg_2\chb{\tr}e_1$, then $\LatestRemoteBegin{\tr}{t_1}(t_2)$ is the latest such begin event $\bg_2$, otherwise $\LatestRemoteBegin{\tr}{t_1}(t_2)=\bot$.
\end{compactitem}

\input{figures/lrb_ele_example}

If $\LatestRemoteBegin{\tr}{t_1}(t_2)\neq \bot$, we define the \emph{earliest local event} in $t_1$ for $t_2$ as $\EarliestLocalEvent{\tr}{t_1}(t_2)=e'_1$, where $e'_1$ is the earliest event in $t_1$ such that $\LatestRemoteBegin{\tr}{t_1}(t_2)\chb{\tr}e'_1$.

\Paragraph{Example.}
Consider the trace $\tr_8$ in \cref{fig:lrb_ele_example}.
At event $e_7$, we have $\LatestRemoteTransaction{\tr_8[:e_7]}{t_3}(t_1) = T_1$, while $\LatestRemoteBegin{\tr_8[:e_7]}{t_3}(t_1) = \bot$.
At event $e_{9}$, $\LatestRemoteBegin{\tr_8[:e_{9}]}{t_3}(t_1) = T_1$, and $\EarliestLocalEvent{\tr_8[:e_{9}]}{t_3}(t_1) = e_9$.

The following lemma is analogous to \cref{lem:open_thb_condition}, and states the principle of operation of $\ouralgo$ for detecting increasing-cycle violations.

\begin{restatable}{lemma}{lemincreasingpathcondition}\label{lem:increasing_cycle_condition}
The trace $\tr$ has an increasing-cycle violation iff there exist two distinct events $e_1$ and $e_2$ with $\threadof{e_1}=t_1$ and $\threadof{e_2}=t_2$ such that the following conditions hold:
\begin{enumerate*}[label*=(\roman*)]
\item $t_1\neq t_2$ and $e_1$ is a $\chb{\tr}$-immediate predecessor of $e_2$,
\item $\LatestRemoteBegin{\tr[{:}e]}{t_1}(t_2)$ is the $\bg$ event of $\txof(e_2)$, and
\item $\EarliestLocalEvent{\tr[{:}e]}{t_1}(t_2) \threadOrdR{\tr} e_1$,
\end{enumerate*}
where $e$ is a $\traceOrd{\tr}$-immediate predecessor of $e_2$.
\end{restatable}

\input{material/algorithms/increasing_cycle}

\Paragraph{AtomSanitizer for increasing-cycle violations.}
We modify $\ouralgo$ to detect increasing-cycle violations by maintaining in the data structure $\ds$ the latest remote begins $\LatestRemoteBegin{\tr[{:}e]}{t_1}(t_2)$ and earliest local events $\EarliestLocalEvent{\tr[{:}e]}{t_1}(t_2)$, when processing event $e$ of $\tr$, as opposed to  latest remote transactions $\LatestRemoteTransaction{\tr[{:}e]}{t_1}(t_2)$ and earliest local transactions $\EarliestReceiverTransaction{\tr[{:}e]}{t_1}(t_2)$, respectively.
We achieve this by replacing \cref{algo:cs} with \cref{algo:increasing-cycle}.
Moreover, now the vector clocks $\XClock$ and $\YClock$ store event ids (see \cref{line:algo_increasing_cycle_bg_id1} and \cref{line:ds-inc-interface-ins-3}) instead of the transaction ids.
The overall $\ouralgo$ framework remains the same, in that \cref{algo:increasing-cycle} also implements \FunctionCheckAndUpdate and \FunctionInsertEdge, only \FunctionInsertEdge is now somewhat simpler.

\SubParagraph{Function \FunctionCheckAndUpdate{$v$, $e$}.}
This function is similar to the corresponding one in \cref{algo:cs}, with the difference that $n_1$ in \cref{line:algo_increasing_cycle_n1}  represents the event $v$ conflicting with $e$, as opposed to the transaction of $v$.

\SubParagraph{Function \FunctionInsertEdge{$u,v$}.}
This function is similar to the corresponding one in \cref{algo:cs}, and is responsible for registering a new ordering $v\chb{\tr[{:}e]}e$ captured in \FunctionCheckAndUpdate (\cref{line:algo_increasing_cycle_insert_edge}).
The main difference is that the transitive closure only requires backward propagation (\cref{line:ds-inc-interface-ins-1}), as opposed to both forward and backward done in \cref{algo:cs}.
This is because $e$ cannot have any $\chb{\tr[{:}e]}$ successors.
In turn, this reduces the running time for inserting an edge to $O(\numThreads)$, as opposed to $O(\numThreads^2)$ in the case of \cref{algo:cs}.

\input{figures/inc_algo_example}

\Paragraph{Example.}
For brevity, we write $\XClock_{i}(j) = e$ and $\YClock_i(j) = e$ to mean that they store $\FunctionId(e)$.
Consider the trace $\tr_9$ in Figure \ref{fig:inc_algo_example}.
First, executing $e_4$ updates $\ds$ to $\XClock_{t_2}(t_1) = e_1$ and $\YClock_{t_2}(t_1) = e_4$, capturing that $e_1 \chb{\tr_9[:e_4]} e_4$.
Second, executing $e_9$, updates $\ds$ to $\XClock_{t_3}(t_2) = e_8$ and $\YClock_{t_3}(t_2) = e_9$, capturing that $e_8 \chb{\tr_9[:e_9]} e_9$.
In addition, $\ds$ also updates $\XClock_{t_3}(t_1) = e_1$ and $\YClock_{t_3}(t_1) = e_9$, capturing $e_1 \chb{\tr_9[:e_9]} e_9$, by computing the transitive closure.
Finally, when executing $e_{12}$, we have $\XClock_{t_3}(t_1) = e_1$ and $\YClock_{t_3}(t_1) = e_9$.
Since $e_9 \threadOrdR{\tr_9} e_{10}$, $\ouralgo$ reports an increasing-cycle violation.
%\end{example}

\Paragraph{Correctness and complexity.}
The correctness of $\ouralgo$ for increasing-cycle violations follows from \cref{lem:increasing_cycle_condition}, while its time complexity is similar to that for conflict-serializability (\cref{lem:atomsanitizer_complexity}), taking into consideration that each call to \FunctionInsertEdge now takes $O(\numThreads)$ amortized time, as opposed to $O(\numThreads^2)$.
Formally, we have the following theorem.

\begin{restatable}{theorem}{thmatomsanitizerincreasingpath}\label{thm:atomsanitizer_increasing_cycle}
Let $\tr$ be an input trace of $\numEvents$ events, $\numThreads$ threads, $\numLocations$ variables and $\numLocks$ locks.
$\ouralgo$ (increasing-cycle) reports a violation iff $\tr$ has an increasing-cycle violation , and uses $O(\numEvents \numThreads)$ time and $O(\numThreads(\numThreads+\numLocations) + \numLocks)$ space.
\end{restatable}

\Paragraph{Comparison to existing work.}
Standard algorithms for the conflict-happens-before order, and thus increasing-cycle violations, operate in $O(\numEvents\numThreads)$ time and $O(\numThreads(\numThreads+\numLocations+\numLocks))$ space~\cite{Ma-2021,Mathur-2022}.
$\ouralgo$ has the same time complexity, but uses a factor $\numThreads$ less per lock in space.
Moreover, as already stated in the case of conflict-serializability violations in \cref{subsec:atomsanitizer},
$\ouralgo$ only uses $O(1)$ space for each pair $(t,x)$ such that thread $t$ accesses variable $x$, which in practice is less than storing a vector clock per variable, in contrast to~\cite{Ma-2021,Mathur-2022}.

%% file: figures/lrb_ele_example.tex
\begin{wrapfigure}[9]{r}{0.40\textwidth}
\centering
 \vspace{-0.6cm}
% \begin{subfigure}[b]{0.45\textwidth}
\begin{tikzpicture}[xscale=1.5, yscale=0.35, 
every node/.style={font=\small},
transaction/.style={draw, dashed, thick},
event/.style={anchor=center, inner sep=1pt},
-, >={Triangle[width=10pt,length=5pt]}]

% starting coordinates
\def\xstart{0} % starting x coordinate
\def\ystart{0} % starting y coordinate

% width and height of the rectangle
\def\width{3}  % width of the rectangle
\def\height{12} % height of the rectangle

% write labels for the three threads
\node[anchor=south] at (\xstart + \width/6, \ystart + \height) {$t_1$};
\node[anchor=south] at (\xstart + \width/2, \ystart + \height) {$t_2$};
\node[anchor=south] at (\xstart + 5*\width/6, \ystart + \height) {$t_3$};

% Draw the three transaction lines
\draw[dashed] (\xstart + \width/3,\ystart) -- (\xstart + \width/3,\ystart + \height);
\draw[dashed] (\xstart + 2*\width/3, \ystart) -- (\xstart + 2*\width/3, \ystart + \height);

% Draw the grid
\foreach \y in {1,...,11} {
\pgfmathtruncatemacro\yy{12-\y} % Compute reverse ID: 12 -> 1
\node[left] at (\xstart,\ystart+\y) {\yy};  % Keep the coordinates, but reverse the label
}

% events in t1
\node[event] (e1) at (\xstart + \width/6, \ystart + \height - 1) {\bbg{T_1}};   
\node[event] (e2) at (\xstart + \width/6, \ystart + \height - 2) {$\wt(v)$};
\node[event] (e11) at (\xstart + \width/6, \ystart + \height - 11) {$\rd(z)$};

% events in t2
\node[event] (e3) at (\xstart + \width/2, \ystart + \height - 3) {\bbg{T_2}};
\node[event] (e4) at (\xstart + \width/2, \ystart + \height - 4) {$\wt(x)$};
\node[event] (e7) at (\xstart + \width/2, \ystart + \height - 7) {$\rd(v)$};
\node[event] (e8) at (\xstart + \width/2, \ystart + \height - 8) {$\wt(y)$};

% events in t3
\node[event] (e5) at (\xstart + 5*\width/6, \ystart + \height - 5) {\bbg{T_3}};
\node[event] (e6) at (\xstart + 5*\width/6, \ystart + \height - 6) {$\rd(x)$};
\node[event] (e9) at (\xstart + 5*\width/6, \ystart + \height - 9) {$\rd(y)$};
\node[event] (e10) at (\xstart + 5*\width/6, \ystart + \height - 10) {$\wt(z)$};

% e2 -> e7
\draw[->] (e2.east) -- (e7.west);
% e4 -> e6
\draw[->] (e4.east) -- (e6.west);
% e8 -> e9
\draw[->] (e8.east) -- (e9.west);
% e10 -> e11
\draw[->] (e10.west) to (e11.east);

% Draw border
\draw[thick] (\xstart, \ystart) rectangle ++(\width, \height);

% \draw[thick] (0.5,0.5) rectangle (3.5,12.5);
\end{tikzpicture}
\caption{
A trace $\tr_8$.
% $\LatestRemoteTransaction{\tr_6}{t_1}(t_2) = T_1$ via the transaction $T_3$ of $t_3$, but this fact is not recoverable from $\LatestRemoteTransaction{\tr_6}{t_3}(t_2)$, which has progressed to $T_2$.
}
\label{fig:lrb_ele_example}
\Description{}
% \end{subfigure}
\end{wrapfigure}

%% file: material/algorithms/increasing_cycle.tex
%!TEX root = ../../main.tex
%\setlength{\textfloatsep}{12pt}

\begin{algorithm*}[ht!]
	\small
	\DontPrintSemicolon
	\SetInd{0.28em}{0.35em}
	\BlankLine
	\vspace*{-\multicolsep}
	\begin{multicols}{2}
		\MyFunction{\FunctionCheckAndUpdate{$v$, $e$}} {
				\uIf{$\tid(v)\neq \tid(e)$}{
				$n_1 \gets \pair{\tid(v)}{\FunctionId(v)}$\\ \label{line:algo_increasing_cycle_n1}
				$n_2 \gets \pair{\tid(e)}{\FunctionId(\txof(e).\bg)}$\\ \label{line:algo_increasing_cycle_bg_id1}
				\If{$\ds.\reachable(n_2, n_1)$} {
					\label{line:inc-check-reachability}
					$\reportError{Increasing-cycle violation}$
				}
				$\insedge(v, e)$\label{line:algo_increasing_cycle_insert_edge}
			}
			}    
			\BlankLine

			\MyFunction{\FunctionInsertEdge{$u, v$}}{
				$t_1 \gets \tid(u)$\\
				$n_v \gets \pair{\tid(v)}{\FunctionId(v)}$\\
				\ForEach{$t \in \threads{\sigma}$}{\label{line:ds-inc-interface-ins-1}
					\lIf{$t = t_1$}{
						$j \xleftarrow[]{}  \FunctionId(\txof(u).\bg)$\label{line:ds-inc-interface-ins-3}
					}
					\lElse{
						$j \xleftarrow[]{} \ds.\pred(\pair{t_1}{\FunctionId(u)}, t)$\label{line:ds-inc-interface-ins-4}
					} 
				
					$\ds.\addsuc(\pair{t}{j}, n_v)$\label{line:ds-inc-interface-ins-8}
				}
			}
			\BlankLine
						
		\end{multicols}
		\BlankLine
		
		\caption{
			Checking increasing-cycle violations with $\ouralgo$.
		}
		\label{algo:increasing-cycle}
		\vspace*{-\multicolsep}
	\end{algorithm*}
	\normalsize

%% file: figures/inc_algo_example.tex
\begin{wrapfigure}[14]{r}{0.50\textwidth}
    \vspace{-11mm}
%\begin{subfigure}[htbp]{0.50\textwidth}
    \centering
    \begin{tikzpicture}[xscale=1.6, yscale=0.35, 
        every node/.style={font=\small},
        transaction/.style={draw, dashed, thick},
        event/.style={anchor=center, inner sep=1pt},
        -, >={Triangle[width=10pt,length=5pt]}]
    
        % starting coordinates
        \def\xstart{0} % starting x coordinate
        \def\ystart{0} % starting y coordinate
        
        % width and height of the rectangle
        \def\width{3}  % width of the rectangle
        \def\height{15} % height of the rectangle
    
        % write labels for the three threads
        \node[anchor=south] at (\xstart + \width/6, \ystart + \height) {$t_1$};
        \node[anchor=south] at (\xstart + \width/2, \ystart + \height) {$t_2$};
        \node[anchor=south] at (\xstart + 5*\width/6, \ystart + \height) {$t_3$};
        
        % Draw the three transaction lines
        \draw[dashed] (\xstart + \width/3,\ystart) -- (\xstart + \width/3,\ystart + \height);
        \draw[dashed] (\xstart + 2*\width/3, \ystart) -- (\xstart + 2*\width/3, \ystart + \height);
        
        % Draw the grid
        \foreach \y in {1,...,14} {
            \pgfmathtruncatemacro\yy{15-\y} % Compute reverse ID: 14 -> 1
            \node[left] at (\xstart,\ystart+\y) {\yy};  % Keep the coordinates, but reverse the label
        }
        
        % events in t1
        \node[event] (e1) at (\xstart + \width/6, \ystart + \height - 1) {\bbg{T_1}};   
        \node[event] (e2) at (\xstart + \width/6, \ystart + \height - 2) {$\wt(x)$};
        \node[event] (e12) at (\xstart + \width/6, \ystart + \height - 12) {$\rd(z)$};
        \node[event] (e13) at (\xstart + \width/6, \ystart + \height - 13) {\ben{T_1}};
        
        % events in t2
        \node[event] (e3) at (\xstart + \width/2, \ystart + \height - 3) {\bbg{T_2}};
        \node[event] (e4) at (\xstart + \width/2, \ystart + \height - 4) {$\rd(x)$};
        \node[event] (e5) at (\xstart + \width/2, \ystart + \height - 5) {\ben{T_2}};

        \node[event] (e7) at (\xstart + \width/2, \ystart + \height - 7) {\bbg{T_3}};
        \node[event] (e8) at (\xstart + \width/2, \ystart + \height - 8) {$\wt(y)$};
        \node[event] (e14) at (\xstart + \width/2, \ystart + \height - 14) {\ben{T_3}};
        
        % events in t3
        \node[event] (e6) at (\xstart + 5*\width/6, \ystart + \height - 6) {\bbg{T_4}};
        \node[event] (e9) at (\xstart + 5*\width/6, \ystart + \height - 9) {$\rd(y)$};
        \node[event] (e10) at (\xstart + 5*\width/6, \ystart + \height - 10) {$\wt(z)$};
        \node[event] (e11) at (\xstart + 5*\width/6, \ystart + \height - 11) {\ben{T_4}};
        
        % e2 -> e4
        \draw[po] (e2.east) to node[above, circled]{1} (e4.west);
        % e8 -> e9
        \draw[po] (e8.east) to node[above, circled]{2} (e9.west);
        % e10 -> e12
        \draw[vio, bend left=0] (e10.west) to node[above, circled]{3} (e12.east);

        % e1 -> e9
        %\draw[thb,dash pattern=on 5pt off 5pt, bend left=6] (e1.east) to node[right, circled]{3} (e9.west);
        
        % Draw border
        \draw[thick] (\xstart, \ystart) rectangle ++(\width, \height);
    
        % \draw[thick] (0.5,0.5) rectangle (3.5,12.5);
    \end{tikzpicture}
    \caption{
    A trace $\tr_9$ with an increasing-cycle violation.
    Numbers indicate the order of $\chb{\tr_9}$ edges inferred by $\ouralgo$.
    Edge 3 leads to a violation report.
    }
    \label{fig:inc_algo_example}
    \Description{This figure is an example case of $\ouralgo$ finding increasing cycle.}    
    %\end{subfigure}
\end{wrapfigure}
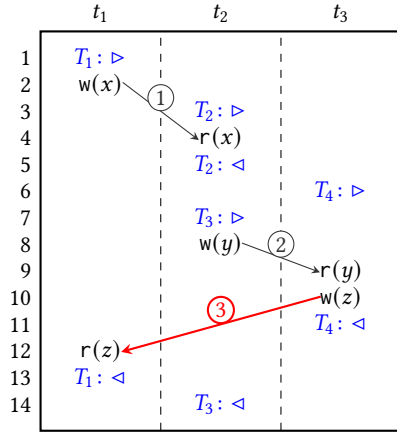

%% file: sections/online.tex
\input{figures/online_locking}

\section{Concurrent Runtime Setting}\label{sec:online}

We now turn our attention to the online setting, where atomicity must be monitored during program execution.
This poses the challenge of sharing the data structure $\ds$  among threads, and accessing/updating it concurrently.
To avoid data races (which can lead to wrong results), such accesses must be appropriately locked.
Existing atomicity checkers address this challenge by forcing each thread to acquire a global lock at each step, which can cause considerable congestion.
Instead, the data structure $\ds$ of $\ouralgo$ optimizes for the case when the transaction of the thread executing the current step is not known to any other thread, which is by far the most common case.
In such cases, updates on $\ds$ can be performed concurrently via a combination of read and write locks, while completely avoiding the use of a global lock, and therefore reducing congestion.
\cref{fig:online_locking} provides an illustration.

\Paragraph{Online $\ouralgo$.}
\cref{algo:ds-online} shows the pseudocode of a wrapper that allows concurrent accesses to $\ds$.
For the online setting, $\ouralgo$ consists of the same components (\cref{algo:bookkeping}, \cref{algo:cs}, \cref{algo:ds}), but
calls to $\ds.\FunctionQuery$ and $\ds.\FunctionInsertEdge$ from \cref{algo:cs} are replaced by calls to the same function in \cref{algo:ds-online}.
Conceptually, $\ds$ is shared between threads, so that $\XClock_{t}$ and $\YClock_{t}$ belong to thread $t$.
Concurrent accesses to $\ds$ are protected by
\begin{enumerate*}[label=(\roman*)]
\item a pair of read/write locks $\elock_t$ for each thread $t$, and
\item a global read/write lock $\elock$.
\end{enumerate*}
We now describe the functions of \cref{algo:ds-online} in more detail.

\SubParagraph{Function $\ds.\FunctionQuery$.}
This function simply wraps the reachability call to $\ds$ around read-acquiring the lock $\elock_{t_2}$ of the target thread $t_2$.
The function is called from thread $t_1$, and no other thread ever write-acquires $\elock_{t_1}$, thus $t_1$ does not need to read-acquire its own lock.

\SubParagraph{Function $\FunctionInsertEdge$.}
This function operates as follows.
First, it calls $\FunctionInsertEdgeBackwards$, which attempts to insert the edge while making only backwards propagation (\cref{line:ds-online-insedgeback}). 
If it fails, $\FunctionInsertEdge$ performs the full transitive closure, by acquiring the global lock $\elock$, as well as all locks $\{ \elock_{t_i} \}_{i}$ (\cref{line:ds-online-elock-1,line:ds-online-local-lock}), to avoid racing with any other thread updating $\ds$.
Finally, the variable $\ltxn_t$ stores the id of the latest transaction $T$ that thread $t$ has executed, and it is appropriately initialized every time $t$ starts a new transaction (not explicitly shown in the pseudocode).
The flag $\oeflag_t$ marks whether $\ltxn_t$ has outgoing edges, helping to avoid locking the whole $\ds$ in future updates.
It is set to True when inserting an edge $(v,e)$ where $v$ belongs to the latest transaction of its thread, and $T$ is not unary (\cref{line:ds-online-oeflag-1}) and is set to False when $T$ is closed (not explicitly shown in the pseudocode).

\SubParagraph{Function $\FunctionInsertEdgeBackwards$.}
Finally, this function attempts to perform a backwards edge propagation from the thread $t_2=\threadof{e}$.
If successful, this only modifies the clocks $\XClock_{t_2}$ and $\YClock_{t_2}$ owned by $t_2$.
In high level, it first read-acquires the global lock $\elock$ (\cref{line:ds-online-elock-2}) to avoid racing with a thread that might be performing a full transitive closure computation, thereby also modifying the clocks  $\XClock_{t_2}$ and $\YClock_{t_2}$.
To perform the backwards propagation, $t_2$ also needs to read the clocks $\XClock_{t_1}$ and $\YClock_{t_1}$.
For this, it read-acquires the lock $\elock_{t_1}$.
This acquisition is done in a loop ( \cref{line:ds-online-while-lock}) to avoid deadlocks occurring when $t_1$ simultaneously read-acquires $\elock_{t_2}$.
After $t_2$ has read-acquired $\elock_{t_1}$, it again checks whether $\txof(e)$ has an outgoing edge.
Note that such an edge might have been added between the call to $\FunctionInsertEdgeBackwards$ made by $\FunctionInsertEdge$ (\cref{line:ds-online-insedgeback}) and when $\elock$ was read-acquired in $\FunctionInsertEdgeBackwards$ (\cref{line:ds-online-elock-2}).
If $\txof(e)$ indeed has an outgoing edge at this point, $\FunctionInsertEdgeBackwards$ fails and returns False to $\FunctionInsertEdge$, forcing the latter to perform a full transitive closure (\cref{line:ds-online-transitive}).
Otherwise, $\FunctionInsertEdgeBackwards$ proceeds to perform the backward propagation (\cref{line:ds-online-ins-1}).
Finally, lines \ref{line:ds-online-ltxn-2}–\ref{line:ds-online-oeflag-2} update the flag $\oeflag_{t_1}$ to indicate whether $\ltxn_{t_1}$ has outgoing edges, similarly to lines \ref{line:ds-online-ltxn-1}–\ref{line:ds-online-oeflag-1} in $\FunctionInsertEdge$.

\input{material/algorithms/online}

\Paragraph{Correctness.}
We now outline the correctness for \cref{algo:ds-online}, by focusing on three key aspects.

\SubParagraph{Deadlock-freeness.}
First, note that the algorithm cannot deadlock:~$\FunctionReachable$ does not have nested locking, whereas in all other cases, acquiring a thread lock of the form $\elock_{t_i}$ (in lines \ref{line:ds-online-local-lock1}, \ref{line:ds-online-llock-2}, \ref{line:ds-online-lwlock}) is preceded by accessing the global lock $\elock$ (in lines \ref{line:ds-online-elock-1}, \ref{line:ds-online-elock-2}).
Note, however, that the hot path is in lines \ref{line:ds-online-while-lock} to \ref{line:ds-online-lunlock} in $\FunctionInsertEdgeBackwards$, where $\elock$ is read-acquired, thus allows the concurrent execution of multiple threads, avoiding congestion.
If multiple readers hold $\elock$, deadlocks are avoided due to the use of $\trywritelock(\elock_{t_2})$ at line \ref{line:ds-online-lwlock}.
This can cause a livelock, which in practice resolves quickly.

\SubParagraph{Mutual exclusion.}
Second, the algorithm guarantees mutual exclusion:~every time a full transitive closure is started by a thread, the data structure is protected by $\elock$ (line \ref{line:ds-online-elock-1}) so that no other thread can modify it.
Moreover, the locks $\elock_t$ of all threads $t$ have been read-acquired (line \ref{line:ds-online-local-lock1}), which also avoids write-read races.
All other data structure updates are such that thread $t_2$ only writes on the $\XClock_{t_2}$ and $\YClock_{t_2}$ clocks at line \ref{line:ds-online-ins-5}, which avoids write-write races.
Moreover, write-read races are avoided because whenever thread $t_2$ accesses the vector clocks $\XClock_{t_1}$ and $\YClock_{t_1}$ of some other thread $t_1$ (line \ref{line:ds-online-ins-4}), it has first read-locked $\elock_{t_1}$ (line \ref{line:ds-online-llock-2}) and write-locked $\elock_{t_2}$ (line \ref{line:ds-online-lwlock}).

\SubParagraph{A benign race condition.}
Finally, notice a subtlety with regards to the accesses of the flag $\oeflag_{t_1}$.
Here, a race condition may cause some thread $t$ executing lines \ref{line:ds-online-oeflag-1} and \ref{line:ds-online-oeflag-2} to set $\oeflag_{t_1}$, even if thread $t_1$ has progressed with a new transaction in the meantime. 
This does not threaten correctness, but may only cause $t_1$ to execute line \ref{line:ds-online-checkoe} later and thus may cause congestion. 
In practice, however, this occurs rarely to warrant the overhead of avoiding this race condition via locking.

%% file: figures/online_locking.tex
\begin{figure}[!tbp]
\centering
\def\scaleboxvalue{0.95}
\begin{subfigure}[b]{0.45\textwidth}
\centering
\scalebox{\scaleboxvalue} {
\begin{tikzpicture}[xscale=1.15, yscale=0.35, 
every node/.style={font=\small},
transaction/.style={draw, dashed, thick},
event/.style={anchor=center, inner sep=1pt},
-, >={Triangle[width=10pt,length=5pt]}]

% starting coordinates
\def\xstart{0} % starting x coordinate
\def\ystart{0} % starting y coordinate

% width and height of the rectangle
\def\width{4}  % width of the rectangle
\def\height{12} % height of the rectangle

% write labels for the three threads
\node[anchor=south] at (\xstart + \width/8, \ystart + \height) {$t_1$};
\node[anchor=south] at (\xstart + 3*\width/8, \ystart + \height) {$t_2$};
\node[anchor=south] at (\xstart + 5*\width/8, \ystart + \height) {$t_3$};
\node[anchor=south] at (\xstart + 7*\width/8, \ystart + \height) {$t_4$};

% Draw the three transaction lines
\draw[dashed] (\xstart + \width/4,\ystart) -- (\xstart + \width/4,\ystart + \height);
\draw[dashed] (\xstart + \width/2, \ystart) -- (\xstart + \width/2, \ystart + \height);
\draw[dashed] (\xstart + 3*\width/4,\ystart) -- (\xstart + 3*\width/4,\ystart + \height);

% Draw the grid
\foreach \y in {1,...,11} {
  \pgfmathtruncatemacro\yy{12-\y} % Compute reverse ID: 11 -> 1
  \node[left] at (\xstart,\ystart+\y) {\yy};  % Keep the coordinates, but reverse the label
}

% events in t1
\node[event] (e1) at (\xstart + \width/8, \ystart + \height - 1) {\bbg{T_1}};   
\node[event] (e2) at (\xstart + \width/8, \ystart + \height - 2) {$\wt(x)$};

% events in t2
\node[event] (e3) at (\xstart + 3*\width/8, \ystart + \height - 3) {\bbg{T_2}};
\node[event] (e4) at (\xstart + 3*\width/8, \ystart + \height - 4) {$\wt(y)$};
\node[event] (e5) at (\xstart + 3*\width/8, \ystart + \height - 5) {$\rd(x)$};
\node[event] (e7) at (\xstart + 3*\width/8, \ystart + \height - 7) {$\wt(z)$};

% events in t3
\node[event] (e6) at (\xstart + 5*\width/8, \ystart + \height - 6) {\bbg{T_3}};
\node[event] (e7f) at (\xstart + 5*\width/8, \ystart + \height - 7) {$\wt(a)$};
\node[event] (e11) at (\xstart + 5*\width/8, \ystart + \height - 11) {$\rd(y)$};

% events in t4
\node[event] (e8) at (\xstart + 7*\width/8, \ystart + \height - 8) {\bbg{T_4}};
\node[event] (e9) at (\xstart + 7*\width/8, \ystart + \height - 9) {$\wt(b)$};
\node[event] (e10) at (\xstart + 7*\width/8, \ystart + \height - 10) {$\rd(z)$};

% e2 -> e5
\draw[->] (e2.east) -- (e5.west);
% e4 -> e11
\draw[->] (e4.east) -- (e11.west);
% e7 -> e10
\draw[->] (e7.east) -- (e10.west);

% Draw border
\draw[thick] (\xstart, \ystart) rectangle ++(\width, \height);

% \draw[thick] (0.5,0.5) rectangle (3.5,12.5);
\end{tikzpicture}
}
\caption{A trace $\tr_{10}$ where $T_3\not \thb{\tr_{10}}T_4$.}
\label{subfig:online_lock_example1}
\end{subfigure}
\hfill
\begin{subfigure}[b]{0.45\textwidth}
\centering
\scalebox{\scaleboxvalue} {
\begin{tikzpicture}[xscale=1.15, yscale=0.35, 
every node/.style={font=\small},
transaction/.style={draw, dashed, thick},
event/.style={anchor=center, inner sep=1pt},
-, >={Triangle[width=10pt,length=5pt]}]

% starting coordinates
\def\xstart{0} % starting x coordinate
\def\ystart{0} % starting y coordinate

% width and height of the rectangle
\def\width{4}  % width of the rectangle
\def\height{12} % height of the rectangle

% write labels for the three threads
\node[anchor=south] at (\xstart + \width/8, \ystart + \height) {$t_1$};
\node[anchor=south] at (\xstart + 3*\width/8, \ystart + \height) {$t_2$};
\node[anchor=south] at (\xstart + 5*\width/8, \ystart + \height) {$t_3$};
\node[anchor=south] at (\xstart + 7*\width/8, \ystart + \height) {$t_4$};

% Draw the three transaction lines
\draw[dashed] (\xstart + \width/4,\ystart) -- (\xstart + \width/4,\ystart + \height);
\draw[dashed] (\xstart + \width/2, \ystart) -- (\xstart + \width/2, \ystart + \height);
\draw[dashed] (\xstart + 3*\width/4,\ystart) -- (\xstart + 3*\width/4,\ystart + \height);

% Draw the grid
\foreach \y in {1,...,11} {
\pgfmathtruncatemacro\yy{12-\y} % Compute reverse ID: 11 -> 1
\node[left] at (\xstart,\ystart+\y) {\yy};  % Keep the coordinates, but reverse the label
}

% events in t1
\node[event] (e1) at (\xstart + \width/8, \ystart + \height - 1) {\bbg{T_1}};   
\node[event] (e2) at (\xstart + \width/8, \ystart + \height - 2) {$\wt(x)$};

% events in t2
\node[event] (e3) at (\xstart + 3*\width/8, \ystart + \height - 3) {\bbg{T_2}};
\node[event] (e4) at (\xstart + 3*\width/8, \ystart + \height - 4) {$\wt(y)$};
\node[event] (e5) at (\xstart + 3*\width/8, \ystart + \height - 5) {$\rd(x)$};
\node[event] (e7) at (\xstart + 3*\width/8, \ystart + \height - 7) {$\wt(z)$};

% events in t3
\node[event] (e6) at (\xstart + 5*\width/8, \ystart + \height - 6) {\bbg{T_3}};
\node[event] (e7f) at (\xstart + 5*\width/8, \ystart + \height - 7) {$\wt(a)$};
\node[event] (e11) at (\xstart + 5*\width/8, \ystart + \height - 11) {$\rd(y)$};

% events in t4
\node[event] (e8) at (\xstart + 7*\width/8, \ystart + \height - 8) {\bbg{T_4}};
\node[event] (e9) at (\xstart + 7*\width/8, \ystart + \height - 9) {$\wt(a)$};
\node[event] (e10) at (\xstart + 7*\width/8, \ystart + \height - 10) {$\rd(z)$};

% e2 -> e5
\draw[->] (e2.east) -- (e5.west);
% e4 -> e11
\draw[->] (e4.east) -- (e11.west);
% e7 -> e10
\draw[->] (e7.east) -- (e10.west);
% e6 -> e8
\draw[red, ->] (e7f.east) -- (e9.west);

% Draw border
\draw[thick] (\xstart, \ystart) rectangle ++(\width, \height);

% \draw[thick] (0.5,0.5) rectangle (3.5,12.5);
\end{tikzpicture}
}
\caption{A trace $\tr_{11}$ where $T_3 \thb{\tr_{11}}T_4$.}
\label{subfig:online_lock_example2}
\end{subfigure}
\caption{\label{fig:online_locking}
In $\tr_{10}$ (\subref{subfig:online_lock_example1}), the transactions $T_3$ and $T_4$ do not have outgoing $\thb{\tr_{10}}$-orderings, and thus \cref{algo:ds-online} can process the events $e_{10}$ and $e_{11}$ concurrently, without acquiring a global lock, even though it has to perform a (backwards) transitive-closure to make $t_3$ and $t_4$ aware of $T_1$.
In contrast, in $\tr_{11}$ (\subref{subfig:online_lock_example2}), $T_3$ has an outgoing ordering (marked in red), making \cref{algo:ds-online} acquire a global lock when processing $e_{11}$.
}
\Description{This figure is an example case of increasing cycle.}
\end{figure}

%% file: material/algorithms/online.tex
\begin{algorithm}[htbp]
\small
\DontPrintSemicolon
\SetInd{0.3em}{0.3em}
\begin{multicols}{2}
\BlankLine
\BlankLine
\MyFunction{\FunctionQuery{$\graphnode{t_1}{ j_1}, \graphnode{t_2}{j_2}$}}{\label{line:ds-online-query-1}
% \greytcp{\texttt{If we never insert any non-increasing paths, then this is guaranteed to be non-blocking. Because no thread can be holding a write lock on $l_{t_2}$.}}
% $\readlock(\elock)$\\
$\readlock(\elock_{t_2})$ \label{line:ds-online-llock-1} \\
$r\gets \ds.\FunctionQuery{ $\graphnode{t_1}{ j_1}$, $\graphnode{t_2}{j_2}$ }$\\
%\lIf{$\pred(\graphnode{ t_2}{ j_2}, t_1) \geq j_1$} { \label{line:ds-online-query-2}
%   $r \xleftarrow[]{}$ True
%}
%\lElse{
%   $r \xleftarrow[]{}$ False
%}
$\readunlock(\elock_{t_2})$\\
% $\readunlock(\elock)$ \\
\Return r
}

\BlankLine
\BlankLine
\MyFunction{\FunctionInsertEdge{$v, e$}}{
$t_1 \gets \tid(v)$\\
\uIf{$\insedgeback(v, e)$} { \label{line:ds-online-insedgeback}
\Return
}

\greytcp{\texttt{executes rarely in practice}}
$T \gets \aload(\ltxn_{t_1})$\\ \label{line:ds-online-ltxn-1}
\uIf{$\unary{v} = \text{False} \land T = \txof(v)$} {
$\astore(\oeflag_{t_1}, \text{True})$ \label{line:ds-online-oeflag-1}
}
$\writelock(\elock)$\\ \label{line:ds-online-elock-1}
\ForEach{$t \in \threads{\tr}$} { \label{line:ds-online-local-lock}
$\readlock(\elock_t)$ \label{line:ds-online-local-lock1}
}
full transitive closure as in \cref{algo:cs}\\ \label{line:ds-online-transitive}
\ForEach{$t \in \threads{\tr}$} {
$\readunlock(\elock_t)$
}
$\writeunlock(\elock)$\\
}
\BlankLine
\BlankLine

% \greytcp{\texttt{We are holding the write lock on $l_{t_2}$, so no other thread can observe a stale value. We are holding the read lock on $l_{t_1}$, so the reachability information that involve $t_1$ cannot change.}}
\MyFunction{\FunctionInsertEdgeBackwards{$v, e$}} {
$t_1 \gets \tid(v)$ ; $t_2 \gets \tid(e)$\\
$n_e \gets \graphnode{t_2}{\txid{\txof(e)}}$\\
$\readlock(\elock)$\\ \label{line:ds-online-elock-2}
\While{True} {\label{line:ds-online-while-lock}
$\readlock(\elock_{t_1})$ \label{line:ds-online-llock-2}\\
\lIf{$\trywritelock(\elock_{t_2})$} {\Break} \label{line:ds-online-lwlock}
\lElse{$\readunlock(\elock_{t_1})$} \label{line:ds-online-lunlock}
}
\uIf{$\aload(\oeflag_{t_2}) = \text{True}$} { \label{line:ds-online-checkoe}
\greytcp{\texttt{executes rarely in practice}}
$\writeunlock(\elock_{t_2})$;
$\readunlock(\elock_{t_1})$;
$\readunlock(\elock)$;\\
\Return False
}
\ForEach{$t \in \threads{\tr}$} { \label{line:ds-online-ins-for}
\label{line:ds-online-ins-1}
\lIf{$t = t_1$}{
\label{line:ds-online-ins-3}
$j \gets \txid{\txof(v)}$
}
\lElse{
$j \gets \ds.\pred(\pair{t_1}{\txid{\txof(v)}}, t)$ \label{line:ds-online-ins-4}
} 
\label{line:ds-online-ins-8}
$\ds$.$\addsuc(\pair{t}{j}, n_e)$\\ \label{line:ds-online-ins-5}
}
$\writeunlock(\elock_{t_2})$;
$\readunlock(\elock_{t_1})$;
$\readunlock(\elock)$\\
$T \gets \aload(\ltxn_{t_1})$\\ \label{line:ds-online-ltxn-2}
\uIf{$\unary{v} = \text{False} \land T = \txof(v)$} {
$\astore(\oeflag_{t_1}, \text{True})$ \label{line:ds-online-oeflag-2}
}
\Return True
}
\end{multicols}
\caption{
Concurrent wrapper around $\ds$. 
}
\label{algo:ds-online}
\end{algorithm}

%% file: sections/experiments-main.tex
%!TEX root = ../main.tex

\section{Experimental Evaluation}\label{sec:experiments}

In this section, we report on an implementation of $\ouralgo$, both as a standalone atomicity checker, and as a runtime monitor inside $\TSAN$.
We also report  on an experimental evaluation of $\ouralgo$ in each setting, and compare it to other atomicity checkers.
All experiments are run on an Ubuntu 22.04 machine with 2.4GHz CPU and 64GB of memory.

\input{sections/experiments-offline}
\input{sections/experiments-online}

%% file: sections/experiments-offline.tex
\subsection{Standalone Experiments}

\Paragraph{Implementation.}
We implement $\ouralgo$ inside the $\RAPID$ framework for dynamic analyses~\cite{Umang-rapid} in Java, closely following the pseudocode in \cref{sec:atomicity_testing} (for conflict-serializability violations) and \cref{sec:increasing_cycle} (for increasing-cycle violations).

\Paragraph{Baselines.}
We compare the performance of $\ouralgo$ to the standard and SOTA atomicity checkers 
(i)~$\velodrome$~\cite{Flanagan-2008},
(ii)~$\aerodrome$~\cite{Mathur-2020}, and
(iii)~$\regiontrack$~\cite{Ma-2021}.
For this, we port $\regiontrack$ into $\RAPID$, which already implements $\velodrome$ and $\aerodrome$.
Each tool processes the same input trace and reports whether it is conflict-serializable.
We also compare $\ouralgo$ and $\regiontrack$ for increasing cycles, since they are the only tools that report increasing cycles in a sound and complete manner.
By default, $\regiontrack$ performs conflict-serializability and increasing-cycle checking simultaneously.
For a fair comparison, we have separated the two components, which reduces its running time for each individual task.

\Paragraph{Benchmarks.}
Our benchmark traces are derived from the DeCaPo benchmark suite~\cite{Blackburn-2006}, Java Grande suite~\cite{Smith-2001}, and microbenchmarks~\cite{vonPraun-2003}, taken from the experimental setup of $\aerodrome$~\cite{Mathur-2020} and also used in other works~\cite{Biswas-2014,Ma-2021}.
Benchmarks ending in \texttt{``\_ext''} are our own variants of existing benchmarks that scale up the number of threads $\numThreads$.
We report the average running time of each tool over 10 runs, imposing a timeout of 2 hours in each run.

\input{tables/tab_offline}
\Paragraph{Results.}
The results are shown in \cref{tab:offline}, for both conflict-serializability and increasing-cycle violations.
For conflict-serializability, we observe that $\ouralgo$ is faster than each of the baselines, on each benchmark.
The geometric mean of the speedup across all benchmarks is $2.0\times$, $4.5\times$, and $12\times$ over $\regiontrack$, $\aerodrome$, and $\velodrome$, respectively.
Overall, $\ouralgo$ and $\regiontrack$ stand out as the faster algorithms compared to $\aerodrome$ and $\velodrome$, with $\ouralgo$ being decisively the fastest.
The comparison on increasing-cycle violations is similar, with $\ouralgo$ being faster than $\regiontrack$ on each benchmark.
The geometric mean of the speedup over all benchmarks is $1.7\times$.

%% file: tables/tab_offline.tex
\begin{table}[tbp]
\centering
\setlength\tabcolsep{1.3pt}
\def\angle{50}
\caption{
Offline experiments on conflict-serializability and increasing-cycle violations.
Entries marked with $^*$ witness a violation (both conflict-serializability and increasing-cycle).
%TO denotes a timeout after 2 hours.
}
\pgfplotstableread[col sep=comma]{data/CS_offline.csv}\datatable
\small
\setlength\tabcolsep{1.4pt}

\pgfplotstabletypeset[%
columns={Benchmarks,numEvents,numThreads,numLocations,atomsanitizer(ms),regiontrackconflict(ms),aerodrome(ms),velodrome(ms),atomsanitizer-increasing(ms),regiontrack-increasing(ms)
},
assign column name/.style={/pgfplots/table/column name={\textbf{#1}}},
every head row/.style={before row=\toprule, after row=\midrule, font=bold},
every last row/.style={after row=\bottomrule},
columns/{Benchmarks}/.style={column type={l|},string type, column name=,
postproc cell content/.prefix code={%
            \pgfkeyssetvalue{/pgfplots/table/@cell content}%
            {\ttfamily ##1}%
        },
},
columns/{numEvents}/.style={column type={c|},string type, column name=,
postproc cell content/.prefix code={%
        \count0=\pgfplotstablerow
        \advance\count0 by1
        \ifnum\count0=\pgfplotstablerows
            \pgfkeysalso{@cell content=\textbf{##1}}%
        \else
            \ifnum\count0=\numexpr\pgfplotstablerows-1
                % \pgfkeysalso{@cell content={\textbf{##1}}} %
                \pgfkeysalso{@cell content={##1}} %
            \else
                \pgfkeysalso{@cell content={##1}} %
            \fi
        \fi
    },
},
columns/{numThreads}/.style={column type={c|},string type, column name=,
postproc cell content/.prefix code={%
        \count0=\pgfplotstablerow
        \advance\count0 by1
        \ifnum\count0=\pgfplotstablerows
            \pgfkeysalso{@cell content=$\textbf{##1}$}%
        \else
            \ifnum\count0=\numexpr\pgfplotstablerows-1
                % \pgfkeysalso{@cell content={\textbf{##1}}} %
                \pgfkeysalso{@cell content={##1}} %
            \else
                \pgfkeysalso{@cell content={##1}} %
            \fi
        \fi
    },
},
columns/{numLocations}/.style={column type={c|},string type, column name=,
postproc cell content/.prefix code={%
        \count0=\pgfplotstablerow
        \advance\count0 by1
        \ifnum\count0=\pgfplotstablerows
            \pgfkeysalso{@cell content=$\textbf{##1}$}%
        \else
            \ifnum\count0=\numexpr\pgfplotstablerows-1
                % \pgfkeysalso{@cell content={\textbf{##1}}} %
                \pgfkeysalso{@cell content={##1}} %
            \else
                \pgfkeysalso{@cell content={##1}} %
            \fi
        \fi
    },
},
every head row/.append style={before row={
\midrule
\multirow{2}{*}{\textbf{Benchmark}} & \multirow{2}{*}{$\numEvents$} & \multirow{2}{*}{$\numThreads$} & \multirow{2}{*}{$\numLocations$} & \multicolumn{4}{c|}{\textbf{Conflict Serializability (ms)}} & \multicolumn{2}{c}{\textbf{Increasing Cycle (ms)}} \\ \cline{5-10}
}},
columns/{atomsanitizer(ms)}/.style={column type={|r},string type, column name=\begin{turn}{\angle}$\ouralgo$\end{turn},
postproc cell content/.append code={
            \count0=\pgfplotstablerow
            \advance\count0 by1
            \ifnum\count0=\pgfplotstablerows
                \pgfkeysalso{@cell content=$\mathbf{\pgfmathprintnumber[assume math mode=true, fixed, fixed zerofill, precision=0]{##1}}$}%
			\else
                \ifnum\count0=\numexpr\pgfplotstablerows-1
                    % \pgfkeysalso{@cell content=$\mathbf{\pgfmathprintnumber[assume math mode=true, fixed, fixed zerofill, precision=0]{##1}}$}%
                    \pgfkeysalso{@cell content=${\pgfmathprintnumber[assume math mode=true, fixed, fixed zerofill, precision=0]{##1}}$}%
                \else
                    \pgfkeysalso{@cell content=${\pgfmathprintnumber[assume math mode=true, fixed, fixed zerofill, precision=0]{##1}}$}%
                \fi
            \fi
        },
%    },
  },
columns/{regiontrackconflict(ms)}/.style={column type={r},string type, column name=\begin{turn}{\angle}$\regiontrack$\end{turn},
postproc cell content/.append code={
            \count0=\pgfplotstablerow
            \advance\count0 by1
            \ifnum\count0=\pgfplotstablerows
                \pgfkeysalso{@cell content=$\mathbf{\pgfmathprintnumber[assume math mode=true, fixed, fixed zerofill, precision=0]{##1}}$}%
			\else
                \ifnum\count0=\numexpr\pgfplotstablerows-1
                    % \pgfkeysalso{@cell content=$\mathbf{\pgfmathprintnumber[assume math mode=true, fixed, fixed zerofill, precision=0]{##1}}$}%
                    \pgfkeysalso{@cell content=${\pgfmathprintnumber[assume math mode=true, fixed, fixed zerofill, precision=0]{##1}}$}%
                \else
                    \pgfkeysalso{@cell content=${\pgfmathprintnumber[assume math mode=true, fixed, fixed zerofill, precision=0]{##1}}$}%
                \fi
            \fi
        },
%    },
  },
columns/{aerodrome(ms)}/.style={column type={r},string type, column name=\begin{turn}{\angle}$\aerodrome$\end{turn},
postproc cell content/.append code={
            \count0=\pgfplotstablerow
            \advance\count0 by1
            \ifnum\count0=\pgfplotstablerows
                \pgfkeysalso{@cell content=$\mathbf{\pgfmathprintnumber[assume math mode=true, fixed, fixed zerofill, precision=0]{##1}}$}%
			\else
                \ifnum\count0=\numexpr\pgfplotstablerows-1
                    % \pgfkeysalso{@cell content=$\mathbf{\pgfmathprintnumber[assume math mode=true, fixed, fixed zerofill, precision=0]{##1}}$}%
                    \pgfkeysalso{@cell content=${\pgfmathprintnumber[assume math mode=true, fixed, fixed zerofill, precision=0]{##1}}$}%
                \else
                    \pgfkeysalso{@cell content=${\pgfmathprintnumber[assume math mode=true, fixed, fixed zerofill, precision=0]{##1}}$}%
                \fi
            \fi
        },
%    },
},
columns/{velodrome(ms)}/.style={column type={r},string type, column name=\begin{turn}{\angle}$\velodrome$\end{turn},
postproc cell content/.append code={
						\ifthenelse{\equal{\string ##1}{\string TO}}
						 { \pgfkeyssetvalue{/pgfplots/table/@cell content} {##1}}
						{
            \count0=\pgfplotstablerow
            \advance\count0 by1
            \ifnum\count0=\pgfplotstablerows
                \pgfkeysalso{@cell content=$\mathbf{\pgfmathprintnumber[assume math mode=true, fixed, fixed zerofill, precision=0]{##1}}$}%
			\else
                \ifnum\count0=\numexpr\pgfplotstablerows-1
                    % \pgfkeysalso{@cell content=$\mathbf{\pgfmathprintnumber[assume math mode=true, fixed, fixed zerofill, precision=0]{##1}}$}%
                    \pgfkeysalso{@cell content=${\pgfmathprintnumber[assume math mode=true, fixed, fixed zerofill, precision=0]{##1}}$}%
                \else
                    \pgfkeysalso{@cell content=${\pgfmathprintnumber[assume math mode=true, fixed, fixed zerofill, precision=0]{##1}}$}%
                \fi
            \fi
        }
        },
%    },
  },
columns/{atomsanitizer-increasing(ms)}/.style={column type={|r},string type, column name=\begin{turn}{\angle}$\ouralgo$\end{turn},
postproc cell content/.append code={
            \count0=\pgfplotstablerow
            \advance\count0 by1
            \ifnum\count0=\pgfplotstablerows
                \pgfkeysalso{@cell content=$\mathbf{\pgfmathprintnumber[assume math mode=true, fixed, fixed zerofill, precision=0]{##1}}$}%
			\else
                \ifnum\count0=\numexpr\pgfplotstablerows-1
                    % \pgfkeysalso{@cell content=$\mathbf{\pgfmathprintnumber[assume math mode=true, fixed, fixed zerofill, precision=0]{##1}}$}%
                    \pgfkeysalso{@cell content=${\pgfmathprintnumber[assume math mode=true, fixed, fixed zerofill, precision=0]{##1}}$}%
                \else
                    \pgfkeysalso{@cell content=${\pgfmathprintnumber[assume math mode=true, fixed, fixed zerofill, precision=0]{##1}}$}%
                \fi
            \fi
        },
%    },
  },
columns/{regiontrack-increasing(ms)}/.style={column type={r},string type, column name=\begin{turn}{\angle}$\regiontrack$\end{turn},
postproc cell content/.append code={
            \count0=\pgfplotstablerow
            \advance\count0 by1
            \ifnum\count0=\pgfplotstablerows
                \pgfkeysalso{@cell content=$\mathbf{\pgfmathprintnumber[assume math mode=true, fixed, fixed zerofill, precision=0]{##1}}$}%
			\else
                \ifnum\count0=\numexpr\pgfplotstablerows-1
                    % \pgfkeysalso{@cell content=$\mathbf{\pgfmathprintnumber[assume math mode=true, fixed, fixed zerofill, precision=0]{##1}}$}%
                    \pgfkeysalso{@cell content=${\pgfmathprintnumber[assume math mode=true, fixed, fixed zerofill, precision=0]{##1}}$}%
                \else
                    \pgfkeysalso{@cell content=${\pgfmathprintnumber[assume math mode=true, fixed, fixed zerofill, precision=0]{##1}}$}%
                \fi
            \fi
        },
%    },
  },    
every row no 22/.style={
    after row={\midrule},
},
highlight/.append style={
        postproc cell content/.append code={
                \pgfkeysalso{@cell content=\textbf{##1}}%
        },
    },
    highlight last two rows/.style={
        postproc cell content/.append code={
            \count0=\pgfplotstablerow
            \advance\count0 by1
            \ifnum\count0=\pgfplotstablerows
                \pgfkeysalso{@cell content=\textbf{##1}}%
            \fi
            % \ifnum\count0=\numexpr\pgfplotstablerows-1
            %     \pgfkeysalso{@cell content=\textbf{##1}}%
            % \fi
        },
    },
highlight last two rows,
]{\datatable}

\label{tab:offline}

\end{table}

%% file: sections/experiments-online.tex
\subsection{Runtime Experiments}

In this section, our goal is to evaluate the time and memory performance of $\ouralgo$ in a concurrent runtime setting, and compare it to the data race detection engine of $\TSAN$, which is an industry standard in concurrent runtime analyses.

\Paragraph{Implementation.}
We implement $\ouralgo$ inside $\TSAN$~\cite{Serebryany-2009} in C++, following the description in \cref{sec:online}.
We also enable the default optimizations of $\TSAN$, such as thread pooling and storing only the four latest events per memory address.

\Paragraph{Benchmarks.}
Our benchmark set consists of 13 popular open-source projects known for high concurrency, also used in related works (e.g.,~\cite{Tunc-2024}).
Since these programs do not have explicit transactions, we create an atomicity specification for each, following a process used in~\cite{Biswas-2014}, and has been provided the specifications of prior work (e.g., \cite{Mathur-2020,Ma-2021}).
In high level, this process starts with a fine-grained specification which marks each function as atomic.
Then, as long as an atomicity violation is reported, it refines the specification by removing the transactions surrounding the function on which the atomicity violation was reported.
This creates atomicity specifications for which we expect no atomicity violations, while retaining the atomicity requirement of many functions, so that benchmarking remains informative.
The absence of violations is also necessary to meaningfully compare $\ouralgo$ and $\TSAN$, as otherwise an atomicity violation would make $\ouralgo$ halt early, and appear (unfairly) faster.
We report the average running time and memory used by $\ouralgo$ and $\TSAN$ over 20 runs.

\input{tables/online}

\Paragraph{Results.}
Our results are shown in \cref{tab:online}.
As a first observation, we see that the overhead of $\TSAN$ over the uninstrumented program is $3.1\times$ on (geometric) average, and $13.8\times$ in the worst case (\texttt{ImageMagic}), in terms of time.
This matches the expectations of the developers for an overhead $\leq 15\times$ in time~\cite{Serebryany-2009}.
In terms of memory, $\TSAN$ has an overhead of $3.5\times$ on average, and $14.5\times$ in the worst case (\texttt{pigz}).
Although the average case is well within the $\leq 10\times$ bound expected by the developers, the worst case exceeds it.
This, however, occurs on a benchmark with a very small memory footprint (2MB), and thus is insignificant in practice.

Compared to $\TSAN$, $\ouralgo$ incurs a reasonable time overhead, averaging $1.49\times$ over all benchmarks.
The largest overhead occurs on \texttt{x264} and \texttt{7z}, reaching $2.21\times$, and $3.78\times$, respectively.
For the two benchmarks, we observe that $\ouralgo$ performs a fully transitive closure more times, which is the main reason for the slow down, as it requires heavier locking (see \cref{sec:online}).
Finally, we observe that the memory footprint of $\ouralgo$ is nearly identical to $\TSAN$, with the geometric mean of the overhead being only $1.06\times$.
Compared to the original, uninstrumented executions, the geometric mean of the overheads over time and memory are $4.63\times$ and $3.83\times$, respectively.
The time overhead is always $\leq 15\times$, except for one case that goes up to $\simeq 20\times$ (\texttt{7z}).
The memory overhead is always $\leq 6\times$,  except for one case that goes up to $\simeq 21 \times$ (\texttt{pigz}).
Note, however, that $\TSAN$ also suffers a high memory overhead on this benchmark ($\simeq 14 \times$).
Overall, our results support the hypothesis that $\ouralgo$ can be incorporated in a runtime monitoring setting, incurring acceptable overheads in both time and memory.

%% file: tables/online.tex
\begin{table*}[!t]
\caption{
Online experiments on conflict-serializability ($\ouralgo$) and race detection ($\TSAN$) inside the $\TSAN$ framework.
``Origin'' refers to the original, uninstrumented benchmark.
``\ovhtsan'' and ``\ovhorig'' denote the ratios of $\ouralgo$ over $\TSAN$ and $\orig$ respectively.
``Total/Max'' refers to time/memory measurements, respectively.
``Total Overhead'' is the ratio of the total running times; ``Max Overhead'' is the ratio of the maximum memory usages.
}
\def\angle{47}
\pgfplotstableread[col sep=comma]{data/TSan.csv}\datatable
\centering
\small
\setlength\tabcolsep{2.9pt}
\pgfplotstabletypeset[%
columns={Benchmarks,origin(s),atomsanitizer(s),tsan(s),atomsanitizer/origin(time),atomsanitizer/tsan(time),origin(kb),atomsanitizer(kb),tsan(kb),atomsanitizer/origin(memory),atomsanitizer/tsan(memory)},
assign column name/.style={/pgfplots/table/column name={\textbf{#1}}},
every head row/.style={before row=\toprule, after row=\midrule, font=bold},
every last row/.style={after row=\bottomrule},
columns/{Benchmarks}/.style={column type={l|},string type, column name=,
postproc cell content/.prefix code={%
            \pgfkeyssetvalue{/pgfplots/table/@cell content}%
            {\ttfamily ##1}%
        },
},
every head row/.append style={before row={
\midrule
\multirow{2}{*}{\textbf{Benchmark}} & \multicolumn{5}{c|}{\textbf{Time (s)}} & \multicolumn{5}{c}{\textbf{Memory (MB)}} \\ \cline{2-11}
}},
columns/{origin(s)}/.style={column type={|r},string type, column name=\begin{turn}{\angle}$\orig$\end{turn},
postproc cell content/.append code={
            \count0=\pgfplotstablerow
            \advance\count0 by1
            \ifnum\count0=\pgfplotstablerows
                \pgfkeysalso{@cell content=$\textbf{-}$}%
			\else
                \ifnum\count0=\numexpr\pgfplotstablerows-1
                    \pgfkeysalso{@cell content=$\mathbf{\pgfmathprintnumber[assume math mode=true, fixed, fixed zerofill, precision=0]{##1}}$}%
                \else
                    \pgfkeysalso{@cell content=${\pgfmathprintnumber[assume math mode=true, fixed, fixed zerofill, precision=0]{##1}}$}%
                \fi
            \fi
        },
%    },
  },
columns/{atomsanitizer(s)}/.style={column type={r},string type, column name=\begin{turn}{\angle}$\ouralgo$\end{turn},
postproc cell content/.append code={
            \count0=\pgfplotstablerow
            \advance\count0 by1
            \ifnum\count0=\pgfplotstablerows
                \pgfkeysalso{@cell content=$\textbf{-}$}%
			\else
                \ifnum\count0=\numexpr\pgfplotstablerows-1
                    \pgfkeysalso{@cell content=$\mathbf{\pgfmathprintnumber[assume math mode=true, fixed, fixed zerofill, precision=0]{##1}}$}%
                \else
                    \pgfkeysalso{@cell content=${\pgfmathprintnumber[assume math mode=true, fixed, fixed zerofill, precision=0]{##1}}$}%
                \fi
            \fi
        },
%    },
  },
columns/{tsan(s)}/.style={column type={r},string type, column name=\begin{turn}{\angle}$\TSAN$\end{turn},
postproc cell content/.append code={
            \count0=\pgfplotstablerow
            \advance\count0 by1
            \ifnum\count0=\pgfplotstablerows
                \pgfkeysalso{@cell content=$\textbf{-}$}%
			\else
                \ifnum\count0=\numexpr\pgfplotstablerows-1
                    \pgfkeysalso{@cell content=$\mathbf{\pgfmathprintnumber[assume math mode=true, fixed, fixed zerofill, precision=0]{##1}}$}%
                \else
                    \pgfkeysalso{@cell content=${\pgfmathprintnumber[assume math mode=true, fixed, fixed zerofill, precision=0]{##1}}$}%
                \fi
            \fi
        },
%    },
  },
columns/{atomsanitizer/tsan(time)}/.style={column type={r},string type, column name=\begin{turn}{\angle}\ovhtsan\end{turn},
postproc cell content/.append code={
            \count0=\pgfplotstablerow
            \advance\count0 by1
            \ifnum\count0=\pgfplotstablerows
                \pgfkeysalso{@cell content=$\mathbf{\pgfmathprintnumber[assume math mode=true, fixed, fixed zerofill, precision=2]{##1}}$}%
			\else
                \ifnum\count0=\numexpr\pgfplotstablerows-1
                    \pgfkeysalso{@cell content=$\mathbf{\pgfmathprintnumber[assume math mode=true, fixed, fixed zerofill, precision=2]{##1}}$}%
                \else
                    \pgfkeysalso{@cell content=${\pgfmathprintnumber[assume math mode=true, fixed, fixed zerofill, precision=2]{##1}}$}%
                \fi
            \fi
        },
%    },
  },
columns/{atomsanitizer/origin(time)}/.style={column type={r},string type, column name=\begin{turn}{\angle}\ovhorig\end{turn},
postproc cell content/.append code={
            \count0=\pgfplotstablerow
            \advance\count0 by1
            \ifnum\count0=\pgfplotstablerows
                \pgfkeysalso{@cell content=$\mathbf{\pgfmathprintnumber[assume math mode=true, fixed, fixed zerofill, precision=2]{##1}}$}%
			\else
                \ifnum\count0=\numexpr\pgfplotstablerows-1
                    \pgfkeysalso{@cell content=$\mathbf{\pgfmathprintnumber[assume math mode=true, fixed, fixed zerofill, precision=2]{##1}}$}%
                \else
                    \pgfkeysalso{@cell content=${\pgfmathprintnumber[assume math mode=true, fixed, fixed zerofill, precision=2]{##1}}$}%
                \fi
            \fi
        },
%    },
  },
columns/{origin(kb)}/.style={column type={|r},string type, column name=\begin{turn}{\angle}$\orig$\end{turn},
postproc cell content/.append code={
            \count0=\pgfplotstablerow
            \advance\count0 by1
            \ifnum\count0=\pgfplotstablerows
            \pgfkeysalso{@cell content=$\textbf{-}$}%
			\else
                \ifnum\count0=\numexpr\pgfplotstablerows-1
                    \pgfkeysalso{@cell content={{\pgfkeys{/pgf/fpu=true}\pgfmathparse{##1/1024}$\mathbf{\pgfmathprintnumber[assume math mode=true, fixed, fixed zerofill, precision=0]{\pgfmathresult}}$}}} %
                \else
                    \pgfkeysalso{@cell content={{\pgfkeys{/pgf/fpu=true}\pgfmathparse{##1/1024}${\pgfmathprintnumber[assume math mode=true, fixed, fixed zerofill, precision=0]{\pgfmathresult}}$}}} %
                \fi
            \fi
},
},
columns/{atomsanitizer(kb)}/.style={column type={r},string type, column name=\begin{turn}{\angle}$\ouralgo$\end{turn},
postproc cell content/.append code={
            \count0=\pgfplotstablerow
            \advance\count0 by1
            \ifnum\count0=\pgfplotstablerows
                \pgfkeysalso{@cell content=$\textbf{-}$}%
			\else
                \ifnum\count0=\numexpr\pgfplotstablerows-1
                    \pgfkeysalso{@cell content={{\pgfkeys{/pgf/fpu=true}\pgfmathparse{##1/1024}$\mathbf{\pgfmathprintnumber[assume math mode=true, fixed, fixed zerofill, precision=0]{\pgfmathresult}}$}}} %
                \else
                    \pgfkeysalso{@cell content={{\pgfkeys{/pgf/fpu=true}\pgfmathparse{##1/1024}${\pgfmathprintnumber[assume math mode=true, fixed, fixed zerofill, precision=0]{\pgfmathresult}}$}}} %
                \fi
            \fi
},
},
columns/{tsan(kb)}/.style={column type={r},string type, column name=\begin{turn}{\angle}$\TSAN$\end{turn},
postproc cell content/.append code={
            \count0=\pgfplotstablerow
            \advance\count0 by1
            \ifnum\count0=\pgfplotstablerows
            \pgfkeysalso{@cell content=$\textbf{-}$}%
			\else
                \ifnum\count0=\numexpr\pgfplotstablerows-1
                    \pgfkeysalso{@cell content={{\pgfkeys{/pgf/fpu=true}\pgfmathparse{##1/1024}$\mathbf{\pgfmathprintnumber[assume math mode=true, fixed, fixed zerofill, precision=0]{\pgfmathresult}}$}}} %
                \else
                    \pgfkeysalso{@cell content={{\pgfkeys{/pgf/fpu=true}\pgfmathparse{##1/1024}${\pgfmathprintnumber[assume math mode=true, fixed, fixed zerofill, precision=0]{\pgfmathresult}}$}}} %
                \fi
            \fi
},
},
columns/{atomsanitizer/tsan(memory)}/.style={column type={r},string type, column name=\begin{turn}{\angle}\ovhtsan\end{turn},
postproc cell content/.append code={
            \count0=\pgfplotstablerow
            \advance\count0 by1
			\pgfmathparse{##1}
            \ifnum\count0=\pgfplotstablerows
                \pgfkeysalso{@cell content=$\mathbf{\pgfmathprintnumber[assume math mode=true, fixed, fixed zerofill, precision=2]{##1}}$}%
			\else
                \ifnum\count0=\numexpr\pgfplotstablerows-1
                    \pgfkeysalso{@cell content=$\mathbf{\pgfmathprintnumber[assume math mode=true, fixed, fixed zerofill, precision=2]{##1}}$}%
                \else
                    \pgfkeysalso{@cell content=${\pgfmathprintnumber[assume math mode=true, fixed, fixed zerofill, precision=2]{##1}}$}%
                \fi
            \fi
        },
  },
columns/{atomsanitizer/origin(memory)}/.style={column type={r},string type, column name=\begin{turn}{\angle}\ovhorig\end{turn},
postproc cell content/.append code={
            \count0=\pgfplotstablerow
            \advance\count0 by1
			\pgfmathparse{##1}
            \ifnum\count0=\pgfplotstablerows
                \pgfkeysalso{@cell content=$\mathbf{\pgfmathprintnumber[assume math mode=true, fixed, fixed zerofill, precision=2]{##1}}$}%
			\else
                \ifnum\count0=\numexpr\pgfplotstablerows-1
                    \pgfkeysalso{@cell content=$\mathbf{\pgfmathprintnumber[assume math mode=true, fixed, fixed zerofill, precision=2]{##1}}$}%
                \else
                    \pgfkeysalso{@cell content=${\pgfmathprintnumber[assume math mode=true, fixed, fixed zerofill, precision=2]{##1}}$}%
                \fi
            \fi
        },
  },
every row no 12/.style={
    after row={\midrule},
},
highlight/.append style={
        postproc cell content/.append code={
                \pgfkeysalso{@cell content=\textbf{##1}}%
        },
    },
    highlight last row/.style={
        postproc cell content/.append code={
            \count0=\pgfplotstablerow
            \advance\count0 by1
            \ifnum\count0=\pgfplotstablerows
                \pgfkeysalso{@cell content=\textbf{##1}}%
            \fi
            \ifnum\count0=\numexpr\pgfplotstablerows-1
                \pgfkeysalso{@cell content=\textbf{##1}}%
            \fi
        },
    },
highlight last row,
]{\datatable}
\label{tab:online}
\end{table*}

%% file: sections/related_work-main.tex
%!TEX root = ../main.tex
\section{Related Work}\label{sec:related_work}

Detecting atomicity vulnerabilities automatically has been an active research topic and targeted by a range of methods.
Atomizer~\cite{Flanagan-2004} was a precursor of Velodrome, based on Lipton's reduction theory~\cite{Lipton-75}.
Techniques like AtomFuzzer~\cite{Park-2008} and PENELOPE~\cite{Sorrentino-2010} aim at generating schedules that  expose atomicity violations, while others focus on synthesizing tests for the same task~\cite{Samak-2015}.
Algorithms like $\ouralgo$ can speed up such approaches, by efficiently testing whether the generated execution indeed constitutes an atomicity violation.
Since atomicity invariants are not always explicit, several approaches aim at inferring them automatically~\cite{Lu-2006,Park-2009}.
Finally, atomicity has also been targeted by static analyses, most prominently using types and effects~\cite{Flanagan2003,Sasturkar2005,Flanagan2008}.

View serializability is another atomicity notion, which is more general than conflict-serializability.
Intuitively, the serial witness trace $\tr^*$ can now be obtained from the observed trace $\tr$ by swapping conflicting events, as long every read observes the same write in $\tr$ and $\tr^*$.
However, testing view serializability is much harder: it is NP-complete in general~\cite{Papadimitriou1979a}, and even hard when parameterized by the number of threads~\cite{Mathur-2020c}.

Dynamic analyses are a standard approach to concurrency bugs.
Besides atomicity violations, they have been targeting 
data races~\cite{Flanagan-2009,Pavlogiannis-2020,Gorjiara-2022,Lidbury-2017},
deadlocks~\cite{Tunc-2023},
memory vulnerabilities~\cite{Bond-2006}, and
linearizability violations~\cite{Gibbons2002,Emmi2017,Abdulla-25,Lee-25}.
Their efficiency has been an active pursuit in both theory and practice~\cite{Mathur-2022,Tunc-2024,Kulkarni-2021a,Roemer-2020,Bond-2013}.

%% file: sections/conclusion-main.tex
%!TEX root = ../main.tex
\section{Conclusion}
We have introduced $\ouralgo$,  a new algorithm for dynamically testing for atomicity violations.
$\ouralgo$ is theoretically faster than all previous algorithms for the problem, while it also has small space complexity.
Our experimental results show that the theoretical advantage of $\ouralgo$ is also realized in practice, while its lightweight time and memory footprint make it suitable for a runtime monitoring setting.
Interesting future directions include transferring the insights of this work towards monitoring transactional consistency in databases, against a range of popular isolation levels~\cite{Biswas2019,Clark2024,Moldrup2025}.

%% file: sections/data-availability-statement.tex
\section*{Data-Availability Statement}

Data for our experiments and the experimental setup can be found in the accompanying artifact~\cite{hunkar-2026}.
All proofs and additional details of experiments can be found in the technical report~\cite{hunkar-2026a}.

%% file: sections/acknowledgments.tex
\begin{acks}
This work was partially supported by a research grant (VIL42117) from VILLUM FONDEN, and by a research grant from STIBOFONDEN.
\end{acks}

%% file: sections/app_serializability.tex
\section{Proofs of  \cref{sec:atomicity_testing}}\label{sec:app_serializability}

\lemopenthbcondition*
\begin{proof}
First, assume that such events $e'$ and $e''$ exist.
Due to (i), we have $\txof(e') \thb{\tr} \txof(e'')$.
Due to (ii) and (iii), we have $\txof(e'') \thb{\tr[{:}e]} \txof(e')$, and by \cref{rem:thb_monotonic}, we have $\txof(e'') \thb{\tr} \txof(e')$.
Thus $\thb{\tr}$ is cyclic, and by \cref{lem:thb_acyclicity_condition}, $\tr$ is not conflict-serializable.

For the opposite direction, assume that $\tr$ is not conflict-serializable, and by \cref{lem:thb_acyclicity_condition}, we have that $\thb{\tr}$ is cyclic.
Due to \cref{rem:thb_monotonic}, there is a first event $e''$ of $\tr$ such that $\thb{\tr[{:}e'']}$ is cyclic.
Let $e$ be the $\traceOrd{\tr}$-immediate predecessor of $e''$, and note that 
\[
\chb{\tr[{:}e'']}=\left(\chb{\tr[{:}e]}\cup \left\{(e', e'')\colon e'\chb{\tr[{:}e'']}e''\right\}\right)
\]
which is equivalent to
\begin{align*}
\chb{\tr[{:}e'']}&=\left(\chb{\tr[{:}e]}\cup \{(e', e'')\colon \right.
 \left. \vphantom{\chb{\tr[{:}e]}\cup \{(e', e'')\colon} e'\text{ is a } \chb{\tr} \text{-immediate predecessor of } e''\}\right)^+
\end{align*}
where $P^+$ denotes the transitive closure of relation $P$.
The latter equation implies that
\begin{align*}
\thb{\tr[{:}e'']}&=\left(\thb{\tr[{:}e]}\cup \{(\txof(e'), \txof(e''))\colon \right.
 \left. \vphantom{\chb{\tr[{:}e]}\cup \{(\txof(e'), \txof(e''))\colon} e'\text{ is a } \chb{\tr} \text{-immediate predecessor of } e''\}\right)^+
\end{align*}
Since $\thb{\tr[{:}e'']}$ is cyclic but $\thb{\tr[{:}e]}$ is not, we must also have $\txof(e'')\thb{\tr[{:}e'']}\txof(e')$, where $e'$ is a $\chb{\tr}$-immediate predecessor of $e''$.
Since the set of successors of $\txof(e'')$ is the same in $\tr[{:}e]$ and $\tr[{:}e'']$ (i.e., for every transaction $T$, we have $\txof(e'')\thb{\tr[{:}e'']}T$ iff $\txof(e'')\thb{\tr[{:}e]}T$), we must also have $\txof(e'')\thb{\tr[{:}e]}\txof(e')$.

The desired result follows.
\end{proof}

\lematomsanitizerinvariant*
\begin{proof}
It is straightforward to see that whenever $\ouralgo$ calls $\insedge(n_1, n_2)$ in \cref{line:algo_csc_insert_edge} $n_1$ and $n_2$ represent $\txof(v)$ and $\txof(e)$, respectively, where $e$ is the event currently being processed and $v$ is such that $v\chb{\tr}e$.
Moreover, it will make such a call at least for every $v$ that is a $\chb{\tr}$-immediate predecessor of $e$.
The invariant then follows by a straightforward induction, observing that all the necessary transitive $\thb{\tr[{:}e]}$ orderings required to update $\LatestRemoteTransaction{\tr[{:}e]}{t_1}(t_2)$ and $\EarliestReceiverTransaction{\tr[{:}e]}{t_1}(t_2)$, for all pairs of threads $t_1 $ and $t_2$, are captured in the two loops of \cref{line:ds-cs-interface-ins-1} and \cref{line:ds-cs-interface-ins-2} in \cref{algo:cs}.

The desired result follows.
\end{proof}

\lematomsanitizercorrectness*
\begin{proof}
It is straightforward to see that whenever $\ouralgo$ (\cref{algo:bookkeping}) calls \FunctionCheckAndUpdate($v, e$), we have that $v\chb{\tr} e$, and it will make such a call for every event $v$ that is a $\chb{\tr}$-immediate predecessor of $e$.
Then, a conflict-serializability violation is reported iff $\ds.\reachable(n_2, n_1)=$~True, which, due to \cref{lem:atomsanitizer_invariant}, holds iff $\LatestRemoteTransaction{\tr[{:}e']}{\threadof{v}}(\threadof{e})=\txof(e)$, and
$\EarliestReceiverTransaction{\tr[{:}e']}{\threadof{v}}(\threadof{e}).\bg \threadOrd{\tr} \txof(e).\bg$,
where $e'$ is the $\traceOrd{\tr}$-immediate predecessor of $e$.
Thus, \cref{lem:open_thb_condition} applies, leading to the desired result.
\end{proof}

\lematomsanitizercomplexity*
\begin{proof}
We begin with an analysis of the time complexity of $\ouralgo$.
First, observe that the initialization functions of \cref{algo:bookkeping} and \cref{algo:ds} take time $O(\numLocations + \numLocks)$ and $O(\numThreads^2)$, respectively.
Every other function in $\ds$ (\cref{algo:ds}) takes trivially $O(1)$ time.
In \cref{algo:cs}, every call to \FunctionInsertEdge{$\pair{t_1}{j_1}, \pair{t_2}{j_2}$} takes $O(\numThreads^2)$ time, which is also the time taken in \FunctionCheckAndUpdate{$v$, $e$} (since it calls \FunctionInsertEdge).
Finally, each handler in \cref{algo:bookkeping} except \FunctionWrite{$e$} takes $O(\numThreads^2)$ time, which is determined by a single call to \FunctionCheckAndUpdate.
On the other hand, \FunctionWrite{$e$} makes $O(\numThreads)$ calls to \FunctionCheckAndUpdate, one for each thread stored in $\lastReads{x}$.
Thus, each call to \FunctionWrite{$e$} takes $O(\numThreads^3)$ time.
Note, however, that this function resets $\lastReads{x}$ (\cref{line:algo_atomsanitizer_write_empty_lrds}).
Thus, between any two successive writes on $x$, this function will iterate over a specific thread $t$ (\cref{line:algo_atomsanitizer_loop_over_lrds}) only if $t$ has executed some read on $x$ between them, so as to insert $t$ back into $\lastReads{x}$ (\cref{line:algo_atomsanitizer_update_lrds}).
Thus, we can charge each call to \FunctionCheckAndUpdate (\cref{line:algo_atomsanitizer_write_chb_reads}) to the corresponding read of thread $t$, leading to $O(\numEvents\numThreads^2)$ total running time.

We now turn our attention to the space complexity of $\ouralgo$.
Each vector clock $\XClock_{t}$ and $\YClock_{t}$ takes $O(\numThreads)$ space, thus the total space of $\ds$ is $O(\numThreads^2)$.
$\ouralgo$ additionally stores $O(\numThreads)$ events in $\lrd{x}$ and  $\lastReads{x}$ for each location $x$, and one event in $\lastWrite{x}$, thus using $O(\numThreads\numLocations)$ across all locations $x$.
Finally, it stores one entry in $\lastRelease{\ell}$, for each lock $\ell$, thus using $O(\numLocks)$ space across all locks $\ell$.
Summing up, the space usage of $\ouralgo$ is $O(\numThreads(\numThreads + \numLocations) + \numLocks)$.

The desired result follows.
\end{proof}

\thmatomsanitizer*
\begin{proof}
The theorem follows directly from \cref{lem:atomsanitizer_correctness} and \cref{lem:atomsanitizer_complexity}.
\end{proof}

\lemtimelowerbound*
\begin{proof}
We establish a reduction from the Online Vector-Matrix-Vector Multiplication (OuMv) problem~\cite{Henzinger2015} to conflict-serializability.
In OuMv, we are given an $m\times m$ boolean matrix $M$, and an online sequence of vector pairs $(u_1, v_1),\dots, (u_m, v_m)$.
The task is to output each product $(u_i)^{\intercal}M v_i$ before we see $(u_{i+1}, v_{i+1})$.
The OMv hypothesis implies that, if we allow only polynomial-time preprocessing on $M$, the online computation of vector-matrix-vector products takes $\Omega(m^{3-\epsilon})$ time, for any $\epsilon >0$~\cite{Henzinger2015}.

We cast OuMv as a conflict-serializability problem on a trace $\tr$, as follows.
\begin{compactenum}
\item The matrix $M$ is simulated by $2m$ threads $\thread^1_1, \dots, \thread^1_m$ and $\thread^2_1, \dots, \thread^2_m$,
where thread $\thread_i^j$ executes a transaction $T_i^j$.
Each $T_i^1$ starts with a write $\wt(x_i)$, while each $T_i^2$ starts with a sequence of reads $\rd(x_j)$ for all $j$ such that $M[i,j]=1$.
Note that this implies that $T^1_j\chb{\tr}T^2_i$.
These events execute in the beginning of $\tr$, so that all writes execute before all reads.
All transactions $T_i$ remain open until the end of $\tr$.
\item For each pair of vectors $(u_i, v_i)$, we start two threads $\thread^u_i$ and $\thread^v_i$, executing the transactions $T^u_i$ and $T^v_i$, respectively.
$T^u_i$ executes a single write event $\wt(y_i)$, after which each transaction $T^1_j$ executes a read event $\rd(y_i)$ iff $u_i[j]=1$.
Note that this implies $T^u_i\thb{\tr} T^1_j$ iff $u_i[j]=1$.
Afterwards, each transaction $T^2_j$ executes a read event $\rd(z_i)$ iff $v_i[j]=1$, and finally $T^v_i$ executes a single write event $\wt(z_i)$.
Note that this implies $T^2_j\thb{\tr} T^v_i$ $v_i[j]=1$.
In total, the imposed transaction orderings are such that $T^u_i\thb{\tr}T^v_i$ iff $(u_i)^{\intercal}Mv_i=1$.

Finally, $T^v_i$ executes a write event $\wt(s_i)$, followed by $T^u_i$ executing a read event $\rd(s_i)$, and then $T^i_i$ and $T^v_i$ close.
Note that this implies $T^v_i\thb{\tr} T^u_i$.
\end{compactenum}
It follows that after the vector pair $(u_i, v_i)$ is revealed, 
the corresponding prefix of $\tr$ exhibits a conflict-serializability violation iff $(u_i)^{\intercal}Mv_i=1$.
Note that after the first $m^2$ events encoding $M$, each vector pair $(u_i, v_i)$ grows $\tr$ by $O(m)$ events.
Thus, after all $m$ pairs of vectors have been revealed, $\tr$ has size $\numEvents=O(m^2)$.
This implies that any algorithm detecting conflict-seralizability violations on $\tr$ must take $\Omega(m^{3-\epsilon})$ time, for any $\epsilon>0$, and thus $\Omega(\numEvents^{3/2-\epsilon})$ time, for any such $\epsilon$.
\end{proof}

\lemspacelowerbound*
\begin{proof}
Given some set $A\subseteq \{1,\dots, \numEvents\}$ and some integer $\ell\in \{1,\dots, \numEvents\}$, we define the trace $\tr_{A, \ell}$ with two threads $\thread_1$, $\thread_2$ and two transactions.
Let $A=\{i_1,\dots, i_m\}$, and $\tr_{A,\ell}$ consists of the following three segments.
\begin{compactenum}
\item In the first segment, $\thread_1$ opens a transaction $T_1$, and executes in sequence $\wt(x_{i_1}), \wt(x_{i_2}),\dots, \wt(x_{i_m})$.
\item In the second segment, $\thread_2$ opens a transaction $T_2$ and executes $\wt(x_{\ell})$, followed by $\wt(y)$, and ends by closing $T_2$.
\item In the third segment, $\thread_1$ executes $\wt(y)$ and ends by closing $T_1$.
\end{compactenum}
It follows that $\tr_{A,\ell}$ is conflict serializable iff $\ell\not \in A$; if not, it exhibits an increasing-path violation.
Moreover, note that $\tr_{A,\ell}$ has length $O(\numEvents)$.

Now assume that there exists a streaming algorithm for detecting conflict-serializability/increasing-path violations that uses $<\numEvents$ space.
Since there are $2^\numEvents$ subsets of $\{1, \dots, \numEvents\}$, and the state space of the algorithm is $<2^n$, there exist two distinct sets $A$ and $A'$ such that the algorithm is in the same state after processing the first segment of $\tr_{A, \ell}$ and $\tr_{A', \ell}$, for any $\ell \in \{1,\dots, \numEvents\}$, since $\ell$ is revealed in the second segment of the trace.
Assume wlog that $A\not \subseteq A'$, and choose $\ell \in A\setminus A'$.
Observe that $\tr_{A, \ell}$ is not conflict-serializable, while $\tr{A', \ell}$ is.
However, the algorithm executes identically on the two traces, since
(i)~the algorithm is in the same state after processing the first segment of $\tr_{A, \ell}$ and $\tr_{A', \ell}$, and
(ii)~the second the third segments of the two traces coincide.
Hence the algorithm must give the wrong answer when executed on one of the two traces.
The desired result follows.
\end{proof}

%% file: sections/app_increasing_cycle.tex
\section{Proofs of  \cref{sec:increasing_cycle}}\label{sec:app_increasing_cycle}

\lemincreasingpathcondition*
\begin{proof}
First, assume that such events $e_1$ and $e_2$ exist. Due to (ii) and (iii), we have $\bg\chb{\tr[:e]}e_1$, which, together with (i), implies $\bg\chb{\tr[:e]}e_1\chb{\tr[:e]}e_2$, yielding an increasing-cycle violation.

For the opposite direction, assume that $\tr$ has an increasing-cycle violation.
Observe that, for any two trances $\tr_1$ and $\tr_3$ such that $\tr_1$ is a prefix of $\tr_3$, we have $\chb{\tr_1}\subseteq \chb{\tr_3}$.
Thus, there is an earliest event $e_2$ of $\tr$ such that $\bg\chb{\tr[:e_2]}e_1\chb{\tr[:e_2]}e_2$, where $bg$ is the begin event of $\txof(e_2)$ and $\threadof{e_1}\neq \threadof{e_2}$.
Let $e$ be the $\traceOrd{\tr}$-immediate predecessor of $e_2$, and note that 
\[
\chb{\tr[:e_2]}=\left(\chb{\tr[:e]}\cup \left\{(e_1, e_2)\colon e_1\chb{\tr[:e_2]}e_2\right\}\right)
\]
The desired result follows by letting $e_1$ by any $\chb{\tr}$-immediate predecessor of $e_2$ such that $\bg\chb{\tr[:e]}e_1$.
\end{proof}

\thmatomsanitizerincreasingpath*
\begin{proof}
The correctness of $\ouralgo$ (increasing-cycle) follows straightforwardly from \cref{lem:increasing_cycle_condition} and the following invariant, which is analogous to \cref{lem:atomsanitizer_invariant}.
After $\ouralgo$ has processed an event $e$, for all threads $t_1$, $t_2$, we have $\XClock_{t_1}(t_2)=\FunctionId(\LatestRemoteBegin{\tr[:e]}{t_1}(t_2))$ and $\YClock_{t_1}(t_2)=\FunctionId(\EarliestLocalEvent{\tr[:e]}{t_1}(t_2))$.

The time complexity follows by an analysis similar to that of \cref{lem:atomsanitizer_complexity}, accounting for the fact that each call to \FunctionInsertEdge now costs $O(\numThreads)$ as opposed to $O(\numThreads^2)$, since there is no forward propagation of edges.
Finally, the space complexity analysis is identical to that of \cref{lem:atomsanitizer_complexity}.
\end{proof}

%% file: sections/app_experiments.tex
\section{Experimental Details}\label{sec:app_experimental_details}
\cref{tab:benchmark_statistics} provides some additional statistics on the benchmarks we use in \cref{sec:experiments}.
\input{tables/benchmark_statistics}

%% file: tables/benchmark_statistics.tex
\begin{table*}[htbp]
\caption{Benchmark statistics: number of transactions ($\numTransactions$), events ($\numEvents$), threads ($\numThreads$), and memory locations ($\numLocations$).}
\centering
\begin{subtable}[t]{0.48\textwidth}
\centering
\input{tables/online_benchmark}
\end{subtable}
\hfill
\begin{subtable}[t]{0.48\textwidth}
\centering
\input{tables/offline_benchmark}
\end{subtable}
\label{tab:benchmark_statistics}
\end{table*}

%% file: tables/online_benchmark.tex
\caption{Online experiments}
\small
\pgfplotstableread[col sep=comma]{data/online_profile.csv}\datatable
\setlength\tabcolsep{2.9pt}
\pgfplotstabletypeset[%
columns={Benchmarks,numTransactions,numEvents,numThreads,numLocations},
assign column name/.style={/pgfplots/table/column name={\textbf{#1}}},
every head row/.style={before row={\toprule\textbf{Benchmark} & $\numTransactions$ & $\numEvents$ & $\numThreads$ & $\numLocations$ \\\midrule}, output empty row},
every last row/.style={after row=\bottomrule},
skip rows between index={13}{15},
columns/{Benchmarks}/.style={column type={l|},string type, column name=,
postproc cell content/.prefix code={%
            \pgfkeyssetvalue{/pgfplots/table/@cell content}%
            {\ttfamily ##1}%
        },
},
columns/{numTransactions}/.style={column type={c|},string type, column name=,
postproc cell content/.prefix code={%
        \count0=\pgfplotstablerow
        \advance\count0 by1
        \ifnum\count0=\pgfplotstablerows
            \pgfkeysalso{@cell content=\textbf{##1}}%
        \else
            \ifnum\count0=\numexpr\pgfplotstablerows-1
                \pgfkeysalso{@cell content={##1}} %
            \else
                \pgfkeysalso{@cell content={##1}} %
            \fi
        \fi
    },
},
columns/{numEvents}/.style={column type={c|},string type, column name=,
postproc cell content/.prefix code={%
        \count0=\pgfplotstablerow
        \advance\count0 by1
        \ifnum\count0=\pgfplotstablerows
            \pgfkeysalso{@cell content=\textbf{##1}}%
        \else
            \ifnum\count0=\numexpr\pgfplotstablerows-1
                \pgfkeysalso{@cell content={##1}} %
            \else
                \pgfkeysalso{@cell content={##1}} %
            \fi
        \fi
    },
},
columns/{numThreads}/.style={column type={c|},string type, column name=,
postproc cell content/.prefix code={%
        \count0=\pgfplotstablerow
        \advance\count0 by1
        \ifnum\count0=\pgfplotstablerows
            \pgfkeysalso{@cell content=$\textbf{##1}$}%
        \else
            \ifnum\count0=\numexpr\pgfplotstablerows-1
                \pgfkeysalso{@cell content={##1}} %
            \else
                \pgfkeysalso{@cell content={##1}} %
            \fi
        \fi
    },
},
columns/{numLocations}/.style={column type={c},string type, column name=,
postproc cell content/.prefix code={%
        \count0=\pgfplotstablerow
        \advance\count0 by1
        \ifnum\count0=\pgfplotstablerows
            \pgfkeysalso{@cell content=$\textbf{##1}$}%
        \else
            \ifnum\count0=\numexpr\pgfplotstablerows-1
                \pgfkeysalso{@cell content={##1}} %
            \else
                \pgfkeysalso{@cell content={##1}} %
            \fi
        \fi
    },
},
]{\datatable}
\label{tab:online_profile}

%% file: tables/offline_benchmark.tex
\caption{Offline experiments}
\small
\pgfplotstableread[col sep=comma]{data/offline_profile.csv}\datatable
\setlength\tabcolsep{2.9pt}
\pgfplotstabletypeset[%
columns={Benchmarks,numTransactions,numEvents,numThreads,numLocations},
assign column name/.style={/pgfplots/table/column name={\textbf{#1}}},
every head row/.style={before row={\toprule\textbf{Benchmark} & $\numTransactions$ & $\numEvents$ & $\numThreads$ & $\numLocations$ \\\midrule}, output empty row},
every last row/.style={after row=\bottomrule},
columns/{Benchmarks}/.style={column type={l|},string type, column name=,
postproc cell content/.prefix code={%
            \pgfkeyssetvalue{/pgfplots/table/@cell content}%
            {\ttfamily ##1}%
        },
},
columns/{numTransactions}/.style={column type={c|},string type, column name=,
postproc cell content/.prefix code={%
        \count0=\pgfplotstablerow
        \advance\count0 by1
        \ifnum\count0=\pgfplotstablerows
            \pgfkeysalso{@cell content=##1}%
        \else
            \ifnum\count0=\numexpr\pgfplotstablerows-1
                \pgfkeysalso{@cell content={##1}} %
            \else
                \pgfkeysalso{@cell content={##1}} %
            \fi
        \fi
    },
},
columns/{numEvents}/.style={column type={c|},string type, column name=,
postproc cell content/.prefix code={%
        \count0=\pgfplotstablerow
        \advance\count0 by1
        \ifnum\count0=\pgfplotstablerows
            \pgfkeysalso{@cell content=##1}%
        \else
            \ifnum\count0=\numexpr\pgfplotstablerows-1
                \pgfkeysalso{@cell content={##1}} %
            \else
                \pgfkeysalso{@cell content={##1}} %
            \fi
        \fi
    },
},
columns/{numThreads}/.style={column type={c|},string type, column name=,
postproc cell content/.prefix code={%
        \count0=\pgfplotstablerow
        \advance\count0 by1
        \ifnum\count0=\pgfplotstablerows
            \pgfkeysalso{@cell content={##1}}%
        \else
            \ifnum\count0=\numexpr\pgfplotstablerows-1
                \pgfkeysalso{@cell content={##1}} %
            \else
                \pgfkeysalso{@cell content={##1}} %
            \fi
        \fi
    },
},
columns/{numLocations}/.style={column type={c},string type, column name=,
postproc cell content/.prefix code={%
        \count0=\pgfplotstablerow
        \advance\count0 by1
        \ifnum\count0=\pgfplotstablerows
            \pgfkeysalso{@cell content={##1}}%
        \else
            \ifnum\count0=\numexpr\pgfplotstablerows-1
                \pgfkeysalso{@cell content={##1}} %
            \else
                \pgfkeysalso{@cell content={##1}} %
            \fi
        \fi
    },
},
]{\datatable}
\label{tab:offline_profile}

%% file: sections/app_scalability.tex
\section{Controlled Scalability Experiments}\label{sec:app_scalability}

\input{sections/experiments-scalability}

%% file: sections/experiments-scalability.tex
We have conducted some controlled scalability experiments, aiming to clarify the better theoretical complexity of $\ouralgo$ for conflict-serializability.
In particular, recall that the $O(\numEvents\numThreads^2)$ time bound of $\ouralgo$ is better than the time bounds of existing algorithms (\cref{tab:intro-conflict-complexity}).
However, these bounds may be loose, in the sense that the true worst-case complexity of some algorithm might be strictly below its proven bound.
\emph{Does $\ouralgo$ truly achieve  a better worst-case complexity than existing algorithms?}

\input{figures/scalability}

\Paragraph{Benchmarks.}
To address the above question, for each of the baselines of $\velodrome$, $\aerodrome$, and $\regiontrack$, we construct a sequence of traces $(\tr_j)_j$, shown respectively in \cref{fig:sc_velodrome_trace}, \cref{fig:sc_aerodrome_trace} \cref{fig:sc_region_trace}.
Each trace family is constructed so that
\begin{enumerate*}[label=(\roman*)]
\item it is challenging for the respective baseline to process, and
\item $\ouralgo$ suffers its worst case complexity, i.e., $\Theta(\numEvents\numThreads^2)$.
\end{enumerate*}

If the speedup of $\ouralgo$ increases along $(\tr_j)_j$, we obtain strong experimental evidence that the baseline does not have equal (or smaller) complexity compared to $\ouralgo$.

\input{figures/SC_velodrome}
\input{figures/SC_aerodrome}
\input{figures/SC_region}

\Paragraph{Results.}
The results are shown in \cref{fig:scalability}.
We indeed observe that the speedup of $\ouralgo$ over each baseline increases (the precise values of the speedup are not important).
This experimentally confirms that none of the existing tools achieves a complexity of $O(\numEvents\numThreads^2)$, which is the complexity of $\ouralgo$.
We remark that this experiment \emph{does not} conclude that $\ouralgo$ is always faster (which is not the case), but rather that the \emph{worst-case} complexity gap of each baseline compared to the \emph{worst-case} complexity of $\ouralgo$ is super-constant.

%% file: figures/scalability.tex
\begin{figure*}[ht!]
\def\scaleboxvalue{0.72}
\def\plotheight{5.5cm}
\def\plotwidth{6.5cm}
\def\markersize{2.25}

\pgfplotsset{AxisStyle/.style={
inner sep=1pt,
	mark size=\markersize,
	height=\plotheight,
	width=\plotwidth,
	legend style={legend pos=north west, draw=black, legend columns=6, legend cell align={left}, font=\normalsize},
	legend style={/tikz/every even column/.append style={column sep=0.5cm}},
	ylabel={Speedup},
	xlabel style={font=\normalfont},
	ylabel style={font=\normalfont},
	xtick distance=5e3,
	ylabel near ticks,
	xlabel near ticks,
	xticklabel style={font=\normalfont},
	yticklabel style={font=\normalfont},
	grid style=dashed,
	grid=both,}
}
\def\lineWidth{1}

\pgfplotsset{
    discard if not/.style 2 args={
        x filter/.code={
            \edef\tempa{\thisrow{#1}}
            \edef\tempb{#2}
            \ifx\tempa\tempb
            \else
                \def\pgfmathresult{inf}
            \fi
        }
    }
}
\scalebox{\scaleboxvalue}{%
\begin{tikzpicture}[tight background, inner sep=0pt]
\begin{groupplot}[
group style={group size= 3 by 1, horizontal sep=1.75cm} ,
]

\nextgroupplot[AxisStyle, xlabel={$\numEvents$}, title={\large  $\ouralgo$ vs $\aerodrome$}, every x tick scale label/.style={at={(1,0)},xshift=1pt,anchor=south west,inner sep=0pt},]
\addplot[color=RoyalBlue, line width=\lineWidth, mark=*, ] table [x=events_num(n), y=scalability_vs_aerodrome, col sep=comma,] {data/SC_aerodrome.csv};

\nextgroupplot[AxisStyle, xlabel={$\numEvents$}, title={\large  $\ouralgo$ vs $\velodrome$}, xtick distance=1e4, every x tick scale label/.style={at={(1,0)},xshift=1pt,anchor=south west,inner sep=0pt},]
\addplot[color=RoyalBlue, line width=\lineWidth, mark=*, ] table [x=events_num(n), y=scalability_vs_velodrome, col sep=comma,] {data/SC_velodrome.csv};

\nextgroupplot[AxisStyle, xlabel={$\numEvents$}, title={\large  $\ouralgo$ vs $\regiontrack$}, xtick distance=1e4, every x tick scale label/.style={at={(1,0)},xshift=1pt,anchor=south west,inner sep=0pt},]
\addplot[color=RoyalBlue, line width=\lineWidth, mark=*, ] table [x=events_num(n), y=scalability_vs_RegionTrack, col sep=comma,] {data/SC_region.csv};

\end{groupplot}

\end{tikzpicture}
}
\caption{\label{fig:scalability}
Scalability comparison of $\ouralgo$ vs $\aerodrome$, $\velodrome$ and $\regiontrack$, on families of traces where $\ouralgo$ exhibits its worst case complexity of $\Theta(\numEvents \numThreads^2)$.
}
\end{figure*}
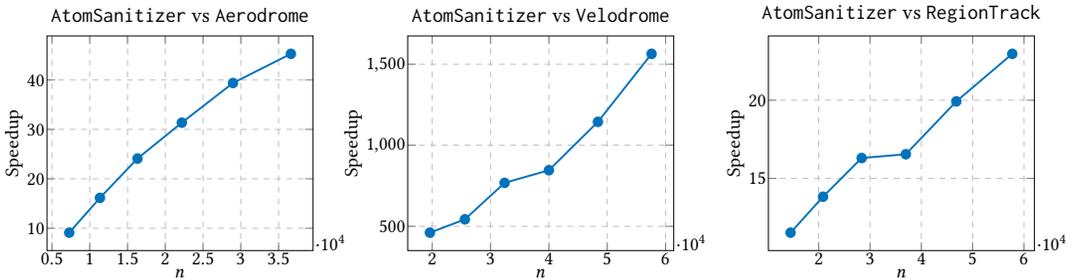

%% file: figures/SC_velodrome.tex
\begin{figure}[htbp]
    \centering
    % \begin{subfigure}[t]{\textwidth}
    \begin{tikzpicture}[xscale=2, yscale=0.35, 
        every node/.style={font=\small},
        transaction/.style={draw, dashed, thick},
        event/.style={anchor=center, inner sep=1pt},
        -, >={Triangle[width=10pt,length=5pt]}]
    
        % starting coordinates
        \def\xstart{0} % starting x coordinate
        \def\ystart{0} % starting y coordinate
        
        % width and height of the rectangle
        \def\width{3}  % width of the rectangle
        \def\height{21} % height of the rectangle
    
        % write labels for the three threads
        \node[anchor=south] at (\xstart + \width/4, \ystart + \height) {$t_1$};
        \node[anchor=south] at (\xstart + 3*\width/4, \ystart + \height) {$t_2$};
        
        % Draw the three transaction lines
        \draw[dashed] (\xstart + \width/2,\ystart) -- (\xstart + \width/2,\ystart + \height);
        
        % Draw the grid
        \foreach \y in {1,...,20} {
            \pgfmathtruncatemacro\yy{21-\y} % Compute reverse ID: 20 -> 1
            \node[left] at (\xstart,\ystart+\y) {\yy};  % Keep the coordinates, but reverse the label
        }
        
        % events in t1
        % T1: rhd, wt(x_1), wt(x_2), wt(w_3), ..., wt(n) T1: lhd
        \node[event] (e1) at (\xstart + \width/4, \ystart + \height - 1) {\bbg{T_1}};   
        \node[event] (e2) at (\xstart + \width/4, \ystart + \height - 2) {$\wt(x_1)$};
        \node[event] (e6) at (\xstart + \width/4, \ystart + \height - 6) {$\wt(x_2)$};
        \node[event] (e10) at (\xstart + \width/4, \ystart + \height - 10) {$\wt(x_3)$};
        \node[event] (e14) at (\xstart + \width/4, \ystart + \height - 14) {$\ldots$};
        \node[event] (e16) at (\xstart + \width/4, \ystart + \height - 16) {$\wt(x_{j})$};
        \node[event] (e20) at (\xstart + \width/4, \ystart + \height - 20) {\ben{T_1}};

        % events in t2
        % T2: rhd, wt(x_1), T2: lhd
        \node[event] (e3) at (\xstart + 3*\width/4, \ystart + \height - 3) {\bbg{T_2}};
        \node[event] (e4) at (\xstart + 3*\width/4, \ystart + \height - 4) {$\wt(x_1)$};
        \node[event] (e5) at (\xstart + 3*\width/4, \ystart + \height - 5) {\ben{T_2}};

        % T3: rhd, wt(x_2), T3: lhd
        \node[event] (e7) at (\xstart + 3*\width/4, \ystart + \height - 7) {\bbg{T_3}};
        \node[event] (e8) at (\xstart + 3*\width/4, \ystart + \height - 8) {$\wt(x_2)$};
        \node[event] (e9) at (\xstart + 3*\width/4, \ystart + \height - 9) {\ben{T_3}};

        % T4: rhd, wt(x_3), T4: lhd
        \node[event] (e11) at (\xstart + 3*\width/4, \ystart + \height - 11) {\bbg{T_4}};
        \node[event] (e12) at (\xstart + 3*\width/4, \ystart + \height - 12) {$\wt(x_3)$};
        \node[event] (e13) at (\xstart + 3*\width/4, \ystart + \height - 13) {\ben{T_4}};

        % ...
        \node[event] (e15) at (\xstart + 3*\width/4, \ystart + \height - 15) {$\ldots$};
        % T5: rhd, wt(n), T5: lhd
        \node[event] (e17) at (\xstart + 3*\width/4, \ystart + \height - 17) {\bbg{T_n}};
        \node[event] (e18) at (\xstart + 3*\width/4, \ystart + \height - 18) {$\wt(x_{j})$};
        \node[event] (e19) at (\xstart + 3*\width/4, \ystart + \height - 19) {\ben{T_n}};

        % Draw border
        \draw[thick] (\xstart, \ystart) rectangle ++(\width, \height);
    
        % \draw[thick] (0.5,0.5) rectangle (3.5,12.5);
    \end{tikzpicture}
    % \end{subfigure}
    \caption{
    A family of traces $(\tr_j)_j$ for the scalability experiment with $\velodrome$.
}
    \label{fig:sc_velodrome_trace}
    \Description{This figure is an example case of conflict serializable}
    \end{figure}
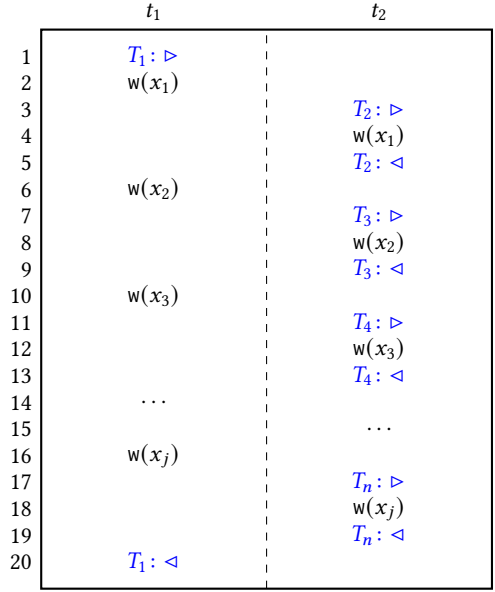

%% file: figures/SC_aerodrome.tex
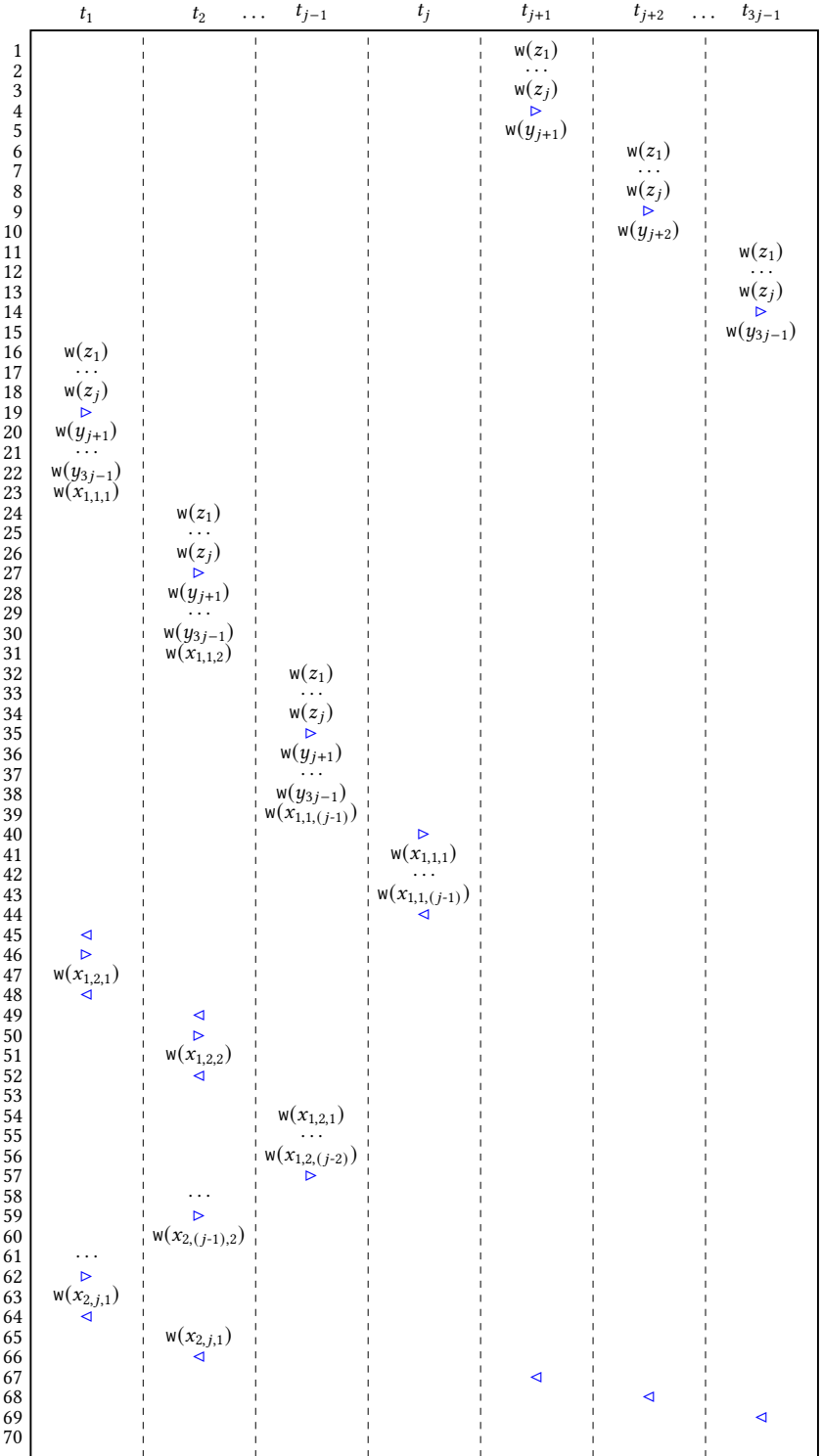
\begin{figure*}[ht!]
    \centering
    % \begin{subfigure}[t]{\textwidth}
    \scalebox{0.9}{%
    \begin{tikzpicture}[xscale=1.7, yscale=0.30, 
        every node/.style={font=\small},
        transaction/.style={draw, dashed, thick},
        event/.style={anchor=center, inner sep=1pt},
        -, >={Triangle[width=10pt,length=5pt]}]
    
        % starting coordinates
        \def\xstart{0} % starting x coordinate
        \def\ystart{0} % starting y coordinate
        
        % width and height of the rectangle
        \def\width{7}  % width of the rectangle
        \def\height{71} % height of the rectangle
    
        % write labels for the three threads
        \node[anchor=south] at (\xstart + \width/14, \ystart + \height) {$t_1$};
        \node[anchor=south] at (\xstart + 3*\width/14, \ystart + \height) {$t_2$};
        \node[anchor=south] at (\xstart + 4*\width/14, \ystart + \height) {$\cdots$};
        \node[anchor=south] at (\xstart + 5*\width/14, \ystart + \height) {$t_{j-1}$};
        \node[anchor=south] at (\xstart + 7*\width/14, \ystart + \height) {$t_j$};
        \node[anchor=south] at (\xstart + 9*\width/14, \ystart + \height) {$t_{j+1}$};
        \node[anchor=south] at (\xstart + 11*\width/14, \ystart + \height) {$t_{j+2}$};
        \node[anchor=south] at (\xstart + 12*\width/14, \ystart + \height) {$\cdots$};
        \node[anchor=south] at (\xstart + 13*\width/14, \ystart + \height) {$t_{3j-1}$};
        
        % Draw the three transaction lines
        \draw[dashed] (\xstart + \width/7,\ystart)   -- (\xstart + \width/7,\ystart + \height);
        \draw[dashed] (\xstart + 2*\width/7,\ystart) -- (\xstart + 2*\width/7,\ystart + \height);
        \draw[dashed] (\xstart + 3*\width/7,\ystart) -- (\xstart + 3*\width/7,\ystart + \height);
        \draw[dashed] (\xstart + 4*\width/7,\ystart) -- (\xstart + 4*\width/7,\ystart + \height);
        \draw[dashed] (\xstart + 5*\width/7,\ystart) -- (\xstart + 5*\width/7,\ystart + \height);
        \draw[dashed] (\xstart + 6*\width/7,\ystart) -- (\xstart + 6*\width/7,\ystart + \height);
        
        % Draw the grid
        \foreach \y in {1,...,70} {
            \pgfmathtruncatemacro\yy{71-\y} % Compute reverse ID: 20 -> 1
            \node[left] at (\xstart,\ystart+\y) {\yy};  % Keep the coordinates, but reverse the label
        }

        %events in t_{j+1}
        % wt(z_1), ..., wt(z_{j}), T:rhd, wt(y_{j+1}), T:lhd
        \node[event] (e1) at (\xstart + 9*\width/14, \ystart + \height - 1) {$\wt(z_1)$};
        \node[event] (e2) at (\xstart + 9*\width/14, \ystart + \height - 2) {$\ldots$};
        \node[event] (e3) at (\xstart + 9*\width/14, \ystart + \height - 3) {$\wt(z_{j})$};
        \node[event] (e4) at (\xstart + 9*\width/14, \ystart + \height - 4) {\pbg};
        \node[event] (e5) at (\xstart + 9*\width/14, \ystart + \height - 5) {$\wt(y_{j+1})$};
        
        % events in t_{j+2}
        % wt(z_1), ..., wt(z_{j}), T:rhd, wt(y_{j+2}), T:lhd
        \node[event] (e6) at (\xstart + 11*\width/14, \ystart + \height - 6) {$\wt(z_1)$};
        \node[event] (e7) at (\xstart + 11*\width/14, \ystart + \height - 7) {$\ldots$};
        \node[event] (e8) at (\xstart + 11*\width/14, \ystart + \height - 8) {$\wt(z_{j})$};
        \node[event] (e9) at (\xstart + 11*\width/14, \ystart + \height - 9) {\pbg};
        \node[event] (e10) at (\xstart + 11*\width/14, \ystart + \height - 10) {$\wt(y_{j+2})$};

        % events in t_{3j-1}
        % wt(z_1), ..., wt(z_{j}), T:rhd, wt(y_{3j-1}), T_{3j-1}:lhd
        \node[event] (e11) at (\xstart + 13*\width/14, \ystart + \height - 11) {$\wt(z_1)$};
        \node[event] (e12) at (\xstart + 13*\width/14, \ystart + \height - 12) {$\ldots$};
        \node[event] (e13) at (\xstart + 13*\width/14, \ystart + \height - 13) {$\wt(z_{j})$};
        \node[event] (e14) at (\xstart + 13*\width/14, \ystart + \height - 14) {\pbg};
        \node[event] (e15) at (\xstart + 13*\width/14, \ystart + \height - 15) {$\wt(y_{3j-1})$};

        % events in t_1
        % wt(z_1), ..., wt(z_{j}), T:rhd, wt(y_{j+1}), ..., wt(y_{3j-1}), wt(1-1-1)
        \node[event] (e16) at (\xstart + \width/14, \ystart + \height - 16) {$\wt(z_1)$};
        \node[event] (e17) at (\xstart + \width/14, \ystart + \height - 17) {$\ldots$};
        \node[event] (e18) at (\xstart + \width/14, \ystart + \height - 18) {$\wt(z_{j})$};
        \node[event] (e19) at (\xstart + \width/14, \ystart + \height - 19) {\pbg};
        \node[event] (e20) at (\xstart + \width/14, \ystart + \height - 20) {$\wt(y_{j+1})$};
        \node[event] (e21) at (\xstart + \width/14, \ystart + \height - 21) {$\ldots$};
        \node[event] (e22) at (\xstart + \width/14, \ystart + \height - 22) {$\wt(y_{3j-1})$};
        \node[event] (e23) at (\xstart + \width/14, \ystart + \height - 23) {$\wt(x_{\sclabel{1}{1}{1}})$};

        % events in t_2
        % wt(z_1), ..., wt(z_{j}), T:rhd, wt(y_{j+1}), ..., wt(y_{3j-1}), wt(1-1-2)
        \node[event] (e24) at (\xstart + 3*\width/14, \ystart + \height - 24) {$\wt(z_1)$};
        \node[event] (e25) at (\xstart + 3*\width/14, \ystart + \height - 25) {$\ldots$};
        \node[event] (e26) at (\xstart + 3*\width/14, \ystart + \height - 26) {$\wt(z_{j})$};
        \node[event] (e27) at (\xstart + 3*\width/14, \ystart + \height - 27) {\pbg};
        \node[event] (e28) at (\xstart + 3*\width/14, \ystart + \height - 28) {$\wt(y_{j+1})$};
        \node[event] (e29) at (\xstart + 3*\width/14, \ystart + \height - 29) {$\ldots$};
        \node[event] (e30) at (\xstart + 3*\width/14, \ystart + \height - 30) {$\wt(y_{3j-1})$};
        \node[event] (e31) at (\xstart + 3*\width/14, \ystart + \height - 31) {$\wt(x_{\sclabel{1}{1}{2}})$};

        % events in t_{k-1}
        % wt(z_1), ..., wt(z_{j}), T:rhd, wt(y_{j+1}), ..., wt(y_{3j-1}), wt(1-1-(k-1))
        \node[event] (e32) at (\xstart + 5*\width/14, \ystart + \height - 32) {$\wt(z_1)$};
        \node[event] (e33) at (\xstart + 5*\width/14, \ystart + \height - 33) {$\ldots$};
        \node[event] (e34) at (\xstart + 5*\width/14, \ystart + \height - 34) {$\wt(z_{j})$};
        \node[event] (e35) at (\xstart + 5*\width/14, \ystart + \height - 35) {\pbg};
        \node[event] (e36) at (\xstart + 5*\width/14, \ystart + \height - 36) {$\wt(y_{j+1})$};
        \node[event] (e37) at (\xstart + 5*\width/14, \ystart + \height - 37) {$\ldots$};
        \node[event] (e38) at (\xstart + 5*\width/14, \ystart + \height - 38) {$\wt(y_{3j-1})$};
        \node[event] (e39) at (\xstart + 5*\width/14, \ystart + \height - 39) {$\wt(x_{\sclabel{1}{1}{(j\sd1)}})$};

        % events in t_k
        % T:rhd, wt(1-1-1), ..., wt(1-1-(k-1)), T:lhd
        \node[event] (e40) at (\xstart + 7*\width/14, \ystart + \height - 40) {\pbg};
        \node[event] (e41) at (\xstart + 7*\width/14, \ystart + \height - 41) {$\wt(x_{\sclabel{1}{1}{1}})$};
        \node[event] (e42) at (\xstart + 7*\width/14, \ystart + \height - 42) {$\ldots$};
        \node[event] (e43) at (\xstart + 7*\width/14, \ystart + \height - 43) {$\wt(x_{\sclabel{1}{1}{(j\sd1)}})$};
        \node[event] (e44) at (\xstart + 7*\width/14, \ystart + \height - 44) {\pen};

        % events in t_1
        % T:lhd, T:rhb, wt(1-2-1), T:lhb
        \node[event] (e45) at (\xstart + \width/14, \ystart + \height - 45) {\pen};
        \node[event] (e46) at (\xstart + \width/14, \ystart + \height - 46) {\pbg};
        \node[event] (e47) at (\xstart + \width/14, \ystart + \height - 47) {$\wt(x_{\sclabel{1}{2}{1}})$};
        \node[event] (e48) at (\xstart + \width/14, \ystart + \height - 48) {\pen};

        % events in t_2
        % T:lhd, T:rhb, wt(1-2-2), T:lhb
        \node[event] (e49) at (\xstart + 3*\width/14, \ystart + \height - 49) {\pen};
        \node[event] (e50) at (\xstart + 3*\width/14, \ystart + \height - 50) {\pbg};
        \node[event] (e51) at (\xstart + 3*\width/14, \ystart + \height - 51) {$\wt(x_{\sclabel{1}{2}{2}})$};
        \node[event] (e52) at (\xstart + 3*\width/14, \ystart + \height - 52) {\pen};

        % events in t_{k-1}
        % wt(1-2-1), ..., wt(1-2-(k-2)), T:rhb
        \node[event] (e54) at (\xstart + 5*\width/14, \ystart + \height - 54) {$\wt(x_{\sclabel{1}{2}{1}})$};
        \node[event] (e55) at (\xstart + 5*\width/14, \ystart + \height - 55) {$\ldots$};
        \node[event] (e56) at (\xstart + 5*\width/14, \ystart + \height - 56) {$\wt(x_{\sclabel{1}{2}{(j\sd2)}})$};
        \node[event] (e57) at (\xstart + 5*\width/14, \ystart + \height - 57) {\pbg};

        % events in t_2
        % ..., T:rhb, wt(v-(k-1)-2)
        \node[event] (e58) at (\xstart + 3*\width/14, \ystart + \height - 58) {$\ldots$};
        \node[event] (e59) at (\xstart + 3*\width/14, \ystart + \height - 59) {\pbg};
        \node[event] (e60) at (\xstart + 3*\width/14, \ystart + \height - 60) {$\wt(x_{\sclabel{2}{(j\sd1)}{2}})$};

        % events in t_1
        % ..., T:rhb, wt(v-k-1), T:lhb
        \node[event] (e61) at (\xstart + \width/14, \ystart + \height - 61) {$\ldots$};
        \node[event] (e62) at (\xstart + \width/14, \ystart + \height - 62) {\pbg};
        \node[event] (e63) at (\xstart + \width/14, \ystart + \height - 63) {$\wt(x_{\sclabel{2}{j}{1}})$};
        \node[event] (e64) at (\xstart + \width/14, \ystart + \height - 64) {\pen};

        % events in t_2
        % wt(v-k-1), T:lhb
        \node[event] (e65) at (\xstart + 3*\width/14, \ystart + \height - 65) {$\wt(x_{\sclabel{2}{j}{1}})$};
        \node[event] (e66) at (\xstart + 3*\width/14, \ystart + \height - 66) {\pen};

        % events in t_{j+1}
        % T:lhd
        \node[event] (e67) at (\xstart + 9*\width/14, \ystart + \height - 67) {\pen};

        % events in t_{j+2}
        % T:lhd
        \node[event] (e68) at (\xstart + 11*\width/14, \ystart + \height - 68) {\pen};

        % events in t_{3j-1}
        % T_{3j-1}:lhd
        \node[event] (e69) at (\xstart + 13*\width/14, \ystart + \height - 69) {\pen};

        % Draw border
        \draw[thick] (\xstart, \ystart) rectangle ++(\width, \height);
    
        % \draw[thick] (0.5,0.5) rectangle (3.5,12.5);
    \end{tikzpicture}
    }%
    % \end{subfigure}
    \caption{
    A family of traces $(\tr_j)_j$ for the scalability experiment with $\velodrome$.
}
    \label{fig:sc_aerodrome_trace}
    \Description{This figure is an example case of conflict serializable}
\end{figure*}

%% file: figures/SC_region.tex
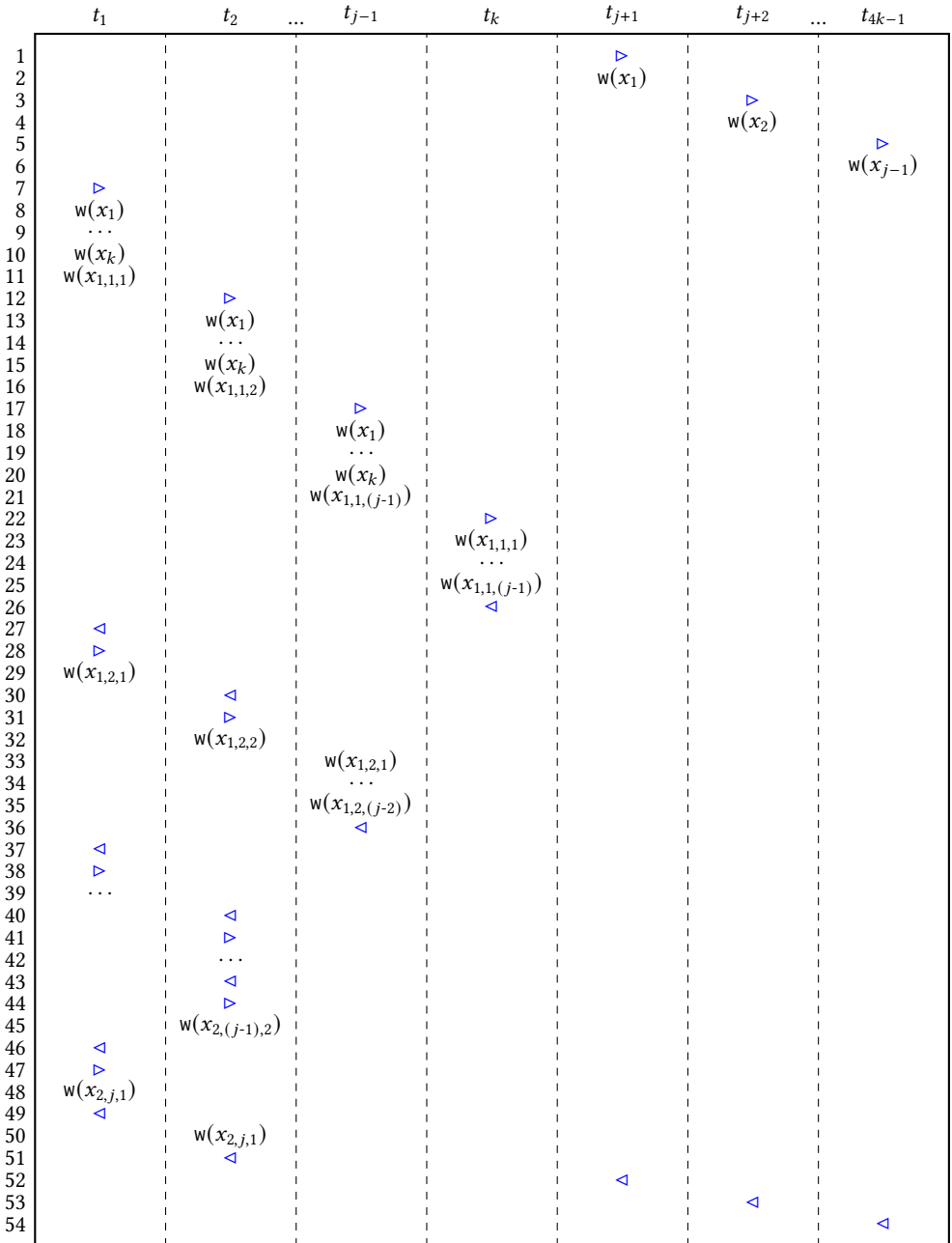
\begin{figure*}[ht!]
    \centering
    % \begin{subfigure}[t]{\textwidth}
    \begin{tikzpicture}[xscale=1.8, yscale=0.30, 
        every node/.style={font=\small},
        transaction/.style={draw, dashed, thick},
        event/.style={anchor=center, inner sep=1pt},
        -, >={Triangle[width=10pt,length=5pt]}]
    
        % starting coordinates
        \def\xstart{0} % starting x coordinate
        \def\ystart{0} % starting y coordinate
        
        % width and height of the rectangle
        \def\width{7}  % width of the rectangle
        \def\height{55} % height of the rectangle
    
        % write labels for the three threads
        \node[anchor=south] at (\xstart + \width/14, \ystart + \height) {$t_1$};
        \node[anchor=south] at (\xstart + 3*\width/14, \ystart + \height) {$t_2$};
        \node[anchor=south] at (\xstart + 4*\width/14, \ystart + \height) {$...$};
        \node[anchor=south] at (\xstart + 5*\width/14, \ystart + \height) {$t_{j-1}$};
        \node[anchor=south] at (\xstart + 7*\width/14, \ystart + \height) {$t_k$};
        \node[anchor=south] at (\xstart + 9*\width/14, \ystart + \height) {$t_{j+1}$};
        \node[anchor=south] at (\xstart + 11*\width/14, \ystart + \height) {$t_{j+2}$};
        \node[anchor=south] at (\xstart + 12*\width/14, \ystart + \height) {$...$};
        \node[anchor=south] at (\xstart + 13*\width/14, \ystart + \height) {$t_{4k-1}$};
        
        % Draw the three transaction lines
        \draw[dashed] (\xstart + \width/7,\ystart)   -- (\xstart + \width/7,\ystart + \height);
        \draw[dashed] (\xstart + 2*\width/7,\ystart) -- (\xstart + 2*\width/7,\ystart + \height);
        \draw[dashed] (\xstart + 3*\width/7,\ystart) -- (\xstart + 3*\width/7,\ystart + \height);
        \draw[dashed] (\xstart + 4*\width/7,\ystart) -- (\xstart + 4*\width/7,\ystart + \height);
        \draw[dashed] (\xstart + 5*\width/7,\ystart) -- (\xstart + 5*\width/7,\ystart + \height);
        \draw[dashed] (\xstart + 6*\width/7,\ystart) -- (\xstart + 6*\width/7,\ystart + \height);
        
        % Draw the grid
        \foreach \y in {1,...,54} {
            \pgfmathtruncatemacro\yy{55-\y} % Compute reverse ID: 20 -> 1
            \node[left] at (\xstart,\ystart+\y) {\yy};  % Keep the coordinates, but reverse the label
        }

        % events in t_{j+1}
        % T:rhd, wt(x_1)
        \node[event] (e1) at (\xstart + 9*\width/14, \ystart + \height - 1) {\pbg};
        \node[event] (e2) at (\xstart + 9*\width/14, \ystart + \height - 2) {$\wt(x_1)$};

        % events in t_{j+2}
        % T:rhd, wt(x_2)
        \node[event] (e3) at (\xstart + 11*\width/14, \ystart + \height - 3) {\pbg};
        \node[event] (e4) at (\xstart + 11*\width/14, \ystart + \height - 4) {$\wt(x_2)$};

        % events in t_{2k-1}
        % T:rhd, wt(x_k)
        \node[event] (e5) at (\xstart + 13*\width/14, \ystart + \height - 5) {\pbg};
        \node[event] (e6) at (\xstart + 13*\width/14, \ystart + \height - 6) {$\wt(x_{j-1})$};

        %events in t_1
        % T:rhd, wt(x_1), ..., wt(x_k), wt(1-1-1)
        \node[event] (e7) at (\xstart + \width/14, \ystart + \height - 7) {\pbg};
        \node[event] (e8) at (\xstart + \width/14, \ystart + \height - 8) {$\wt(x_1)$};
        \node[event] (e9) at (\xstart + \width/14, \ystart + \height - 9) {$\ldots$};
        \node[event] (e10) at (\xstart + \width/14, \ystart + \height - 10) {$\wt(x_k)$};
        \node[event] (e11) at (\xstart + \width/14, \ystart + \height - 11) {$\wt(x_{\sclabel{1}{1}{1}})$};

        %events in t_2
        % T:rhd, wt(x_1), ..., wt(x_k), wt(1-1-2)
        \node[event] (e12) at (\xstart + 3*\width/14, \ystart + \height - 12) {\pbg};
        \node[event] (e13) at (\xstart + 3*\width/14, \ystart + \height - 13) {$\wt(x_1)$};
        \node[event] (e14) at (\xstart + 3*\width/14, \ystart + \height - 14) {$\ldots$};
        \node[event] (e15) at (\xstart + 3*\width/14, \ystart + \height - 15) {$\wt(x_k)$};
        \node[event] (e16) at (\xstart + 3*\width/14, \ystart + \height - 16) {$\wt(x_{\sclabel{1}{1}{2}})$};

        %events in t_{j-1}
        % T:rhd, wt(x_1), ..., wt(x_k), wt(1-1-(k-1))
        \node[event] (e17) at (\xstart + 5*\width/14, \ystart + \height - 17) {\pbg};
        \node[event] (e18) at (\xstart + 5*\width/14, \ystart + \height - 18) {$\wt(x_1)$};
        \node[event] (e19) at (\xstart + 5*\width/14, \ystart + \height - 19) {$\ldots$};
        \node[event] (e20) at (\xstart + 5*\width/14, \ystart + \height - 20) {$\wt(x_k)$};
        \node[event] (e21) at (\xstart + 5*\width/14, \ystart + \height - 21) {$\wt(x_{\sclabel{1}{1}{(j\sd1)}})$};

        %events in t_k
        % T:rhd, wt(1-1-1), ..., wt(1-1-(k-1)), T:lhd
        \node[event] (e22) at (\xstart + 7*\width/14, \ystart + \height - 22) {\pbg};
        \node[event] (e23) at (\xstart + 7*\width/14, \ystart + \height - 23) {$\wt(x_{\sclabel{1}{1}{1}})$};
        \node[event] (e24) at (\xstart + 7*\width/14, \ystart + \height - 24) {$\ldots$};
        \node[event] (e25) at (\xstart + 7*\width/14, \ystart + \height - 25) {$\wt(x_{\sclabel{1}{1}{(j\sd1)}})$};
        \node[event] (e26) at (\xstart + 7*\width/14, \ystart + \height - 26) {\pen};

        %events in t_1
        % T:lhd, T:rhd, wt(1-2-1)
        \node[event] (e27) at (\xstart + \width/14, \ystart + \height - 27) {\pen};
        \node[event] (e28) at (\xstart + \width/14, \ystart + \height - 28) {\pbg};
        \node[event] (e29) at (\xstart + \width/14, \ystart + \height - 29) {$\wt(x_{\sclabel{1}{2}{1}})$};

        %events in t_2
        % T:lhd, T:rhd, wt(1-2-2)
        \node[event] (e30) at (\xstart + 3*\width/14, \ystart + \height - 30) {\pen};
        \node[event] (e31) at (\xstart + 3*\width/14, \ystart + \height - 31) {\pbg};
        \node[event] (e32) at (\xstart + 3*\width/14, \ystart + \height - 32) {$\wt(x_{\sclabel{1}{2}{2}})$};

        %events in t_{j-1}
        % T:lhd, T:rhd, wt(1-2-1), ..., wt(1-2-(k-2)), T:lhd
        \node[event] (e33) at (\xstart + 5*\width/14, \ystart + \height - 33) {$\wt(x_{\sclabel{1}{2}{1}})$};
        \node[event] (e34) at (\xstart + 5*\width/14, \ystart + \height - 34) {$\ldots$};
        \node[event] (e35) at (\xstart + 5*\width/14, \ystart + \height - 35) {$\wt(x_{\sclabel{1}{2}{(j\sd2)}})$};
        \node[event] (e36) at (\xstart + 5*\width/14, \ystart + \height - 36) {\pen};

        %events in t_1
        % T:lhd, T:rhd, ...
        \node[event] (e37) at (\xstart + \width/14, \ystart + \height - 37) {\pen};
        \node[event] (e38) at (\xstart + \width/14, \ystart + \height - 38) {\pbg};
        \node[event] (e39) at (\xstart + \width/14, \ystart + \height - 39) {$\ldots$};

        %events in t_2
        % T:lhd, T:rhd, ...
        \node[event] (e40) at (\xstart + 3*\width/14, \ystart + \height - 40) {\pen};
        \node[event] (e41) at (\xstart + 3*\width/14, \ystart + \height - 41) {\pbg};
        \node[event] (e42) at (\xstart + 3*\width/14, \ystart + \height - 42) {$\ldots$};

        %events in t_2
        % T:lhd, T:rhd, wt(v-(k-1)-2)
        \node[event] (e43) at (\xstart + 3*\width/14, \ystart + \height - 43) {\pen};
        \node[event] (e44) at (\xstart + 3*\width/14, \ystart + \height - 44) {\pbg};
        \node[event] (e45) at (\xstart + 3*\width/14, \ystart + \height - 45) {$\wt(x_{\sclabel{2}{(j\sd1)}{2}})$};

        %events in t_1
        % T:lhd, T:rhd, wt(v-k-1), T:lhd
        \node[event] (e46) at (\xstart + \width/14, \ystart + \height - 46) {\pen};
        \node[event] (e47) at (\xstart + \width/14, \ystart + \height - 47) {\pbg};
        \node[event] (e48) at (\xstart + \width/14, \ystart + \height - 48) {$\wt(x_{\sclabel{2}{j}{1}})$};
        \node[event] (e49) at (\xstart + \width/14, \ystart + \height - 49) {\pen};

        %events in t_2
        % wt(v-k-1), T:lhd
        \node[event] (e50) at (\xstart + 3*\width/14, \ystart + \height - 50) {$\wt(x_{\sclabel{2}{j}{1}})$};
        \node[event] (e51) at (\xstart + 3*\width/14, \ystart + \height - 51) {\pen};

        %events in t_{j+1}
        % T:lhd
        \node[event] (e52) at (\xstart + 9*\width/14, \ystart + \height - 52) {\pen};

        %events in t_{j+2}
        % T:lhd
        \node[event] (e53) at (\xstart + 11*\width/14, \ystart + \height - 53) {\pen};

        %events in t_{2k-1}
        % T:lhd
        \node[event] (e54) at (\xstart + 13*\width/14, \ystart + \height - 54) {\pen};

        % Draw border
        \draw[thick] (\xstart, \ystart) rectangle ++(\width, \height);
    
        % \draw[thick] (0.5,0.5) rectangle (3.5,12.5);
    \end{tikzpicture}
    % \end{subfigure}
    \caption{
    A family of traces $(\tr_j)_j$ for the scalability experiment with $\velodrome$.
}
    \label{fig:sc_region_trace}
    \Description{This figure is an example case of conflict serializable}
\end{figure*}

%% file: sections/app_doublechecker.tex
\section{Experiments on DoubleChecker}\label{sec:app_doublechecker}

We have also implemented the first phase of $\doublechecker$, which performs the \emph{Imprecise cycle detection} (ICD).
In particular, it uses $\octet$~\cite{Bond-2013}, a software based concurrency control mechanism, to maintain the imprecise happens-before relations.
With this technique, the ICD phase does not record the last write or read event of each memory location.
This allows it to process each event faster, while maintaining a conflict graph that overapproximates $\thb{\tr}$.
If a cycle is detected in this graph, $\doublechecker$ runs $\velodrome$ from scratch, to try and verify the cycle.
Since $\doublechecker$ does not make precise violation reports, we have not used in our main experimental setting.
Here we report how the ICD phase $\doublechecker$ performs in our benchmark set.

The results are shown in \cref{tab:dbc_offline}.
We see that ICD is overall considerably faster than $\velodrome$.
For some benchmarks, $\doublechecker$ stops earlier than other tools, because it will conservatively insert unnecessary edges into the transactional happens-before graph, which may introduce cycles.
Note, however, that ICD reports violations on all benchmarks except two (\texttt{philo} and \texttt{fop}), and has a high false positive rate.
Rerunning $\velodrome$ on all the positive reports of ICD clearly results in running times that are far larger than $\ouralgo$, as shown in \cref{tab:offline}.

\input{tables/dbc_offline.tex}

%% file: tables/dbc_offline.tex
\begin{table}[htbp]
  \centering
  \setlength\tabcolsep{1.3pt}
  \def\angle{45}
  \caption{
    Offline experiments of the first phase of $\doublechecker$ on conflict-serializability.
  }
  \pgfplotstableread[col sep=comma]{data/CS_dbc_offline.csv}\datatable
  \small
  \setlength\tabcolsep{1.4pt}
  \pgfplotstabletypeset[%
    columns={Benchmarks,doublechecker(ms),violation_reported,true_positive},
    assign column name/.style={/pgfplots/table/column name={\textbf{#1}}},
    % every head row/.style={before row=\toprule, after row=\midrule, font=bold},
    every head row/.style={before row=, after row=\midrule, output empty row},
    every last row/.style={after row=\bottomrule},
    columns/{Benchmarks}/.style={
      column type={l|},
      string type,
      column name=,
      postproc cell content/.prefix code={%
        \pgfkeyssetvalue{/pgfplots/table/@cell content}%
        {\ttfamily ##1}%
      },
    },
    every head row/.append style={
      before row={
        \midrule
        \textbf{Benchmark} & \textbf{Conflict Serializability (ms)} & \textbf{Violation Reported?} & \textbf{True Positive?} \\
        % \cline{5-10}
      }
    },
    columns/{doublechecker(ms)}/.style={
      column type={r|},
      string type,
      postproc cell content/.append code={
        \count0=\pgfplotstablerow
        \advance\count0 by1
        \ifnum\count0=\pgfplotstablerows
          \pgfkeysalso{@cell content=$\mathbf{\pgfmathprintnumber[assume math mode=true, fixed, fixed zerofill, precision=0]{##1}}$}%
        \else
          \ifnum\count0=\numexpr\pgfplotstablerows-1
            % \pgfkeysalso{@cell content=$\mathbf{\pgfmathprintnumber[assume math mode=true, fixed, fixed zerofill, precision=0]{##1}}$}%
            \pgfkeysalso{@cell content=${\pgfmathprintnumber[assume math mode=true, fixed, fixed zerofill, precision=0]{##1}}$}%
          \else
            \pgfkeysalso{@cell content=${\pgfmathprintnumber[assume math mode=true, fixed, fixed zerofill, precision=0]{##1}}$}%
          \fi
        \fi
      },
    },
    columns/{violation_reported}/.style={
      column type={c|},
      string type,
      postproc cell content/.append code={
        \count0=\pgfplotstablerow
        \advance\count0 by1
        \ifnum\count0=\pgfplotstablerows
          \pgfkeysalso{@cell content=\textbf{-}}%
        \else
          \edef\temp{##1}%
          \def\tempone{1}%
          \ifx\temp\tempone
            \pgfkeysalso{@cell content={$\checkmark$}}%
          \else
            \pgfkeysalso{@cell content={$\times$}}%
          \fi
        \fi
      },
    },
    columns/{true_positive}/.style={
      column type={c},
      string type,
      postproc cell content/.append code={
        \count0=\pgfplotstablerow
        \advance\count0 by1
        \ifnum\count0=\pgfplotstablerows
          \pgfkeysalso{@cell content=\textbf{-}}%
        \else
          \edef\temp{##1}%
          \def\tempone{1}%
          \def\tempzero{0}%
          \ifx\temp\tempone
            \pgfkeysalso{@cell content={$\checkmark$}}%
          \else
            \ifx\temp\tempzero
              \pgfkeysalso{@cell content={$\times$}}%
            \else
              \pgfkeysalso{@cell content={-}}%
            \fi
          \fi
        \fi
      },
    },
    every row no 16/.style={
      after row={\midrule},
    },
    highlight/.append style={
      postproc cell content/.append code={
        \pgfkeysalso{@cell content=\textbf{##1}}%
      },
    },
    highlight last two rows/.style={
      postproc cell content/.append code={
        \count0=\pgfplotstablerow
        \advance\count0 by1
        \ifnum\count0=\pgfplotstablerows
          \pgfkeysalso{@cell content=\textbf{##1}}%
        \fi
      },
    },
    highlight last two rows,
  ]{\datatable}

  \label{tab:dbc_offline}

\end{table}

%% file: sections/app_aerodrome.tex
\section{Version of Aerodrome in our Experiments}\label{sec:app_aerodrome}

$\RAPID$ implements an optimized version of $\aerodrome$, following Algorithm 3 in the technical report of \cite{Mathur_2020}.
We found that this optimized algorithm has a correctness issue, failing to detect conflict-serializability violations in some traces (see \cref{fig:aerodrome_counter_example}).
We reported this issue to the main maintainer of $\RAPID$, who provided a fix.
Although the fix does not affect the worst-case complexity of $\aerodrome$, we observed that it incurs a performance slowdown in practice, compared to the original optimized algorithm.
\input{figures/aerodrome_counter_example}

%% file: figures/aerodrome_counter_example.tex
\begin{figure}[htbp]
    \centering
    \begin{tikzpicture}[xscale=1.6, yscale=0.35, 
        every node/.style={font=\small},
        transaction/.style={draw, dashed, thick},
        event/.style={anchor=center, inner sep=1pt},
        -, >={Triangle[width=10pt,length=5pt]}]
    
        % starting coordinates
        \def\xstart{0} % starting x coordinate
        \def\ystart{0} % starting y coordinate
        
        % width and height of the rectangle
        \def\width{3}  % width of the rectangle
        \def\height{13} % height of the rectangle
    
        % write labels for the three threads
        \node[anchor=south] at (\xstart + \width/6, \ystart + \height) {$t_1$};
        \node[anchor=south] at (\xstart + \width/2, \ystart + \height) {$t_2$};
        \node[anchor=south] at (\xstart + 5*\width/6, \ystart + \height) {$t_3$};
        
        % Draw the three transaction lines
        \draw[dashed] (\xstart + \width/3,\ystart) -- (\xstart + \width/3,\ystart + \height);
        \draw[dashed] (\xstart + 2*\width/3, \ystart) -- (\xstart + 2*\width/3, \ystart + \height);
        
        % Draw the grid
        \foreach \y in {1,...,12} {
            \pgfmathtruncatemacro\yy{13-\y} % Compute reverse ID: 14 -> 1
            \node[left] at (\xstart,\ystart+\y) {\yy};  % Keep the coordinates, but reverse the label
        }
        
        % events in t1
        \node[event] (e1) at (\xstart + \width/6, \ystart + \height - 1) {\bbg{T_1}};   
        \node[event] (e2) at (\xstart + \width/6, \ystart + \height - 2) {$\wt(x)$};
        \node[event] (e5) at (\xstart + \width/6, \ystart + \height - 5) {$\wt(y)$};
        \node[event] (e6) at (\xstart + \width/6, \ystart + \height - 6) {\ben{T_1}};

        % events in t2
        \node[event] (e3) at (\xstart + \width/2, \ystart + \height - 3) {\bbg{T_2}};
        \node[event] (e4) at (\xstart + \width/2, \ystart + \height - 4) {$\wt(y)$};
        \node[event] (e9) at (\xstart + \width/2, \ystart + \height - 9) {$\wt(z)$};
        \node[event] (e10) at (\xstart + \width/2, \ystart + \height - 10) {\ben{T_2}};

        % events in t3
        \node[event] (e7) at (\xstart + 5*\width/6, \ystart + \height - 7) {\bbg{T_3}};
        \node[event] (e8) at (\xstart + 5*\width/6, \ystart + \height - 8) {$\wt(z)$};
        \node[event] (e11) at (\xstart + 5*\width/6, \ystart + \height - 11) {$\wt(x)$};
        \node[event] (e12) at (\xstart + 5*\width/6, \ystart + \height - 12) {\ben{T_3}};

        % e4 -> e5
        \draw[po] (e4.west) to node[above, circled, left]{1} (e5.east);
        % e8 -> e9
        \draw[po] (e8.west) to node[above, circled]{2} (e9.east);
        % e2 -> e11
        \draw[vio, bend left=0] (e2.east) to node[above, circled]{3} (e11.west);

        % e1 -> e9
        %\draw[thb,dash pattern=on 5pt off 5pt, bend left=6] (e1.east) to node[right, circled]{3} (e9.west);
        
        % Draw border
        \draw[thick] (\xstart, \ystart) rectangle ++(\width, \height);
    
        % \draw[thick] (0.5,0.5) rectangle (3.5,12.5);
    \end{tikzpicture}
    \caption{
    A conflict-serializability violation missed by the original optimized version of $\aerodrome$ in $\RAPID$.
    }
    \label{fig:aerodrome_counter_example}
    \Description{This figure is an counter example of $\aerodrome$}
\end{figure}
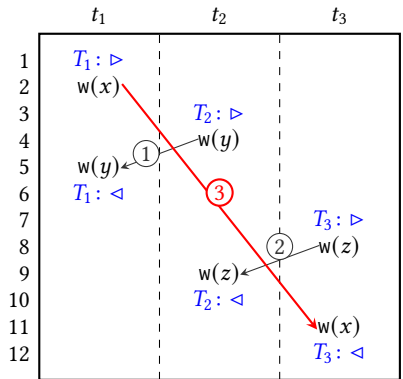

%% file: references.bib
@inproceedings{Clark2024,
  title = {Validating {{Database System Isolation Level Implementations}} with {{Version Certificate Recovery}}},
  booktitle = {Proceedings of the {{Nineteenth European Conference}} on {{Computer Systems}}},
  author = {Clark, Jack and Donaldson, Alastair F. and Wickerson, John and Rigger, Manuel},
  year = 2024,
  month = apr,
  series = {{{EuroSys}} '24},
  pages = {754--768},
  publisher = {Association for Computing Machinery},
  address = {New York, NY, USA},
  doi = {10.1145/3627703.3650080},
  url = {https://dl.acm.org/doi/10.1145/3627703.3650080},
  isbn = {979-8-4007-0437-6}
}

@article{Biswas2019,
  title = {On the Complexity of Checking Transactional Consistency},
  author = {Biswas, Ranadeep and Enea, Constantin},
  year = 2019,
  month = oct,
  journal = {Proceedings of the ACM on Programming Languages},
  volume = {3},
  number = {OOPSLA},
  pages = {165:1--165:28},
  doi = {10.1145/3360591},
  url = {https://dl.acm.org/doi/10.1145/3360591}
}

@article{Moldrup2025,
  title = {{{AWDIT}}: {{An Optimal Weak Database Isolation Tester}}},
  shorttitle = {{{AWDIT}}},
  author = {M{\o}ldrup, Lasse and Pavlogiannis, Andreas},
  year = 2025,
  month = jun,
  journal = {Artifact for "AWDIT: An Optimal Weak Database Isolation Tester"},
  volume = {9},
  number = {PLDI},
  pages = {209:1540--209:1564},
  doi = {10.1145/3742465},
  url = {https://dl.acm.org/doi/10.1145/3742465}
}

@article{Gibbons2002,
  title = {Black-{{Box Correctness Tests}} for {{Basic Parallel Data Structures}}},
  author = {Gibbons, Phillip B. and Bruno, John L. and Phillips, Steven},
  year = {2002},
  month = aug,
  journal = {Theory of Computing Systems},
  volume = {35},
  number = {4},
  pages = {391--432},
  issn = {1433-0490},
  doi = {10.1007/s00224-002-1046-6},
  url = {https://doi.org/10.1007/s00224-002-1046-6}
}

@article{Emmi2017,
author = {Emmi, Michael and Enea, Constantin},
title = {Sound, complete, and tractable linearizability monitoring for concurrent collections},
year = {2017},
issue_date = {January 2018},
publisher = {Association for Computing Machinery},
address = {New York, NY, USA},
volume = {2},
number = {POPL},
url = {https://doi.org/10.1145/3158113},
doi = {10.1145/3158113},
journal = {Proc. ACM Program. Lang.},
month = dec,
articleno = {25},
numpages = {27}
}

@inproceedings{Henzinger2015,
  title = {Unifying and {{Strengthening Hardness}} for {{Dynamic Problems}} via the {{Online Matrix-Vector Multiplication Conjecture}}},
  booktitle = {Proceedings of the Forty-Seventh Annual {{ACM}} Symposium on {{Theory}} of {{Computing}}},
  author = {Henzinger, Monika and Krinninger, Sebastian and Nanongkai, Danupon and Saranurak, Thatchaphol},
  year = {2015},
  month = jun,
  series = {{{STOC}} '15},
  pages = {21--30},
  publisher = {Association for Computing Machinery},
  address = {New York, NY, USA},
  doi = {10.1145/2746539.2746609},
  isbn = {978-1-4503-3536-2}
}

@article{vonPraun-2003,
author = {von Praun, Christoph and Gross, Thomas R.},
title = {Static conflict analysis for multi-threaded object-oriented programs},
year = {2003},
issue_date = {May 2003},
publisher = {Association for Computing Machinery},
address = {New York, NY, USA},
volume = {38},
number = {5},
issn = {0362-1340},
url = {https://doi.org/10.1145/780822.781145},
doi = {10.1145/780822.781145},
journal = {SIGPLAN Not.},
month = may,
pages = {115–128},
numpages = {14}
}

@inproceedings{Lidbury-2017,
author = {Lidbury, Christopher and Donaldson, Alastair F.},
title = {Dynamic race detection for C++11},
year = {2017},
isbn = {9781450346603},
publisher = {Association for Computing Machinery},
address = {New York, NY, USA},
url = {https://doi.org/10.1145/3009837.3009857},
doi = {10.1145/3009837.3009857},
booktitle = {Proceedings of the 44th ACM SIGPLAN Symposium on Principles of Programming Languages},
pages = {443–457},
numpages = {15},
location = {Paris, France},
series = {POPL '17}
}

@article{Bond-2013,
author = {Bond, Michael D. and Kulkarni, Milind and Cao, Man and Zhang, Minjia and Fathi Salmi, Meisam and Biswas, Swarnendu and Sengupta, Aritra and Huang, Jipeng},
title = {OCTET: capturing and controlling cross-thread dependences efficiently},
year = {2013},
issue_date = {October 2013},
publisher = {Association for Computing Machinery},
address = {New York, NY, USA},
volume = {48},
number = {10},
issn = {0362-1340},
url = {https://doi.org/10.1145/2544173.2509519},
doi = {10.1145/2544173.2509519},
journal = {SIGPLAN Not.},
month = oct,
pages = {693–712},
numpages = {20}
}

@article{Bond-2006,
author = {Bond, Michael D. and McKinley, Kathryn S.},
title = {Bell: bit-encoding online memory leak detection},
year = {2006},
issue_date = {November 2006},
publisher = {Association for Computing Machinery},
address = {New York, NY, USA},
volume = {41},
number = {11},
issn = {0362-1340},
url = {https://doi.org/10.1145/1168918.1168866},
doi = {10.1145/1168918.1168866},
journal = {SIGPLAN Not.},
month = oct,
pages = {61–72},
numpages = {12}
}

@inproceedings{Kulkarni-2021a,
  title = {Dynamic {{Data-Race Detection Through}} the {{Fine-Grained Lens}}},
  booktitle = {{{DROPS-IDN}}/v2/Document/10.4230/{{LIPIcs}}.{{CONCUR}}.2021.16},
  author = {Kulkarni, Rucha and Mathur, Umang and Pavlogiannis, Andreas},
  year = {2021},
  publisher = {Schloss Dagstuhl -- Leibniz-Zentrum f{\"u}r Informatik},
  doi = {10.4230/LIPIcs.CONCUR.2021.16},
  copyright = {https://creativecommons.org/licenses/by/4.0/legalcode},
  address = {Paris, France},
  numpages = {23},
  pages =	{16:1--16:23},
}

@inproceedings{Gorjiara-2022,
author = {Gorjiara, Hamed and Xu, Guoqing Harry and Demsky, Brian},
title = {Yashme: detecting persistency races},
year = {2022},
isbn = {9781450392051},
publisher = {Association for Computing Machinery},
address = {New York, NY, USA},
url = {https://doi.org/10.1145/3503222.3507766},
doi = {10.1145/3503222.3507766},
booktitle = {Proceedings of the 27th ACM International Conference on Architectural Support for Programming Languages and Operating Systems},
pages = {830–845},
numpages = {16},
location = {Lausanne, Switzerland},
series = {ASPLOS '22}
}

@article{Alur-2000,
title = {Model-Checking of Correctness Conditions for Concurrent Objects},
journal = {Information and Computation},
volume = {160},
number = {1},
pages = {167-188},
year = {2000},
issn = {0890-5401},
doi = {https://doi.org/10.1006/inco.1999.2847},
url = {https://www.sciencedirect.com/science/article/pii/S089054019992847X},
author = {Rajeev Alur and Ken McMillan and Doron Peled}
}

@inproceedings{Farzan-2008,
  title = {Monitoring {{Atomicity}} in {{Concurrent Programs}}},
  booktitle = {Computer {{Aided Verification}}},
  author = {Farzan, Azadeh and Madhusudan, P.},
  editor = {Gupta, Aarti and Malik, Sharad},
  year = {2008},
  pages = {52--65},
  publisher = {Springer},
  address = {Berlin, Heidelberg},
  doi = {10.1007/978-3-540-70545-1_8},
  isbn = {978-3-540-70545-1}
}

@inproceedings{Mathur-2022,
author = {Mathur, Umang and Pavlogiannis, Andreas and Tun\c{c}, H\"{u}nkar Can and Viswanathan, Mahesh},
title = {A Tree Clock Data Structure for Causal Orderings in Concurrent Executions},
year = {2022},
isbn = {9781450392051},
publisher = {Association for Computing Machinery},
address = {New York, NY, USA},
url = {https://doi.org/10.1145/3503222.3507734},
doi = {10.1145/3503222.3507734},
booktitle = {Proceedings of the 27th ACM International Conference on Architectural Support for Programming Languages and Operating Systems},
pages = {710–725},
numpages = {16},
location = {Lausanne, Switzerland},
series = {ASPLOS '22}
}

@article{Pavlogiannis-2020,
author = {Pavlogiannis, Andreas},
title = {Fast, Sound, and Effectively Complete Dynamic Race Prediction},
year = {2020},
issue_date = {January 2020},
publisher = {Association for Computing Machinery},
address = {New York, NY, USA},
volume = {4},
number = {POPL},
url = {https://doi.org/10.1145/3371085},
doi = {10.1145/3371085},
journal = {Proc. ACM Program. Lang.},
articleno = {17},
numpages = {29}
}

@article{Flanagan2008,
author = {Flanagan, Cormac and Freund, Stephen N. and Lifshin, Marina and Qadeer, Shaz},
title = {Types for atomicity: Static checking and inference for Java},
year = {2008},
issue_date = {July 2008},
publisher = {Association for Computing Machinery},
address = {New York, NY, USA},
volume = {30},
number = {4},
issn = {0164-0925},
url = {https://doi.org/10.1145/1377492.1377495},
doi = {10.1145/1377492.1377495},
journal = {ACM Trans. Program. Lang. Syst.},
month = aug,
articleno = {20},
numpages = {53}
}

@inproceedings{Sasturkar2005,
author = {Sasturkar, Amit and Agarwal, Rahul and Wang, Liqiang and Stoller, Scott D.},
title = {Automated type-based analysis of data races and atomicity},
year = {2005},
isbn = {1595930809},
publisher = {Association for Computing Machinery},
address = {New York, NY, USA},
url = {https://doi.org/10.1145/1065944.1065956},
doi = {10.1145/1065944.1065956},
booktitle = {Proceedings of the Tenth ACM SIGPLAN Symposium on Principles and Practice of Parallel Programming},
pages = {83–94},
numpages = {12},
location = {Chicago, IL, USA},
series = {PPoPP '05}
}

@article{Flanagan2003,
author = {Flanagan, Cormac and Qadeer, Shaz},
title = {A type and effect system for atomicity},
year = {2003},
issue_date = {May 2003},
publisher = {Association for Computing Machinery},
address = {New York, NY, USA},
volume = {38},
number = {5},
issn = {0362-1340},
url = {https://doi.org/10.1145/780822.781169},
doi = {10.1145/780822.781169},
journal = {SIGPLAN Not.},
month = may,
pages = {338–349},
numpages = {12}
}

@inproceedings{Musuvathi-2008,
	author = {Musuvathi, Madanlal and Qadeer, Shaz and Ball, Thomas and Basler, Gerard and Nainar, Piramanayagam Arumuga and Neamtiu, Iulian},
	title = {Finding and Reproducing Heisenbugs in Concurrent Programs},
	year = {2008},
	publisher = {USENIX Association},
	address = {USA},
	booktitle = {Proceedings of the 8th USENIX Conference on Operating Systems Design and Implementation},
	pages = {267–280},
	numpages = {14},
  url = {https://doi.org/10.5555/1855741.1855760},
  doi = {10.5555/1855741.1855760},
	location = {San Diego, California},
	series = {OSDI'08}
}

@article{Papadimitriou1979a,
  title = {The Serializability of Concurrent Database Updates},
  author = {Papadimitriou, Christos H.},
  year = {1979},
  month = oct,
  journal = {Journal of the ACM},
  volume = {26},
  number = {4},
  pages = {631--653},
  issn = {0004-5411, 1557-735X},
  doi = {10.1145/322154.322158},
  url = {https://dl.acm.org/doi/10.1145/322154.322158}
}

@inproceedings{Flanagan-2009,
	author = {Flanagan, Cormac and Freund, Stephen N.},
	title = {FastTrack: Efficient and Precise Dynamic Race Detection},
	year = {2009},
	isbn = {9781605583921},
	publisher = {Association for Computing Machinery},
	address = {New York, NY, USA},
	url = {https://doi.org/10.1145/1542476.1542490},
	doi = {10.1145/1542476.1542490},
	booktitle = {Proceedings of the 30th ACM SIGPLAN Conference on Programming Language Design and Implementation},
	pages = {121–133},
	numpages = {13},
	location = {Dublin, Ireland},
	series = {PLDI '09}
}

@article{Flanagan-2004,
author = {Flanagan, Cormac and Freund, Stephen N},
title = {Atomizer: a dynamic atomicity checker for multithreaded programs},
year = {2004},
issue_date = {January 2004},
publisher = {Association for Computing Machinery},
address = {New York, NY, USA},
volume = {39},
number = {1},
issn = {0362-1340},
url = {https://doi.org/10.1145/982962.964023},
doi = {10.1145/982962.964023},
journal = {SIGPLAN Not.},
month = jan,
pages = {256–267},
numpages = {12}
}

@inproceedings{Serebryany-2009,
	author = {Serebryany, Konstantin and Iskhodzhanov, Timur},
	title = {ThreadSanitizer: Data Race Detection in Practice},
	year = {2009},
	isbn = {9781605587936},
	publisher = {Association for Computing Machinery},
	address = {New York, NY, USA},
	url = {https://doi.org/10.1145/1791194.1791203},
	doi = {10.1145/1791194.1791203},
	booktitle = {Proceedings of the Workshop on Binary Instrumentation and Applications},
	pages = {62–71},
	numpages = {10},
	location = {New York, New York, USA},
	series = {WBIA '09}
}

@inproceedings{Roemer-2020,
  title = {{{SmartTrack}}: Efficient Predictive Race Detection},
  shorttitle = {{{SmartTrack}}},
  booktitle = {Proceedings of the 41st {{ACM SIGPLAN Conference}} on {{Programming Language Design}} and {{Implementation}}},
  author = {Roemer, Jake and Gen{\c c}, Kaan and Bond, Michael D.},
  year = {2020},
  month = jun,
  series = {{{PLDI}} 2020},
  pages = {747--762},
  publisher = {{Association for Computing Machinery}},
  address = {{New York, NY, USA}},
  doi = {10.1145/3385412.3385993},
  url = {https://dl.acm.org/doi/10.1145/3385412.3385993},
  isbn = {978-1-4503-7613-6}
}

@inproceedings{Flanagan-2008,
	author = {Flanagan, Cormac and Freund, Stephen N. and Yi, Jaeheon},
	title = {Velodrome: A Sound and Complete Dynamic Atomicity Checker for Multithreaded Programs},
	year = {2008},
	isbn = {9781595938602},
	publisher = {Association for Computing Machinery},
	address = {New York, NY, USA},
	url = {https://doi.org/10.1145/1375581.1375618},
	doi = {10.1145/1375581.1375618},
	booktitle = {Proceedings of the 29th ACM SIGPLAN Conference on Programming Language Design and Implementation},
	pages = {293–303},
	numpages = {11},
	location = {Tucson, AZ, USA},
	series = {PLDI '08}
}

@inproceedings{Mathur-2020c,
  title = {The {{Complexity}} of {{Dynamic Data Race Prediction}}},
  booktitle = {Proceedings of the 35th {{Annual ACM}}/{{IEEE Symposium}} on {{Logic}} in {{Computer Science}}},
  author = {Mathur, Umang and Pavlogiannis, Andreas and Viswanathan, Mahesh},
  year = {2020},
  month = jul,
  pages = {713--727},
  publisher = {ACM},
  address = {Saarbr{\"u}cken Germany},
  doi = {10.1145/3373718.3394783},
  url = {https://dl.acm.org/doi/10.1145/3373718.3394783},
  isbn = {978-1-4503-7104-9}
}

@article{Ma-2021,
  title = {{{RegionTrack}}: {{A Trace-Based Sound}} and {{Complete Checker}} to {{Debug Transactional Atomicity Violations}} and {{Non-Serializable Traces}}},
  shorttitle = {{{RegionTrack}}},
  author = {Ma, Xiaoxue and Wu, Shangru and Pobee, Ernest and Mei, Xiupei and Zhang, Hao and Jiang, Bo and Chan, Wing-Kwong},
  year = {2021},
  month = dec,
  journal = {ACM Transactions on Software Engineering and Methodology},
  volume = {30},
  number = {1},
  pages = {7:1--7:49},
  issn = {1049-331X},
  doi = {10.1145/3412377},
  url = {https://dl.acm.org/doi/10.1145/3412377}
}

@article{Tunc-2023,
author = {Tun\c{c}, H\"{u}nkar Can and Mathur, Umang and Pavlogiannis, Andreas and Viswanathan, Mahesh},
title = {Sound Dynamic Deadlock Prediction in Linear Time},
year = {2023},
issue_date = {June 2023},
publisher = {Association for Computing Machinery},
address = {New York, NY, USA},
volume = {7},
number = {PLDI},
url = {https://doi.org/10.1145/3591291},
doi = {10.1145/3591291},
journal = {Proc. ACM Program. Lang.},
month = {jun},
articleno = {177},
numpages = {26}
}

@inproceedings{Biswas-2014,
author = {Biswas, Swarnendu and Huang, Jipeng and Sengupta, Aritra and Bond, Michael D.},
title = {DoubleChecker: Efficient Sound and Precise Atomicity Checking},
year = {2014},
isbn = {9781450327848},
publisher = {Association for Computing Machinery},
address = {New York, NY, USA},
url = {https://doi.org/10.1145/2594291.2594323},
doi = {10.1145/2594291.2594323},
booktitle = {Proceedings of the 35th ACM SIGPLAN Conference on Programming Language Design and Implementation},
pages = {28–39},
numpages = {12},
location = {Edinburgh, United Kingdom},
series = {PLDI '14}
}

@inproceedings{Mathur-2020,
author = {Mathur, Umang and Viswanathan, Mahesh},
title = {Atomicity Checking in Linear Time Using Vector Clocks},
year = {2020},
isbn = {9781450371025},
publisher = {Association for Computing Machinery},
address = {New York, NY, USA},
url = {https://doi.org/10.1145/3373376.3378475},
doi = {10.1145/3373376.3378475},
booktitle = {Proceedings of the Twenty-Fifth International Conference on Architectural Support for Programming Languages and Operating Systems},
pages = {183–199},
numpages = {17},
location = {Lausanne, Switzerland},
series = {ASPLOS '20}
}

@inproceedings{Tunc-2024,
author = {Tun\c{c}, H\"{u}nkar Can and Deshmukh, Ameya Prashant and Cirisci, Berk and Enea, Constantin and Pavlogiannis, Andreas},
title = {CSSTs: A Dynamic Data Structure for Partial Orders in Concurrent Execution Analysis},
year = {2024},
isbn = {9798400703867},
publisher = {Association for Computing Machinery},
address = {New York, NY, USA},
doi = {10.1145/3620666.3651358},
booktitle = {Proceedings of the 29th ACM International Conference on Architectural Support for Programming Languages and Operating Systems, Volume 3},
pages = {223–238},
numpages = {16},
location = {La Jolla, CA, USA},
series = {ASPLOS '24}
}

@inproceedings{Lu-2008,
author = {Lu, Shan and Park, Soyeon and Seo, Eunsoo and Zhou, Yuanyuan},
title = {Learning from mistakes: a comprehensive study on real world concurrency bug characteristics},
year = {2008},
isbn = {9781595939586},
publisher = {Association for Computing Machinery},
address = {New York, NY, USA},
doi = {10.1145/1346281.1346323},
booktitle = {Proceedings of the 13th International Conference on Architectural Support for Programming Languages and Operating Systems},
pages = {329–339},
numpages = {11},
location = {Seattle, WA, USA},
series = {ASPLOS XIII}
}

@inproceedings{Lu-2006,
author = {Lu, Shan and Tucek, Joseph and Qin, Feng and Zhou, Yuanyuan},
title = {AVIO: detecting atomicity violations via access interleaving invariants},
year = {2006},
isbn = {1595934510},
publisher = {Association for Computing Machinery},
address = {New York, NY, USA},
doi = {10.1145/1168857.1168864},
booktitle = {Proceedings of the 12th International Conference on Architectural Support for Programming Languages and Operating Systems},
pages = {37–48},
numpages = {12},
location = {San Jose, California, USA},
series = {ASPLOS XII}
}

@inproceedings{Park-2009,
author = {Park, Soyeon and Lu, Shan and Zhou, Yuanyuan},
title = {CTrigger: exposing atomicity violation bugs from their hiding places},
year = {2009},
isbn = {9781605584065},
publisher = {Association for Computing Machinery},
address = {New York, NY, USA},
doi = {10.1145/1508244.1508249},
booktitle = {Proceedings of the 14th International Conference on Architectural Support for Programming Languages and Operating Systems},
pages = {25–36},
numpages = {12},
location = {Washington, DC, USA},
series = {ASPLOS XIV}
}

@inproceedings{Park-2008,
author = {Park, Chang-Seo and Sen, Koushik},
title = {Randomized active atomicity violation detection in concurrent programs},
year = {2008},
isbn = {9781595939951},
publisher = {Association for Computing Machinery},
address = {New York, NY, USA},
url = {https://doi.org/10.1145/1453101.1453121},
doi = {10.1145/1453101.1453121},
booktitle = {Proceedings of the 16th ACM SIGSOFT International Symposium on Foundations of Software Engineering},
pages = {135–145},
numpages = {11},
location = {Atlanta, Georgia},
series = {SIGSOFT '08/FSE-16}
}

@inproceedings{Sorrentino-2010,
author = {Sorrentino, Francesco and Farzan, Azadeh and Madhusudan, P.},
title = {PENELOPE: weaving threads to expose atomicity violations},
year = {2010},
isbn = {9781605587912},
publisher = {Association for Computing Machinery},
address = {New York, NY, USA},
url = {https://doi.org/10.1145/1882291.1882300},
doi = {10.1145/1882291.1882300},
booktitle = {Proceedings of the Eighteenth ACM SIGSOFT International Symposium on Foundations of Software Engineering},
pages = {37–46},
numpages = {10},
location = {Santa Fe, New Mexico, USA},
series = {FSE '10}
}

@inproceedings{Samak-2015,
author = {Samak, Malavika and Ramanathan, Murali Krishna},
title = {Synthesizing tests for detecting atomicity violations},
year = {2015},
isbn = {9781450336758},
publisher = {Association for Computing Machinery},
address = {New York, NY, USA},
url = {https://doi.org/10.1145/2786805.2786874},
doi = {10.1145/2786805.2786874},
booktitle = {Proceedings of the 2015 10th Joint Meeting on Foundations of Software Engineering},
pages = {131–142},
numpages = {12},
location = {Bergamo, Italy},
series = {ESEC/FSE 2015}
}

@misc{Umang-rapid,
  author       = {Umang Mathur},
  title        = {Rapid: Dynamic Analysis for Concurrent Programs},
  year         = {2024},
  howpublished = {\url{https://github.com/focs-lab/rapid}},
  note         = {Accessed: 2024-12-03}
}

@article{Blackburn-2006,
author = {Blackburn, Stephen M. and Garner, Robin and Hoffmann, Chris and Khang, Asjad M. and McKinley, Kathryn S. and Bentzur, Rotem and Diwan, Amer and Feinberg, Daniel and Frampton, Daniel and Guyer, Samuel Z. and Hirzel, Martin and Hosking, Antony and Jump, Maria and Lee, Han and Moss, J. Eliot B. and Phansalkar, Aashish and Stefanovi\'{c}, Darko and VanDrunen, Thomas and von Dincklage, Daniel and Wiedermann, Ben},
title = {The DaCapo benchmarks: java benchmarking development and analysis},
year = {2006},
issue_date = {October 2006},
publisher = {Association for Computing Machinery},
address = {New York, NY, USA},
volume = {41},
number = {10},
issn = {0362-1340},
url = {https://doi.org/10.1145/1167515.1167488},
doi = {10.1145/1167515.1167488},
journal = {SIGPLAN Not.},
month = oct,
pages = {169–190},
numpages = {22}
}

@inproceedings{Smith-2001,
author = {Smith, L. A. and Bull, J. M. and Obdrz\'{a}lek, J.},
title = {A parallel java grande benchmark suite},
year = {2001},
isbn = {158113293X},
publisher = {Association for Computing Machinery},
address = {New York, NY, USA},
url = {https://doi.org/10.1145/582034.582042},
doi = {10.1145/582034.582042},
booktitle = {Proceedings of the 2001 ACM/IEEE Conference on Supercomputing},
pages = {8},
numpages = {1},
location = {Denver, Colorado},
series = {SC '01}
}

@online{mysql-04,
  author       = {Jocelyn Fournier},
  title        = {Atomicity violation leading to crash in MySQL 4.1.2},
  year         = {2004},
  organization = {MySQL AB},
  url          = {https://bugs.mysql.com/bug.php?id=3596},
}

@online{mysql-08,
  author       = {Jeremy Cole},
  title        = {Atomicity violation leading to crash in MySQL 5.4.3},
  year         = {2008},
  organization = {MySQL AB},
  url          = {https://bugs.mysql.com/bug.php?id=38883},
}

@online{mysql-08-2,
  author       = {Shane Bester},
  title        = {Atomicity violation leading to crash in MySQL 6.0.7},
  year         = {2008},
  organization = {MYSQL AB},
  url          = {https://bugs.mysql.com/bug.php?id=38816},
}

@inproceedings{Mathur_2020, series={ASPLOS ’20},
   title={Atomicity Checking in Linear Time using Vector Clocks V3},
   url={https://arxiv.org/abs/2001.04961},
   booktitle={Proceedings of the Twenty-Fifth International Conference on Architectural Support for Programming Languages and Operating Systems},
   publisher={ACM},
   author={Mathur, Umang and Viswanathan, Mahesh},
   year={2020},
   month=mar, pages={183–199},
   collection={ASPLOS ’20}
}

@article{Lipton-75,
author = {Lipton, Richard J.},
title = {Reduction: a method of proving properties of parallel programs},
year = {1975},
issue_date = {Dec. 1975},
publisher = {Association for Computing Machinery},
address = {New York, NY, USA},
volume = {18},
number = {12},
issn = {0001-0782},
url = {https://doi.org/10.1145/361227.361234},
doi = {10.1145/361227.361234},
journal = {Commun. ACM},
month = dec,
pages = {717–721},
numpages = {5}
}

@article{Abdulla-25,
author = {Abdulla, Parosh Aziz and Grahn, Samuel and Jonsson, Bengt and Krishna, Shankaranarayanan and Mishra, Om Swostik},
title = {Efficient Linearizability Monitoring},
year = {2025},
issue_date = {June 2025},
publisher = {Association for Computing Machinery},
address = {New York, NY, USA},
volume = {9},
number = {PLDI},
url = {https://doi.org/10.1145/3729328},
doi = {10.1145/3729328},
journal = {Proc. ACM Program. Lang.},
month = jun,
articleno = {225},
numpages = {24}
}

@article{Lee-25,
author = {Lee, Zheng Han and Mathur, Umang},
title = {Efficient Decrease-and-Conquer Linearizability Monitoring},
year = {2025},
issue_date = {October 2025},
publisher = {Association for Computing Machinery},
address = {New York, NY, USA},
volume = {9},
number = {OOPSLA2},
url = {https://doi.org/10.1145/3763123},
doi = {10.1145/3763123},
journal = {Proc. ACM Program. Lang.},
month = oct,
articleno = {345},
numpages = {28}
}

@software{hunkar-2026,
  author       = {Tunç, Hünkar Can and
                  Dong, Yifan and
                  Pavlogiannis, Andreas},
  title        = {Fast Atomicity Monitoring},
  month        = apr,
  year         = 2026,
  publisher    = {Zenodo},
  doi          = {10.5281/zenodo.19615809},
  url          = {https://doi.org/10.5281/zenodo.19615809},
}

@misc{hunkar-2026a,
      title={Fast Atomicity Monitoring}, 
      author={Hünkar Can Tun and Yifan Dong and Andreas Pavlogiannis},
      year={2026},
      eprint={2604.11369},
      archivePrefix={arXiv},
      primaryClass={cs.PL},
      url={https://arxiv.org/abs/2604.11369}, 
}
